\documentclass[11pt]{article}

\usepackage[utf8]{inputenc}
\usepackage[T1]{fontenc}
\usepackage{microtype}

\usepackage{amsmath,amssymb,amsfonts,amsthm}
\newtheorem{theorem}{Theorem}
\newtheorem{proposition}{Proposition}

\newtheorem{assumption}{Assumption}

\usepackage[margin=1in]{geometry}
\usepackage[section]{placeins}
\let\origsubsection\subsection
\renewcommand{\subsection}{\FloatBarrier\origsubsection}
\let\origsubsubsection\subsubsection
\renewcommand{\subsubsection}{\FloatBarrier\origsubsubsection}

\usepackage{graphicx}
\usepackage{float}
\usepackage{booktabs}
\usepackage{multirow}
\usepackage{caption}
\usepackage{subcaption}

\usepackage{algorithm}
\usepackage{algorithmic}
\usepackage{enumitem}
\usepackage{xcolor}

\makeatletter
\newenvironment{breakablealgorithm}
  {\begin{center}
    \refstepcounter{algorithm}%
    \hrule height.8pt depth0pt \kern2pt%
    \renewcommand{\caption}[2][\relax]{%
      {\raggedright\textbf{\ALG@name~\thealgorithm} ##2\par}%
      \ifx\relax##1\relax%
        \addcontentsline{loa}{algorithm}{\protect\numberline{\thealgorithm}##2}%
      \else%
        \addcontentsline{loa}{algorithm}{\protect\numberline{\thealgorithm}##1}%
      \fi%
      \kern2pt\hrule\kern2pt%
    }%
  }{\kern2pt\hrule\relax%
   \end{center}%
  }
\makeatother

\usepackage{tikz}
\usetikzlibrary{shapes.geometric,arrows.meta,positioning,fit,calc,decorations.pathreplacing}

\usepackage[numbers,sort&compress]{natbib}

\usepackage{authblk}

\usepackage{hyperref}
\hypersetup{
    colorlinks=true,
    linkcolor=blue,
    citecolor=blue,
    urlcolor=blue,
    pdftitle={Continuous Hidden Markov Models for Equity Returns: Heavy-Tail Emission Families and Regime-Conditional Value-at-Risk},
    pdfauthor={Abdulrahman Alswaidan, Cade Jin, Jeffrey D. Varner},
    pdfsubject={CHMM synthetic-data generator for daily US-equity returns},
    pdfkeywords={CHMM, Baum-Welch, Student-t emissions, semi-Markov, stylized facts, copula, Value-at-Risk, Kupiec, synthetic data}
}
\makeatletter
\g@addto@macro\UrlBreaks{\do\-\do\_}
\makeatother

\newcommand{\R}{\mathbb{R}}
\newcommand{\E}{\mathbb{E}}
\newcommand{\Prob}{\mathbb{P}}

\newcommand{\NPaths}[1]{\ensuremath{N_{\text{paths}} = #1}}

\title{Continuous Hidden Markov Models for Equity Returns:\\ Heavy-Tail Emission Families\\ and Regime-Conditional Value-at-Risk}

\author[1]{Abdulrahman Alswaidan}
\author[2]{Cade Jin}
\author[1]{Jeffrey D.~Varner\thanks{\raggedright Corresponding author: \texttt{jdv27@cornell.edu}. Coauthor contacts: \texttt{aa2725@cornell.edu} (A.A.), \texttt{cj383@cornell.edu} (C.J.).\par}}
\affil[1]{Robert Frederick Smith School of Chemical and Biomolecular Engineering, Cornell University, Ithaca, NY 14850, USA}
\affil[2]{Cornell Ann S.~Bowers College of Computing and Information Science, Cornell University, Ithaca, NY 14850, USA}

\date{}

\begin{document}

\maketitle

% ========================================================================================= %
\begin{abstract}
Synthetic generators of daily US-equity returns support stress testing, backtesting, and scenario design, none of which a single realized market history can supply on its own. Low-state-count hidden Markov models with Gaussian emissions, however, were long held to fail at reproducing the slow autocorrelation decay of absolute returns, a defining stylized fact of daily returns. We revisited that result with a continuous hidden Markov model whose latent regime chain governs the autocorrelation while per-regime densities govern the marginal, which separates the temporal and distributional sides of the original failure. We developed a unified expectation-maximisation framework placing Gaussian, Student-$t$, Laplace, and generalised-error emissions under shared forward-backward recursions and quantile-based initialisation, differing only in each family's update. A spectral identity expressed that slow decay as a sum of finitely many geometric modes, their number bounded by the rank of the centred transition matrix, which diagnosed when adding states could help. Across SPY, a sector-balanced 30-ticker panel, a CRSP cross-decade transfer, and a six-asset copula extension, the rank bound was not active at the typical ticker once a few states were used, so marginal flexibility explained more of the remaining fit gap than additional decay modes. With a heavy-tailed emission the model reproduced the three symmetric Cont stylized facts and narrowed the kurtosis gap left by the Gaussian baseline without a tuning hyperparameter. A regime-conditional Value-at-Risk passed a joint conditional-coverage test on held-out data, and a Student-$t$ copula reproduced cross-asset correlations more faithfully than a single-index alternative. An i.i.d.\ bootstrap and a maximum-likelihood hidden semi-Markov benchmark matched or exceeded the model on raw single-window fit but supported neither the regime-conditional risk forecast nor the multi-asset composition. The scope is daily US equities in stable out-of-sample periods, with periodic refitting recommended when regimes drift.
\end{abstract}


% ========================================================================================= %
\section{Introduction}
\label{sec:introduction}
Synthetic generators of equity-return time series produce plausible sample paths that match the statistical structure of real markets without replaying the single realized history. This supports stress testing, regulatory backtesting, and scenario design against conditions absent from the observed record~\cite{jordon2022synthetic, assefa2020generating, alswaidan2026hybrid}. A useful generator should reproduce all of the standard stylized facts of daily returns at once: heavy-tailed marginals, negligible linear autocorrelation, and slow decay of the autocorrelation function (ACF) of absolute returns~\cite{mandelbrot1963variation, cont2001empirical}. Generalised autoregressive conditional heteroskedasticity (GARCH) models~\cite{engle1982autoregressive, bollerslev1986generalized} capture volatility clustering but not the multi-regime structure of the marginal; i.i.d.\ generators~\cite{efron1979bootstrap} match the marginal but lose temporal persistence; deep generators~\cite{yoon2019timegan, wiese2020quantgan, rasul2021autoregressive} reproduce some stylized facts at the cost of interpretability and reproducibility. One negative result has shaped two decades of work on regime-switching models. \citet{ryden1998stylized} fitted Gaussian-emission hidden Markov models (HMMs) at low state counts and concluded that the geometric sojourn-time distribution of a standard Markov chain makes the ACF decay too fast to match the observed slow decay; \citet{bulla2006stylized} recovered the ACF fit by replacing the Markov chain with a hidden semi-Markov chain, at the cost of much greater complexity. The Ryd\'{e}n failure combines two constraints: a temporal one (the absolute growth-rate ACF must allow slow-decay modes) and a distributional one (the marginal mixture must carry enough components to match observed kurtosis). A standard closed-form identity for the absolute growth-rate ACF, written as a sum over the non-unit eigenvalues of the transition matrix~\cite{hamilton1994time, krolzig1997markov, timmermann2000moments}, bounds the number of decay modes by the rank of the centred transition matrix. We ask whether that bound actually matters on equity-return data once enough latent states are allowed.

We answer this with an established model class, the continuous-emission hidden Markov model (CHMM): a hidden chain of market regimes in which each regime emits the day's growth rate from its own probability distribution. The regime chain controls the geometric modes available to the absolute growth-rate ACF, while the per-regime distributions control the heavy-tailed marginal, allowing the temporal and distributional constraints of the low-state Gaussian fit to be examined separately. Our contribution is not the CHMM itself. It is a unified comparison of Gaussian, Student-$t$, Laplace, and generalized-error emissions under the same forward-backward framework; quantile-based initialization; an empirical spectral diagnosis of when the finite-mode bound is active; and extensions to regime-conditional Value-at-Risk (VaR) and multi-asset copula composition. We start each fit from quantile splits of the observed growth rates, seeding each regime with its own slice of the distribution, which keeps the starting regimes from overlapping and helps the fit converge with every regime still in use. On the fitted model we then build a regime-conditional VaR that forecasts tail risk from the model's running estimate of which regime the day belongs to, and a multi-asset version that couples several single-asset CHMMs into one joint generator while leaving each asset's fitted distribution untouched.

In this study we put the model through four complementary evaluation settings: a six-fold rolling-origin walk-forward on a decade of SPY, a sector-balanced 30-ticker panel across ten Global Industry Classification Standard (GICS) sectors, a Center for Research in Security Prices (CRSP) cross-decade transfer from a ten-year window to a later two-year slice, and a six-asset US-equity basket, which together reduce reliance on a single window or asset universe. With enough latent states and a heavy-tailed emission, the model reproduced all three symmetric Cont stylized facts and narrowed the kurtosis gap left by the Gaussian baseline. Under the diagnostics used here, the rank bound from the ACF identity was not empirically active at the typical ticker, although a finite-state HMM still permits only finitely many geometric decay modes and the result does not establish a general power-law or multi-scale approximation. The regime-conditional VaR passed the Christoffersen joint conditional-coverage test~\cite{christoffersen1998evaluating} on the main window and on most walk-forward folds, and the multi-asset version reproduced cross-asset correlations better than the single-index alternative in the tested basket. We therefore interpret the Ryd\'{e}n low-state-count failure more narrowly: on these daily US-equity data and metrics, increasing marginal flexibility resolved more of the observed fit gap than adding ACF modes beyond the selected state count. An i.i.d.\ bootstrap and a maximum-likelihood hidden semi-Markov model (HSMM) beat the CHMM on raw single-window fit, but the fitted CHMM additionally supports the regime-conditional risk forecast and multi-asset coupling exercised here. We make no formal privacy guarantee for its synthetic series. The scope is daily US equities: a non-equity stress test on the GLD exchange-traded fund (ETF) broke down under static fitting, and we recommend periodic refit within that scope.

\section{Related Work}
\label{sec:related}
The discrete-state regime-switching tradition opens with \citet{hamilton1989new}, who introduced Markov-switching to economics, and \citet{schaller1997regime}, who estimated regime-switching specifications on US monthly stock returns and documented the predictive content of the latent regime for return forecasting. Continuous-emission HMMs and non-Gaussian state densities are established parts of this tradition; the present paper contributes a common estimation and evaluation framework rather than a new model class. The main evaluation against the Cont stylized facts comes from \citet{ryden1998stylized}, who fitted Gaussian-emission hidden Markov models at $K = 2$ or $3$ states and showed that the geometric sojourn-time distribution of a standard Markov chain produces exponential ACF decay too fast to match the observed slow decay of absolute returns. \citet{bulla2006stylized} recovered the fit with hidden semi-Markov models at much higher complexity, replacing the geometric sojourn with an explicit-duration family at the cost of a duration-augmented forward-backward, and \citet{nystrup2017long} extended the same line to long-memory volatility. Each of these works makes the same kind of structural choice we make here: the latent regime is a finite-state object with per-state emission densities, and what varies across the lineage is the sojourn distribution and the per-state family. The emission families that supply the per-state densities used in this work draw on the maximum-likelihood Student-$t$ formulations of \citet{peel2000robust} and \citet{liu1995ml}, with the Generalised Error Distribution (GED)~\cite{subbotin1923law, box1973bayesian} providing a continuous shape parameter that interpolates between Gaussian ($p = 2$) and Laplace ($p = 1$) and underlies our CHMM-GED variant in the same role that the per-state degree-of-freedom parameter $\nu_k$ plays in CHMM-t. The discrete-state predecessor of this paper, \citet{alswaidan2026hybrid}, reproduced the same three stylized facts via Laplace quantile binning plus a Poisson-jump duration mechanism. The present work replaces that discretisation step with four continuous per-state emission families under a single estimation framework and tests empirically whether the finite collection of eigenvalue-driven ACF modes is adequate on the studied data, rather than claiming that it reproduces general long-memory decay.

Conditional-variance and stochastic-volatility models form a parallel line in which the latent state is the volatility process rather than a discrete regime. The basic Gaussian GARCH framework~\cite{engle1982autoregressive, bollerslev1986generalized} encodes volatility persistence through a parsimonious variance recursion; in our panel the tested GARCH(1,1) specification had a lower Kolmogorov-Smirnov pass rate than the CHMM variants (Table~\ref{tab:model_comparison}). Heavy-tailed innovations~\cite{bollerslev1987conditionally}, asymmetric leverage variants (exponential GARCH, EGARCH, \citealp{nelson1991conditional}; Glosten-Jagannathan-Runkle GARCH, GJR-GARCH, \citealp{glosten1993relation}; threshold GARCH, \citealp{zakoian1994threshold}), regime-switching extensions (Markov-switching GARCH, MS-GARCH, \citealp{haas2004new}, with the \texttt{MSGARCH} R package of \citealp{ardia2019msgarch} the reference implementation), and realised-variance frameworks~\cite{andersen1997heterogeneous, corsi2009simple} add tail, leverage, regime, or multi-scale structure at the variance-process layer. The continuous-latent counterpart is the stochastic-volatility family of \citet{taylor1982financial}, which evolves log-volatility as a latent AR(1) and proceeds by simulation-based estimators~\cite{jacquier1994bayesian, kim1998stochastic}, hidden-Markov approximations on a discretised volatility grid~\cite{rossi2006volatility, abantovalle2017svm}, or Markov-switching stochastic volatility~\cite{so1998stochastic} that interpolates between the continuous-volatility and finite-state formulations. These model classes can produce heavy-tailed predictive distributions and conditional VaR forecasts; unimodality alone does not impose a kurtosis ceiling. The distinction relevant to this study is narrower: the tested CHMM assigns a separate location, scale, and potentially shape parameter to each finite state, whereas the tested conditional-volatility baselines place most regime dependence in the variance recursion. The resulting empirical comparison separated those specifications on marginal fit and kurtosis (Table~\ref{tab:variant_choice}; extended baseline panel in the appendix), without establishing a structural impossibility for the broader GARCH or stochastic-volatility families.

Deep generative market models offer high-capacity distributional and temporal representations, often with less direct state-level interpretation. GAN-based generators require temporal structure in the architecture or training objective to reproduce volatility clustering~\cite{takahashi2019modeling, kwon2024can}; prominent examples include the TimeGAN family of \citet{yoon2019timegan} and the convolutional Wasserstein GAN QuantGAN of \citet{wiese2020quantgan}. Path-signature generators~\cite{chevyrev2016primer, ni2021sig, buehler2020data} and score-based or diffusion approaches~\cite{ho2020denoising, rasul2021autoregressive, tashiro2021csdi} provide alternative mechanisms for representing time-series dependence. Some models in these broad families can be conditional, can support likelihood or score-based prediction, or can be combined with multivariate dependence models; these capabilities are architecture-specific rather than absent from deep generators in general. We included one in-house QuantGAN implementation as the deep-generative row in the main generator panel, read as a negative control on whether that architecture recovered the same stylized facts at our sample size rather than as a representative forecasting baseline for the entire family. In that comparison, the CHMM's explicit finite-state semantics and closed-form state-conditional mixture made the regime-conditional VaR and copula composition used in this paper directly available, while the evaluated QuantGAN was not designed to provide those outputs. Reports of tail mode collapse in GAN-based financial generators~\cite{takahashi2019modeling, kwon2024can} motivate the heavy-tail diagnostic, but do not establish that failure for every deep generative formulation.

Cross-asset synthesis, synthetic-data evaluation, and non-stationarity handling round out the literature this paper draws on. The Single Index Model (SIM)~\cite{sharpe1963simplified} propagates a market factor through a rank-one linear decomposition but distorts each non-market marginal, and Sklar's theorem~\cite{sklar1959fonctions, nelsen2006introduction} supports a cleaner decomposition in which each asset's fitted marginal is preserved and only the copula is estimated, with Gaussian and Student-$t$ copulas~\cite{demarta2005tcopula, mcneil2015quantitative} the usual family choices and rank-reordering~\cite{iman1982distribution} the simulation step that preserves marginals exactly. Dynamic-conditional-correlation models~\cite{engle2002dynamic} and composite-likelihood estimators for vast covariance matrices~\cite{pakel2021fitting} are the standard alternatives at higher dimensions, and our six-asset construction sits well below the dimensions where those estimators become essential while the rank-based copula architecture composes with either at scale. Synthetic-data evaluation in this work is anchored on the framework of \citet{stenger2024thinking}, the proper-scoring-rule literature~\cite{gneiting2007strictly, gneiting2014probabilistic, diebold1995comparing}, kernel two-sample tests~\cite{gretton2012kernel}, and signature-based metrics~\cite{chevyrev2016primer, ni2020conditional}, which together fix the comparison metrics that the main panel uses: marginal goodness-of-fit, autocorrelation fidelity, kurtosis, proper scores, and signature distance. The non-stationarity literature documents structural breaks in single-name equity returns~\cite{pastorstambaugh2001equity, andreoughysels2002breaks, angtimmermann2012regime} and responds with window selection~\cite{pesarantimmermann2007window} and online or change-point updating~\cite{cappe2011online, adams2007bayesian, fearnhead2007online}; the walk-forward evaluation and the periodic-refit recommendation in this work draw directly on that line.

\section{Methods}
\label{sec:methods}
\subsection{Data and model}
\label{sec:data_and_model}

The single-asset analysis uses daily SPDR S\&P~500 ETF (SPY) prices from January~3, 2014 to January~3, 2024, a ten-year in-sample (IS hereafter) window of length $T_{\text{IS}} = 2{,}516$ trading days. The period from January~4, 2024 through April~20, 2026 is held out for out-of-sample (OoS hereafter) evaluation, $T_{\text{OoS}} = 572$ trading days. Prices are split-adjusted but not dividend-adjusted; the data sources are described in the Data Availability Statement. The 323-day Polygon versus Alpaca/IEX overlap check found no visible break at the vendor join. For ticker $i$ and trading day $j$, with $P_{i,j}$ the session volume-weighted average price (VWAP), $\Delta t = 1/252$ the annualised trading-day step, and $r_f$ the risk-free rate (set to zero here), the annualised excess growth rate is
\begin{equation}
    G_{i,j} \equiv \left(\frac{1}{\Delta t}\right)\; \ln\!\left(\frac{P_{i,j}}{P_{i,j-1}}\right) - r_f.
    \label{eq:growth_rate}
\end{equation}
We use daily data and the matching annualised growth rates. Since $\Delta t$ is fixed, the growth rate is a constant positive multiple of the day's log return. The excess kurtosis and the autocorrelations that define the stylized facts are therefore invariant to the rescaling and take the same values on returns and on growth rates; the absolute growth-rate ACF measures the same slow-decay stylized fact the literature reports for absolute returns. The probabilistic framework does not depend on the time step, although application at other frequencies requires frequency-specific preprocessing, state calibration, and evaluation horizons.

The model is fit one ticker at a time. We write $G_t$ for the daily excess growth rate on trading day $t$ for the ticker being fit (equation~\eqref{eq:growth_rate} with the ticker index dropped and the day index renamed from $j$ to $t$). The sequence $\{G_t\}$ is the observed output of a continuous hidden Markov model with $K$ latent states. We write the model as $\mathcal{M} = (\mathcal{S}, \mathbf{T}, \mathbf{F}, \boldsymbol{\pi})$~\cite{rabiner1986introduction}, where $\mathcal{S} = \{1, \ldots, K\}$ is the finite state space and $s_t \in \mathcal{S}$ is the (unobserved) latent state at time $t$. The transition matrix $\mathbf{T} \in \mathbb{R}^{K \times K}$ has its $(i,j)$ entry equal to the probability of moving from state $i$ at time $t$ to state $j$ at time $t+1$, $T_{ij} = \Prob(s_{t+1} = j \mid s_t = i)$; the rows therefore sum to one (each row is a valid probability distribution over next states). The initial-state distribution $\boldsymbol{\pi} \in \mathbb{R}^K$ has $k$th entry $\pi_k = \Prob(s_1 = k)$, the probability that the chain starts in state $k$. The per-state emission densities are collected in $\mathbf{F} = \{f_k(\,\cdot\,;\boldsymbol\theta_k)\}_{k=1}^K$, where $f_k(x; \boldsymbol\theta_k)$ is the density of the day's growth rate $G_t$ when the chain is in state $k$, evaluated at $G_t = x$; the parameter vector $\boldsymbol\theta_k$ holds that state's location, scale, and any shape parameters. The per-state cumulative distribution function (CDF) is $F_k(x; \boldsymbol\theta_k) = \int_{-\infty}^x f_k(u; \boldsymbol\theta_k)\,du$, used in the regime-conditional VaR construction.

Over many days the chain settles into a stable long-run distribution over its states. Write $\bar{\boldsymbol\pi} = (\bar\pi_1, \ldots, \bar\pi_K) \in \mathbb{R}^K$ for it, the one state distribution that a transition step leaves unchanged ($\bar{\boldsymbol{\pi}} = \bar{\boldsymbol{\pi}}\mathbf{T}$), so $\bar\pi_k$ is the long-run fraction of days the chain spends in state $k$. Collect every growth rate the model produces and ignore which state generated each one; the density of that pooled set of growth rates is a weighted average of the $K$ per-state densities, each weighted by how often its state is visited,
\begin{equation}
    f(x) = \sum_{k=1}^{K} \bar{\pi}_k\, f_k(x; \boldsymbol\theta_k), \qquad \bar{\boldsymbol{\pi}} = \bar{\boldsymbol{\pi}} \mathbf{T}.
    \label{eq:mixture}
\end{equation}
We call $f(x)$ the marginal mixture: \emph{marginal} because it averages over the hidden state to describe one day's growth rate on its own, a \emph{mixture} because it is a weighted sum of the per-state densities. A unique $\bar{\boldsymbol\pi}$ exists when $\mathbf{T}$ is irreducible (every state can eventually be reached from every other) and aperiodic (returns to a state are not locked to a fixed period); the formal conditions are given in Assumption~\ref{ass:irred}. We learn $\mathbf{T}$, the per-state parameter vectors $\boldsymbol\theta_{1:K} = (\boldsymbol\theta_1, \ldots, \boldsymbol\theta_K)$, and $\boldsymbol{\pi}$ jointly by expectation-maximisation (EM; Figure~\ref{fig:chmm_architecture}). Unconditional simulation initialises the chain from the stationary distribution $\bar{\boldsymbol\pi}$ rather than from the fitted $\boldsymbol\pi$, which removes transient dependence on the first observed day. The CHMM-N variant uses Gaussian emissions, $f_k = \mathcal{N}(\mu_k, \sigma_k^2)$ with $\sigma_k > 0$. The CHMM-t variant uses Student-$t$, $f_k = t_{\nu_k}(\mu_k, \sigma_k)$ with $\sigma_k > 0$ and $\nu_k \in [\nu_{\min}, \nu_{\max}]$; an optional exponential penalty on $1/\nu_k$ is used only in the $\lambda=20$ sensitivity specification. The CHMM-L variant uses Laplace, $f_k = \mathrm{Laplace}(\mu_k, b_k)$ with $b_k > 0$. The CHMM-GED variant uses the generalised error distribution, $f_k = \mathrm{GED}(\mu_k, \alpha_k, p_k)$ with $\alpha_k > 0$ and $p_k \in [p_{\min}, p_{\max}]$; the Gaussian and Laplace variants sit at the boundary of CHMM-GED, with $p_k = 2$ recovering Gaussian and $p_k = 1$ recovering Laplace.

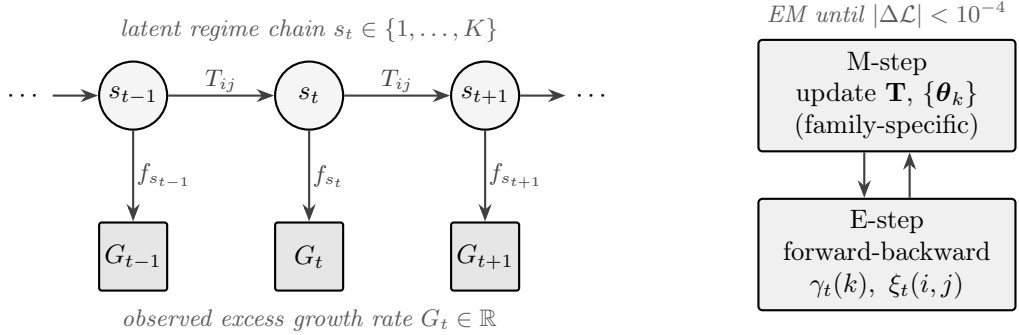
\begin{figure}[!ht]
\centering
\begin{tikzpicture}[
    font=\small,
    >={Stealth[length=2mm, width=1.6mm]},
    state/.style={draw, circle, minimum size=9mm, inner sep=0pt,
                  fill=black!4, thick},
    obs/.style={draw, rectangle, rounded corners=1pt, minimum size=9mm,
                inner sep=1pt, fill=black!10, thick},
    block/.style={draw, rectangle, rounded corners=1.5pt, align=center,
                  minimum height=11mm, minimum width=34mm, inner sep=1.5mm,
                  fill=black!6, thick},
    flow/.style={-Stealth, thick, black!75},
    lbl/.style={font=\footnotesize\itshape, black!65, inner sep=1pt},
    eq/.style={font=\footnotesize, inner sep=1pt},
]
\node[state] (s1) at (0, 0) {$s_{t-1}$};
\node[state, right=14mm of s1] (s2) {$s_t$};
\node[state, right=14mm of s2] (s3) {$s_{t+1}$};
\node[left=6mm of s1] (sdots) {$\cdots$};
\node[right=6mm of s3] (edots) {$\cdots$};
\draw[flow] (sdots) -- (s1);
\draw[flow] (s1) -- (s2) node[midway, above, eq] {$T_{ij}$};
\draw[flow] (s2) -- (s3) node[midway, above, eq] {$T_{ij}$};
\draw[flow] (s3) -- (edots);
\node[lbl, above=2mm of s2] {latent regime chain $s_t \in \{1,\ldots,K\}$};
\node[obs, below=12mm of s1] (o1) {$G_{t-1}$};
\node[obs, below=12mm of s2] (o2) {$G_t$};
\node[obs, below=12mm of s3] (o3) {$G_{t+1}$};
\draw[flow] (s1) -- (o1) node[midway, right, eq] {$f_{s_{t-1}}$};
\draw[flow] (s2) -- (o2) node[midway, right, eq] {$f_{s_t}$};
\draw[flow] (s3) -- (o3) node[midway, right, eq] {$f_{s_{t+1}}$};
\node[lbl, below=2mm of o2] {observed excess growth rate $G_t \in \mathbb{R}$};
\node[block, right=18mm of edots] (mstep)
    {M-step\\ update $\mathbf{T}$, $\{\boldsymbol\theta_k\}$\\ (family-specific)};
\node[block, below=6mm of mstep] (estep)
    {E-step\\ forward-backward\\ $\gamma_t(k),\ \xi_t(i,j)$};
\draw[flow, black!75] ([xshift=-3mm]mstep.south) -- ([xshift=-3mm]estep.north);
\draw[flow, black!75] ([xshift=3mm]estep.north) -- ([xshift=3mm]mstep.south);
\node[lbl, above=1mm of mstep] {EM until $|\Delta \mathcal{L}| < 10^{-4}$};
\end{tikzpicture}
\caption{\textbf{CHMM architecture and training loop.} A $K$-state latent Markov chain with transition matrix $\mathbf{T}$ generates the regime sequence $s_t$; each $s_t$ emits the observed excess growth rate $G_t$ through a state-specific density $f_{s_t}$. The four CHMM variants share everything except the per-state emission family (Gaussian, Student-$t$, Laplace, GED). The EM procedure alternates a log-space forward-backward E-step and a family-specific M-step until $|\Delta \mathcal L| < 10^{-4}$.}
\label{fig:chmm_architecture}
\end{figure}

\subsection{Estimation by EM}
\label{sec:estimation}

With the model specified, we estimate the parameters $(\mathbf{T}, \boldsymbol\theta_{1:K}, \boldsymbol\pi)$ by EM, with all probability computations carried out in log-space for numerical stability. From here on we write $O_t \equiv G_t$ for the observation at time $t$ ($O$ is the standard symbol for an HMM observation; here it just renames $G_t$), $O_{1:T} = (O_1, O_2, \ldots, O_T)$ for the full observed sequence over the fitting window of length $T$, and $\gamma_t(k) = \Prob(s_t = k \mid O_{1:T})$ for the smoothed posterior probability that the chain was in state $k$ at time $t$, given the entire observed sequence and the current parameter iterate (``smoothed'' means computed using both past and future observations, not just past). The pair posterior $\xi_t(i, j) = \Prob(s_t = i, s_{t+1} = j \mid O_{1:T})$ is the smoothed probability of an $i \to j$ transition at time $t$, used in the transition update. The collection $\{\gamma_t(k)\}_{t=1,\ldots,T;\, k=1,\ldots,K}$ is computed at every E-step by the standard forward-backward recursion (one forward pass and one backward pass through the data). We also use the superscript $(n)$ to mark the value of any quantity at EM iteration $n$, so $(0)$ is the initial value. All four emission families share the same log-space forward-backward recursions~\cite{bilmes1998gentle} and the same quantile-based initialisation: observations are sorted, split into $K$ equal-sized chunks, and state $k$'s initial location and scale are seeded from chunk $k$, with the transition matrix and initial distribution started uniform, $T_{ij}^{(0)} = \pi_k^{(0)} = 1/K$. The EM loop runs until the observed-data log-likelihood $\mathcal{L}^{(n)} = \log p(O_{1:T} \mid \mathbf{T}^{(n)}, \boldsymbol\theta_{1:K}^{(n)}, \boldsymbol\pi^{(n)})$, the likelihood of the growth rates alone with the latent states summed out, stops changing meaningfully between iterations, $|\mathcal{L}^{(n)} - \mathcal{L}^{(n-1)}| < \epsilon$ with $\epsilon = 10^{-4}$, or a maximum iteration cap is reached; the iteration cap, bracket widths, one-dimensional-search settings, and the full forward-backward recursion are reported in the algorithmic appendix. Quantile-based initialisation is chosen to reduce overlap among starting state components, with the aim of improving convergence to non-degenerate $K \ge 3$ fits.

The four emission families share the latent-state forward-backward recursion, the transition update, and the quantile-based initialisation; they differ in the emission density evaluated at each E-step and in the M-step update for the per-state parameter vector $\boldsymbol\theta_k$. Given the smoothed posteriors $\gamma_t(k)$, each family applies its own parameter update. For CHMM-N, the per-state parameters are the Gaussian mean and variance, $\boldsymbol\theta_k = (\mu_k, \sigma_k^2)$, and the classical Baum-Welch M-step~\cite{baum1970maximization} updates them in closed form using weighted sample-mean and sample-variance formulas, where each observation is weighted by the posterior probability $\gamma_t(k)$ that state $k$ generated it (all sums $\sum_t$ below run over $t = 1, \ldots, T$):
\[
\mu_k \leftarrow \frac{\sum_t \gamma_t(k)\, O_t}{\sum_t \gamma_t(k)},
\qquad
\sigma_k^2 \leftarrow \frac{\sum_t \gamma_t(k)\, (O_t - \mu_k)^2}{\sum_t \gamma_t(k)}.
\]

For CHMM-t, the per-state parameters are the Student-$t$ location, scale, and degrees of freedom, $\boldsymbol\theta_k = (\mu_k, \sigma_k, \nu_k)$. We use the expectation conditional maximisation (ECM) construction of \citet{peel2000robust} and \citet{liu1995ml}, which replaces the single M-step with a sequence of conditional maximisation (CM) sub-updates that each move one block of parameters at a time. The location and scale use the augmented-data closed-form CM updates, while the degrees of freedom use the expectation/conditional-maximisation-either (ECME) variant of \citet{liu1995ml}, which maximises the observed-data likelihood directly; this is the same heavy-tailed innovation choice used by \citet{abantovalle2017svm} for stochastic-volatility-in-mean models. The construction uses a standard representation of the Student-$t$ as a Gaussian with random variance: conditional on a latent positive scalar $u$, the observation is Gaussian, $G_t \mid s_t = k,\, u \sim \mathcal{N}(\mu_k, \sigma_k^2/u)$, and $u$ itself is drawn from a Gamma distribution, $u \sim \Gamma(\nu_k/2,\, \nu_k/2)$. Averaging over $u$ recovers the Student-$t$. In ECM, $u$ is treated as an additional latent variable, and the E-step adds its conditional expectation under the current parameter iterate,
\begin{equation}
    u_{t,k} \equiv \E[u \mid O_t, s_t = k] = \frac{\nu_k + 1}{\nu_k + ((O_t - \mu_k)/\sigma_k)^2},
    \label{eq:tprecision}
\end{equation}
and the M-step then updates
\begin{equation}
    \mu_k \leftarrow \frac{\sum_t \gamma_t(k) u_{t,k} O_t}{\sum_t \gamma_t(k) u_{t,k}},
    \qquad
    \sigma_k^2 \leftarrow \frac{\sum_t \gamma_t(k) u_{t,k} (O_t - \mu_k)^2}{\sum_t \gamma_t(k)}
    \label{eq:tupdate}
\end{equation}
for the location and scale, both in closed form given $\nu_k$. The degrees of freedom $\nu_k$ are then updated by a finite one-dimensional search on the bounded interval $[\nu_{\min}, \nu_{\max}]$ for a high-value point of the state-$k$ objective $Q_k(\nu) = \sum_t \gamma_t(k) \log t_\nu(O_t; \mu_k, \sigma_k)$, the observed-data state-weighted Student-$t$ log-likelihood rather than the augmented-data CM update in $\E[\log u]$. This direct maximisation is the ECME conditional-maximisation step; in exact form it preserves the monotone observed-data ascent~\cite{wu1983convergence, mengrubin1993ecm}, and the fixed-iteration numerical search used here approximates that maximiser, so we assess numerical convergence from the observed-data log-likelihood trace and stop at the stated tolerance or iteration cap rather than claim exact monotonicity for every finite-precision step.

For CHMM-L, the per-state parameters are the Laplace location and scale, $\boldsymbol\theta_k = (\mu_k, b_k)$, and both have closed-form updates. The location $\mu_k$ is a weighted median of the observations $\{O_t\}_{t=1}^T$, where each observation is weighted by the posterior probability $\gamma_t(k)$ that state $k$ generated it. The scale $b_k \leftarrow \sum_t \gamma_t(k)\, |O_t - \mu_k| / \sum_t \gamma_t(k)$ is the weighted mean absolute deviation of the observations from $\mu_k$, using the same weights.

For CHMM-GED, the per-state parameters are the GED location, scale, and shape, $\boldsymbol\theta_k = (\mu_k, \alpha_k, p_k)$, with conditional density $f_k(x; \mu_k, \alpha_k, p_k) = \mathrm{GED}(x; \mu_k, \alpha_k, p_k)$. The three parameters cannot all be updated together in closed form, so the M-step uses three sequential block updates. The location update searches a bounded interval for a low-value point of $\sum_t \gamma_t(k)\,|O_t - \mu|^{p_k}$. This objective is convex for $p_k \ge 1$ but need not be convex for $p_k < 1$, so the finite search is a numerical update rather than a certified global conditional minimisation. The scale then has a closed-form update at the just-updated location, $\alpha_k \leftarrow \big[(p_k / W_k) \sum_t \gamma_t(k)\,|O_t - \mu_k|^{p_k}\big]^{1/p_k}$, where $W_k = \sum_t \gamma_t(k)$ is the total posterior weight on state $k$. The shape is updated last by a bounded one-dimensional search over $[p_{\min}, p_{\max}]$ at the just-updated location and scale. We therefore describe CHMM-GED as an ECM-style numerical block-coordinate procedure and diagnose convergence from its observed-data log-likelihood trace; the exact ECM monotonicity theorem applies only when every block is conditionally maximised.

\subsection{Evaluation}
\label{sec:evaluation}

With the model fitted, we evaluate the CHMM at two levels, single-asset and multi-asset, across four datasets: the SPY in-sample and out-of-sample windows, a sector-balanced 30-ticker panel spanning ten GICS sectors with three large-cap representatives each, a CRSP cross-decade transfer between a 1994 to 2004 in-sample window and a 2004 to 2006 out-of-sample slice, and a six-asset US-equity basket (SPY, NVDA, JNJ, JPM, AAPL, QQQ). The \emph{single-asset} evaluation fits one CHMM per ticker and covers the stylized-fact diagnostics, $K$ selection, main generator comparison, cross-ticker generalisation panel, CRSP cross-decade transfer, walk-forward stress test, and Value-at-Risk panel. The \emph{multi-asset} evaluation keeps each asset's fitted CHMM marginal fixed and models only how the assets move together. A copula makes this split possible: Sklar's theorem~\cite{sklar1959fonctions} says any joint distribution factors into its one-asset marginals and a copula that carries the dependence between them, so the marginals and the dependence can be fitted separately. We use a Student-$t$ copula~\cite{demarta2005tcopula, mcneil2015quantitative} on the six-asset basket, with its correlation matrix fixed from Kendall's $\tau$ via $\rho=\sin(\pi\tau/2)$; its degrees of freedom $\nu^\star$ are selected by a grid search over the profile likelihood. To draw joint samples we use the rank-reordering scheme of \citet{iman1982distribution}, which reorders the per-asset CHMM draws into the copula's dependence pattern without altering any asset's own values, so the reordering preserves each asset's simulated marginal exactly. A single-index model (SIM), a Gaussian copula on the same CHMM marginals, and a truncated C-vine are reported in the appendix as multi-asset comparators.

The single-asset CHMM is benchmarked against one representative from each model class, all fit on the same IS window: a plain i.i.d.\ empirical bootstrap~\cite{efron1979bootstrap} as the strongest non-parametric baseline; Gaussian and Laplace i.i.d.\ generators as marginal baselines; GARCH(1,1) with Gaussian~\cite{bollerslev1986generalized} and Student-$t$~\cite{bollerslev1987conditionally} innovations as conditional-variance baselines; and Markov-switching GARCH (MS-GARCH)~\cite{haas2004new} as the closest regime-switching comparator. Extended GARCH-family, semi-Markov, stochastic-volatility, multifractal, and jump-diffusion baselines are reported in the appendix; a deep-generative QuantGAN row~\cite{wiese2020quantgan} appears in the main comparison table as a negative control, with its construction detailed in the appendix. For every model in the panel we simulate $P$ independent paths matching the window being scored: $T_{\text{IS}}$ in sample and $T_{\text{OoS}}=572$ out of sample~\cite{stenger2024thinking}; the path count and the deterministic seed used in each run are reported in the corresponding result table. We score the panel on three primary metrics. The two-sample Kolmogorov-Smirnov (KS) pass rate at significance level $\alpha = 0.05$~\cite{kolmogorov1933, smirnov1948} counts the fraction of simulated paths whose marginal distribution is not rejected against the observed window, measuring distributional fidelity at the marginal level. The asymptotic two-sample KS null assumes i.i.d.\ observations, whereas the observed and simulated paths carry serial dependence, so we read this pass rate as a descriptive marginal-fidelity score rather than a calibrated test and recalibrate it with the stationary block bootstrap in the appendix. Mean simulated excess kurtosis, read against the observed kurtosis on the same window, measures how well the heavy-tailed marginal is reproduced. The mean absolute error (MAE) between the observed and path-averaged simulated autocorrelation function (ACF) of absolute growth rates $|G_t|$ over $L = 252$ lags (one trading year) is
\begin{equation}
    \text{ACF-MAE} = \frac{1}{L}\sum_{\tau=1}^{L} \left| \hat\rho^{\text{obs}}_{|G|}(\tau) - \bar{\hat\rho}^{\,\text{sim}}_{|G|}(\tau) \right|,
    \label{eq:acf_mae}
\end{equation}
where $\tau$ is the lag index, $|G| = \{|G_t|\}_{t=1}^T$ is the absolute growth-rate sequence, $\hat\rho^{\text{obs}}_{|G|}(\tau)$ is the sample autocorrelation of $|G|$ at lag $\tau$ on the observed window, and $\bar{\hat\rho}^{\,\text{sim}}_{|G|}(\tau)$ is the same quantity averaged across the $P$ simulated paths; this captures volatility clustering, the third Cont stylized fact~\cite{cont2001empirical}. The Value-at-Risk panel additionally reports the unconditional Kupiec coverage test~\cite{kupiec1995techniques}, the Christoffersen joint conditional-coverage test~\cite{christoffersen1998evaluating} as our primary VaR backtest, and an expected-shortfall (ES) envelope (the $[5\%, 95\%]$ band of simulated ES across paths). Auxiliary distributional metrics (block-bootstrap KS recalibration, Anderson-Darling (AD), Wasserstein-1 ($W_1$), Hellinger distance, and continuous ranked probability score (CRPS)) are reported in the appendix to confirm robustness of the cross-generator ranking. The models and evaluation pipeline were implemented in Julia in the companion repository \texttt{CHMM-Model}.

\subsection{Spectral mechanism}

The route from the transition matrix to slow ACF decay runs through the spectrum of $\mathbf T$. Let $\mathbf v_k$ and $\mathbf w_k$ denote the right and left eigenvectors of $\mathbf T$ at non-unit eigenvalue $\lambda_k$, normalised biorthonormally so that $\mathbf w_k^\top \mathbf v_l = \delta_{kl}$, and let $\sigma_{|G|}^2$ denote the stationary variance of $|G_t|$ under the marginal mixture~\eqref{eq:mixture}. The closed-form bilinear identity for the absolute growth-rate ACF of a stationary CHMM, $\E[|G_t|\,|G_{t+\tau}|] = \mathbf m^\top \mathrm{diag}(\bar{\boldsymbol\pi})\,\mathbf T^\tau\,\mathbf m$ with $\mathbf m_k = \E[|G_t| \mid s_t = k]$, and its eigendecomposition over the non-unit eigenvalues of $\mathbf T$ are standard in the regime-switching literature~\cite{hamilton1994time,krolzig1997markov,timmermann2000moments}. Under irreducibility and aperiodicity of $\mathbf T$ (Assumption~\ref{ass:irred}, with $\bar{\boldsymbol\pi}$ the unique stationary distribution; \citealp{levinperes2017markov}), finite per-state second moments (Assumption~\ref{ass:moments}), and diagonalizability of $\mathbf T$, the spectral decomposition $\mathbf T^\tau = \mathbf 1\bar{\boldsymbol\pi}^\top + \sum_{k=2}^K \lambda_k^\tau\, \mathbf v_k \mathbf w_k^\top$ subtracts the marginal product to give the ACF
\begin{equation}
    \rho_{|G|}(\tau) \;=\; \sum_{k=2}^{K} a_k\, \lambda_k^{\tau},
    \qquad
    a_k = \big(\mathbf m^\top \mathrm{diag}(\bar{\boldsymbol\pi})\, \mathbf v_k\big)\big(\mathbf w_k^\top \mathbf m\big) / \sigma_{|G|}^2.
    \label{eq:acf_normalised}
\end{equation}
This is a real-valued sum of geometric modes at the non-unit eigenvalues of $\mathbf T$; complex-conjugate $\lambda_k$ pairs combine into damped oscillatory modes. For a non-diagonalisable transition matrix, the corresponding Jordan expansion contains polynomial-times-geometric terms instead of equation~\eqref{eq:acf_normalised}. A self-contained derivation, the lag-zero behaviour, and the additional fourth-moment condition needed for a squared growth-rate ACF are reported in the appendix. The appendix also states the scope of the identifiability and numerical-convergence claims; neither is required for the algebraic ACF identity. Our contribution is purely empirical. We apply equation~\eqref{eq:acf_normalised} to recast the \citet{ryden1998stylized} low-$K$ failure on the SPY 2014 to 2024 instance as a rank statement on $\mathbf T - \mathbf 1\bar{\boldsymbol\pi}^\top$. At $K = 2$, the ACF reduces to the mono-exponential form $\rho_{|G|}(\tau) = a_2\lambda_2^\tau$ with $\lambda_2 = T_{11} + T_{22} - 1$, so its persistence can be tuned through $T_{11} + T_{22}$ but its shape remains restricted to one geometric mode. At $K \ge 3$ the matrix admits multiple non-unit eigenvalues and the marginal mixture admits more components. The empirical results test whether those finite-mode and finite-mixture restrictions are active on the studied data; they do not imply that a finite-state HMM reproduces general power-law or multi-scale decay.

\section{Results}
\label{sec:results}
We computed sample stylized-fact diagnostics on the SPY in-sample and out-of-sample windows and observed all three Cont stylized facts on both windows (Figure~\ref{fig:stylized_facts}, Table~\ref{tab:model_comparison} Observed row). The marginal was heavy-tailed (in-sample excess kurtosis $7.68$, out-of-sample $5.29$), the linear autocorrelation in the raw growth rate $G_t$ was negligible, and the volatility clustering in $|G_t|$ was persistent. A stationary block bootstrap at mean block length $L = 20$ gave broadly overlapping IS and OoS $95\%$ CIs of $[2.17, 12.40]$ and $[0.90, 8.26]$ on the kurtosis point estimates (Table~\ref{tab:kurtosis_bootstrap}), so kurtosis targets should be read against the joint envelope rather than the point estimates. Having confirmed the three stylized facts on both windows, we next swept $K \in \{3, 6, 9, 12, 15, 18, 21\}$ for CHMM-N on the IS window and observed that the $K = 3$ versus $K = 6$ comparison sat inside sampling noise under four-fold and six-fold rolling-origin cross-validation (CV; full $K$-sweep CV panel in Table~\ref{tab:k_selection_cv}). The held-out log-likelihood differential satisfied $|z| \le 0.07$ with the sign flipping between fold designs, and the Bayesian information criterion (BIC) also picked $K = 3$. We chose $K^\star = 3$ as the default state count: held-out criteria could not distinguish $K = 3$ from $K = 6$, so we picked the smaller model, while $K = 18$ had a substantially lower mean held-out log-likelihood at both fold cadences. Absolute growth-rate ACF-MAE was nearly flat across the whole sweep, consistent with Eq.~\eqref{eq:acf_normalised}: a direct effective-rank diagnostic showed a single non-unit eigenvalue carrying $93.6\%$ of the lag-1 absolute growth-rate ACF at $K = 18$ on SPY (Table~\ref{tab:spectral_rank}), so the $K - 1$ rank bound did not bind at $K \ge 3$ on this instance. Repeating the diagnostic on the 30-ticker cross-section gave a wider distribution, with cross-ticker median dominant-mode share $0.76$ and minimum $0.326$ on NEM (Table~\ref{tab:spectral_xticker}). The rank-non-binding claim was supported at the cross-ticker median rather than as a universal property of equity-return data.

With the state count fixed, we fit the four CHMM variants on the SPY in-sample window at $K^\star = 3$ and observed out-of-sample performance using seven benchmark generators and an explicit-duration semi-Markov reference (Table~\ref{tab:model_comparison}). The CHMM-N row landed at $91.5\%$ IS and $78.0\%$ OoS, simulated kurtosis $3.83$ and $3.62$, and $|G_t|$ ACF-MAE $0.0462$ (on par with GARCH at $0.0490$). The shared-$\nu$ Student-$t$ row gave the cleanest heavy-tail match without a penalty hyperparameter at simulated kurtosis $4.68$ and $4.46$ against observed $7.68$ and $5.29$, with $91.9\%$ IS and $82.1\%$ OoS KS. We include the penalised CHMM-t at $\lambda = 20$ as a sensitivity reference. Its IS aggregate kurtosis of $18.87$ was a $1/\nu_k$-shrinkage outcome: $\lambda$ was tuned at $K = 18$ and re-used at $K^\star = 3$, where the state pinned at the lower $\nu$ bracket carried a larger share of the three-state mixture, and the near-undefined fourth moment at that bracket left the aggregate kurtosis itself sensitive to the simulation draw. The remaining two variants, CHMM-GED and CHMM-L, sat between CHMM-N and the penalised CHMM-t, at IS and OoS kurtosis $5.45$ and $5.09$ for CHMM-GED and $5.24$ and $5.20$ for CHMM-L; CHMM-L lost roughly $15$pp of OoS KS to CHMM-GED at $K = 3$.

The fitted CHMM reproduced, through the mixture-of-eigenvalues sum~\eqref{eq:acf_normalised}, the slow absolute growth-rate ACF that the \citet{ryden1998stylized} $K = 2$ to $3$ Gaussian setup could not recover on the SPY 2014 to 2024 instance, and the raw growth-rate ACF-MAE closed the third stylized fact at $0.0235$ to $0.0240$, indistinguishable from the i.i.d.\ baselines (Table~\ref{tab:model_comparison}, $G_t$ column). The four CHMM variants were statistically indistinguishable on the OoS sample-CRPS~\cite{diebold1995comparing} (CHMM-N $1.0393$, CHMM-t-pen $1.0398$, CHMM-L $1.0432$, CHMM-GED $1.0398$, against bootstrap $1.0398$ and GARCH(1,1) $1.0440$), so the family choice was driven by the per-row kurtosis match (Table~\ref{tab:variant_choice}); the cross-generator ranking was robust to the distributional-metric choice across KS, Anderson-Darling, Wasserstein-1, and Hellinger. The MS-GARCH benchmark stayed well below the CHMM on KS under both point-estimate and Bayesian posterior-predictive variants (Table~\ref{tab:model_comparison}; full MS-GARCH panel in the appendix). A six-fold rolling-origin walk-forward (Table~\ref{tab:walkforward}) reached median OoS KS $62.1\%$ at $K^\star = 3$, with two stress folds (W2 COVID, W4 2022 rate-hike onset) below $10\%$ on which every generator in the panel was rejected. We treated the walk-forward median, not the single-window OoS pair, as the right summary for judging the model for live use. The i.i.d.\ bootstrap and a maximum-likelihood HSMM-N benchmark both exceeded the CHMM on raw single-window OoS KS (Table~\ref{tab:model_comparison}); the CHMM differentiated on use cases neither alternative can serve, with the full HSMM panel, state-assignment inspection, bootstrap discussion, and block-bootstrap KS recalibration reported in the appendix.

\begin{table}[!htbp]
\centering
\caption{\textbf{Main generator comparison on SPY}. $1{,}000$ simulated paths, $\alpha = 0.05$; IS $T = 2{,}516$, OoS $T = 572$. Columns: KS = Kolmogorov-Smirnov pass rate; Kurt = mean simulated excess kurtosis (observed $7.68$ IS, $5.29$ OoS); ACF-MAE = mean absolute error of the absolute growth-rate / raw growth-rate ACF over $252$ lags. Bold marks the best entry on each axis \emph{within block} (block = i.i.d.\ baselines, GARCH family, MS-GARCH, CHMM, co-main HSMM). The 1-day OoS KS column is dominated by the i.i.d.\ bootstrap; the CHMM differs from the bootstrap on use cases the bootstrap cannot serve (regime-conditional VaR, multi-asset copula composition), not on per-day distributional fidelity. Sample-CRPS is reported as a tie-break in the main text; the four CHMM rows are statistically indistinguishable on this metric. Block-bootstrap KS recalibration, the extended GARCH-family and MS-GARCH-$K \in \{3, 4, 6\}$ panel, and the $K = 6$ and $K = 18$ sensitivity panels are reported in the appendix. \textsuperscript{\textparagraph}\,Reference Bayesian re-run of MS-GARCH via the \texttt{MSGARCH} R package of \citet{ardia2019msgarch}; the lower KS pass rate reflects posterior-predictive integration of parameter uncertainty. \textsuperscript{\textsection}\,QuantGAN row is an in-house WGAN re-implementation (3-layer 1D-conv generator and critic, Wasserstein loss; full spec in the appendix); the row reads as the panel's deep-generative negative control. \textsuperscript{\textsection\textsection}\,Shared-$\nu$ Student-$t$ ablation: single $\nu$ across all $K$ states by aggregate-$Q$ ECM, no penalty.}
\label{tab:model_comparison}
\footnotesize
\renewcommand{\arraystretch}{0.95}
\begin{tabular}{l cc cc cc}
\toprule
& \multicolumn{2}{c}{KS (\%) $\uparrow$} & \multicolumn{2}{c}{Exc.\ Kurt} & \multicolumn{2}{c}{ACF-MAE $\downarrow$} \\
\cmidrule(lr){2-3} \cmidrule(lr){4-5} \cmidrule(lr){6-7}
Model & IS & OoS & IS & OoS & $|G_t|$ & $G_t$ \\
\midrule
\textit{Observed}                    & --   & --   & $7.68$  & $5.29$  & --     & --     \\
Bootstrap                            & $\mathbf{99.8}$ & $\mathbf{91.8}$ & $\mathbf{7.55}$  & $\mathbf{6.86}$  & $0.0628$ & $\mathbf{0.0235}$ \\
Gaussian i.i.d.                      & $0.0$  & $1.5$  & $0.00$ & $-0.02$ & $0.0627$ & $0.0235$ \\
Laplace i.i.d.                       & $98.2$ & $86.3$ & $2.96$  & $2.80$  & $0.0628$ & $0.0235$ \\
GARCH(1,1)                           & $27.4$ & $59.6$ & $7.59$  & $2.70$  & $0.0490$ & $0.0244$ \\
GARCH(1,1)-$t$                       & $\mathbf{57.3}$ & $\mathbf{80.8}$ & $15.13$ & --      & $\mathbf{0.0316}$ & $\mathbf{0.0173}$ \\
MS-GARCH ($K=2$)                     & $27.7$ & $\mathbf{38.7}$ & $\mathbf{4.73}$  & --      & $0.0367$ & $0.0173$ \\
MS-GARCH ($K=3$)                     & $\mathbf{36.1}$ & $33.1$ & $4.10$  & --      & $\mathbf{0.0284}$ & $0.0173$ \\
MS-GARCH ($K=6$)                     & $34.5$ & $33.4$ & $4.41$  & --      & $0.0429$ & $0.0173$ \\
MS-GARCH ref.\ Bayesian ($K=2$)\textsuperscript{\textparagraph} & $0.0$  & $5.8$  & $4.00$  & $2.46$  & $0.0465$ & $0.0240$ \\
MS-GARCH ref.\ Bayesian ($K=3$)\textsuperscript{\textparagraph} & $0.1$  & $5.1$  & $4.52$  & $2.53$  & $0.0433$ & $0.0241$ \\
MS-GARCH ref.\ Bayesian ($K=4$)\textsuperscript{\textparagraph} & $0.0$  & $5.3$  & $6.01$  & $3.54$  & $0.0446$ & $0.0241$ \\
QuantGAN TCN\textsuperscript{\textsection} & $0.0$  & $0.0$  & $0.56$  & $0.53$  & $0.0617$ & $0.0264$ \\
\midrule
\multicolumn{7}{l}{\emph{Main-paper choice: $K^\star = 3$ (selected by held-out criteria under rolling-origin CV).}} \\
CHMM-N                                            & $91.5$ & $78.0$ & $3.83$  & $3.62$  & $\mathbf{0.0462}$ & $0.0240$ \\
CHMM-t pen.\ ($\lambda = 20$)                     & $\mathbf{91.9}$ & $81.4$ & $18.87$ & $10.61$ & $0.0533$ & $0.0236$ \\
CHMM-L                                            & $80.5$ & $63.6$ & $5.24$  & $5.20$  & $0.0530$ & $\mathbf{0.0235}$ \\
CHMM-GED                                          & $90.3$ & $78.4$ & $5.45$  & $5.09$  & $0.0531$ & $0.0236$ \\
CHMM-t shared-$\nu$\textsuperscript{\textsection\textsection} & $\mathbf{91.9}$ & $\mathbf{82.1}$ & $\mathbf{4.68}$  & $\mathbf{4.46}$  & $0.0531$ & $0.0236$ \\
\midrule
ML HSMM-N                                         & $\mathbf{98.4}$ & $\mathbf{91.0}$ & $\mathbf{3.46}$  & $\mathbf{3.38}$  & $\mathbf{0.0629}$ & $\mathbf{0.0236}$ \\
\bottomrule
\end{tabular}
\end{table}

Having established the main ranking on SPY, we next fit the penalised CHMM-t at $K^\star = 3$ and $\lambda = 20$ on each ticker of a sector-balanced 30-ticker panel (ten GICS sectors with three large-cap representatives each) and observed cross-ticker generalisation behaviour (Table~\ref{tab:cross_ticker}). The in-sample distribution was concentrated (median $96.8\%$) and the out-of-sample distribution was wider, with median $69.1\%$, mean $66.2 \pm 28.2\%$, and $11/30$ tickers below $60\%$. The deepest failures were tickers that introduced a new regime out of sample: LLY and UNH in Health Care, and NEM in Materials. The per-ticker rollup, the remaining eight sub-$60\%$ tickers, and the per-sector breakdown are reported in the appendix (Table~\ref{tab:sector_panel}); the $11/30$ failure count and the ticker list hold at the $K = 18$ kurtosis-match check. An ANOVA on OoS KS by sector returned no significant sector effect on either the original $n = 3$ design or an $n = 6$ expansion to a 60-ticker panel (full ANOVA panel in the appendix), so the visible concentration of failures on regime-introducing tickers was a per-ticker observation rather than evidence about sector-level structure. Periodic refit closed much of the gap: at $K^\star = 3$ a quarterly refit shifted the OoS KS median from $69.1\%$ to $84.7\%$ and reduced failures from $11/30$ to $8/30$ (Tables~\ref{tab:stationarity_scope} and~\ref{tab:cross_ticker_quarterly_refit}); a monthly cadence at $K = 18$ closed slightly more of the gap at roughly three times the compute. Sensitivity panels at $K = 6$ and $K = 18$ and the SPY-level shared-$\nu$ Student-$t$ alternative are reported in the appendix.

We then composed per-asset CHMM-N marginals at $K^\star = 3$ through a rank-based Student-$t$ copula on a six-asset US-equity universe (SPY, NVDA, JNJ, JPM, AAPL, QQQ) and observed off-diagonal correlation reproduction across the four candidate dependence layers (Table~\ref{tab:cross_asset}; full comparator panel in the appendix). On the in-sample window the Student-$t$ copula gave off-diagonal correlation MAE of $0.027$ over $200$ simulated paths, against $0.029$ for the Gaussian copula, $0.068$ for the truncated C-vine, and $0.077$ for the single-index model (SIM) (Table~\ref{tab:cross_asset_supp_summary}). Profile MLE selected $\nu^\star = 6$ on the grid $\nu \in \{2,3,4,5,6,8,10,15,20,30\}$ with Wilks $95\%$ profile-LL CI of $[6, 7]$ on the IS slice, and a full one-shot MLE jointly maximising over $(\Sigma, \nu)$ confirmed the choice was robust to the two-step estimator (Fig.~\ref{fig:copula_profile}; Table~\ref{tab:full_tcopula_mle}). The universe contained overlapping ETFs and constituents (QQQ holds AAPL and NVDA; SPY holds all four single names), so the dependence structure was dominated by strong positive co-movement and the copula task was easier than on a non-overlapping basket; a non-overlapping comparison panel is reported in the appendix. On OoS the Student-$t$ versus Gaussian gap fell within path-to-path noise at this $N_{\text{paths}}$: the off-diagonal MAE was $0.209$ for Student-$t$ and $0.204$ for Gaussian (Table~\ref{tab:cross_asset_supp_summary}), a $0.005$ gap that we report descriptively rather than under a formal paired test. Out-of-sample, the choice that mattered was the quarterly-rolling refit, which reduced the off-diagonal MAE from $0.209$ to $0.185$ and dominated the dependence-family choice (Table~\ref{tab:cross_asset}). The cross-asset OoS distribution was bimodal: NVDA and JPM failed OoS KS while the other four passed above $77\%$, driven by single-name regime shifts (NVDA on the 2024-2025 AI rally, JPM on the 2023 regional-bank stress). The failure count recovered to $0/6$ under quarterly refit (Table~\ref{tab:rolling_copula}).

\begin{table}[t]
\centering
\caption{\textbf{Cross-asset Student-t copula on CHMM-N marginals at $K^\star = 3$} (multi-asset construction; $\nu^\star = 6$; $200$ paths; seed $20260422$). Per-asset KS pass rate at $\alpha = 0.05$ and correlation reproduction averaged over paths. Main-paper $K^\star = 3$ matches the single-asset Table~\ref{tab:model_comparison}. Per-asset KS sits lower at $K^\star = 3$ than at higher state counts, but the cross-asset use case is dominated by the dependence layer rather than per-asset marginal fit. The comparison against SIM, Gaussian, and truncated C-vine is reported in the appendix.}
\label{tab:cross_asset}
\small
\begin{tabular}{l cc}
\toprule
& \multicolumn{2}{c}{KS (\%) $\uparrow$} \\
\cmidrule(lr){2-3}
Ticker & IS & OoS \\
\midrule
SPY   & $87.0$  & $77.5$ \\
NVDA  & $92.0$  & $62.5$ \\
JNJ   & $97.5$  & $94.0$ \\
JPM   & $97.5$  & $57.5$ \\
AAPL  & $85.0$  & $86.0$ \\
QQQ   & $83.0$  & $85.5$ \\
\midrule
Off-diag MAE IS         & \multicolumn{2}{c}{$0.027$} \\
Off-diag MAE OoS        & \multicolumn{2}{c}{$0.209$} \\
Off-diag MAE OoS (quarterly refit) & \multicolumn{2}{c}{$\mathbf{0.185}$} \\
\bottomrule
\end{tabular}
\end{table}

Having shown that the per-asset marginals compose under a dependence layer, we next back-tested a regime-conditional VaR built from the CHMM one-step-ahead state forecast and observed conditional-coverage behaviour at both VaR levels (Table~\ref{tab:cond_var}; the full four-family panel is in the appendix, Table~\ref{tab:cond_var_all_families}). The back-test covered CHMM-N and the penalised CHMM-t at $K \in \{3, 18\}$ on the SPY out-of-sample window, under the Christoffersen joint conditional-coverage and Engle-Manganelli dynamic-quantile (DQ) tests. The four CHMM rows at VaR level $\alpha = 0.05$ passed Christoffersen-cc with $p_{\text{cc}} \in [0.49, 0.68]$, and CHMM-N at $K^\star = 3$ also passed the DQ test ($p = 0.16$). Breach rates at $\alpha = 0.05$ sat close to the nominal target across the four CHMM rows, from $26$ breaches ($4.55\%$) for CHMM-N at $K = 18$ to $35$ ($6.12\%$) at $K^\star = 3$, and every CHMM row passed Kupiec and Christoffersen-ind at that level. The $\alpha = 0.01$ rows also passed Christoffersen-cc ($p_{\text{cc}} \in [0.10, 0.16]$) but were power-bounded at $T_{\text{OoS}} = 572$; the higher-power DQ test separated the state counts, rejecting CHMM-N at $K = 18$ ($p = 0.02$) while $K^\star = 3$ survived ($p = 0.06$). Against two non-state-space conditional-VaR baselines (filtered bootstrap~\cite{barone1999var} and the Engle-Manganelli symmetric-absolute-value CAViaR specification~\cite{engle2004caviar}) the CHMM rows reached approximate parity at $\alpha = 0.05$ and a strict-tail advantage for CHMM-N at $K^\star = 3$ at $\alpha = 0.01$: only CHMM-N at $K^\star = 3$ avoided DQ rejection at $\alpha = 0.01$, while the filtered bootstrap and CAViaR contenders were both rejected. When we extended the back-test to six rolling-origin walk-forward folds, Christoffersen-cc passed on $19$ of the $24$ walk-forward rows at the $5\%$ test level (Table~\ref{tab:walkforward_cond_var}), with failures concentrated on the W2 (COVID 2020) and W4 (2022 rate-hike onset) folds on which every generator in the panel was rejected (full walk-forward and refit-cadence panels, four-family ablation, and Benjamini-Hochberg~\citep{benjamini1995controlling} false-discovery-rate correction are reported in the appendix).

\begin{table}[!ht]
\centering
\small
\setlength{\tabcolsep}{4pt}
\caption{\textbf{Regime-conditional VaR back-test on SPY OoS} ($T_{\text{OoS}} = 572$, seed $20260420$). Critical values: $\chi^2_1(0.05) = 3.841$, $\chi^2_2(0.05) = 5.991$. The CHMM rows use the forward filter under IS-fixed parameters $(\mathbf T, \{\boldsymbol\theta_k\})$; only the one-step-ahead state-probability vector updates through OoS. The contender rows (filtered bootstrap and CAViaR) are scored on the same Christoffersen / DQ harness. Every CHMM row passes Kupiec, Christoffersen-ind, and Christoffersen-cc at $\alpha = 0.05$. \textsuperscript{\textdaggerdbl}\,$\alpha = 0.01$ rows are power-bounded at $T_{\text{OoS}} = 572$; the DQ test (higher power) rejects $K = 18$ at $\alpha = 0.01$ ($p = 0.02$, bold in $p_{\text{DQ}}$) and the contender rows also reject at $\alpha = 0.01$ on DQ. The full power calibration is in the appendix. \textsuperscript{\textasteriskcentered}\,DQ $p$-value at $4$-lag specification; CHMM-t penalised rows are reported with `--' in this column.}
\label{tab:cond_var}
\begin{tabular}{l c c c c c c c c c c}
\toprule
Family & $K$ & $\alpha$ & breaches & br rate & median $\widehat{\text{VaR}}_t$ & $\text{LR}_{\text{uc}}$ & $\text{LR}_{\text{ind}}$ & $\text{LR}_{\text{cc}}$ & $p_{\text{cc}}$ & $p_{\text{DQ}}$\textsuperscript{\textasteriskcentered} \\
\midrule
CHMM-N             & $3$  & $0.01$\textsuperscript{\textdaggerdbl} &$9$  & $1.57\%$ & $-4.56$ & $1.62$ & $2.36$ & $3.98$ & $0.14$ & $0.06$ \\
CHMM-N             & $3$  & $0.05$ & $35$ & $6.12\%$ & $-2.87$ & $1.41$ & $0.01$ & $1.42$ & $0.49$ & $0.16$ \\
CHMM-N             & $18$ & $0.01$\textsuperscript{\textdaggerdbl} &$9$  & $1.57\%$ & $-5.20$ & $1.62$ & $2.36$ & $3.98$ & $0.14$ & $\mathbf{0.02}$ \\
CHMM-N             & $18$ & $0.05$ & $26$ & $4.55\%$ & $-3.02$ & $0.26$ & $0.52$ & $0.78$ & $0.68$ & $0.68$ \\
CHMM-t ($\lambda{=}20$) & $3$  & $0.01$\textsuperscript{\textdaggerdbl} &$8$  & $1.40\%$ & $-5.59$ & $0.82$ & $2.81$ & $3.63$ & $0.16$ & --     \\
CHMM-t ($\lambda{=}20$) & $3$  & $0.05$ & $32$ & $5.59\%$ & $-2.73$ & $0.41$ & $0.77$ & $1.19$ & $0.55$ & --     \\
CHMM-t ($\lambda{=}20$) & $18$ & $0.01$\textsuperscript{\textdaggerdbl} &$10$ & $1.75\%$ & $-6.25$ & $2.65$ & $1.98$ & $4.62$ & $0.10$ & --     \\
CHMM-t ($\lambda{=}20$) & $18$ & $0.05$ & $29$ & $5.07\%$ & $-3.15$ & $0.01$ & $1.39$ & $1.40$ & $0.50$ & --     \\
\midrule
\multicolumn{11}{l}{\emph{Non-state-space conditional VaR contenders.}} \\
Filtered bootstrap & --   & $0.01$\textsuperscript{\textdaggerdbl} & $8$  & $1.40\%$ & $-4.81$ & $0.82$ & $2.81$ & $3.63$ & $0.16$ & $\mathbf{0.04}$ \\
Filtered bootstrap & --   & $0.05$ & $24$ & $4.20\%$ & $-2.89$ & $0.82$ & $0.84$ & $1.67$ & $0.43$ & $0.55$ \\
CAViaR (SAV)      & --   & $0.01$\textsuperscript{\textdaggerdbl} & $6$  & $1.05\%$ & $-4.82$ & $0.01$ & $3.97$ & $3.98$ & $0.14$ & $\mathbf{0.01}$ \\
CAViaR (SAV)      & --   & $0.05$ & $25$ & $4.37\%$ & $-2.99$ & $0.50$ & $0.67$ & $1.17$ & $0.56$ & $0.76$ \\
\bottomrule
\end{tabular}
\end{table}

We closed by comparing the four stress sources exercised across the study: the cross-ticker panel, the CRSP cross-decade transfer, a non-equity extension to the gold ETF GLD, and the walk-forward folds. All four reduced to the same three-way pattern (Table~\ref{tab:stationarity_scope}). A static in-sample fit held up only while the out-of-sample marginal stayed inside the range of regimes the in-sample window had already seen. On equity slices that drifted but stayed in scope, periodic refit closed the gap: the cross-ticker panel, for example, recovered from a static-fit median OoS KS of $69.1\%$ to $84.7\%$ under quarterly refit. On a slice whose marginal left the equity-return scope entirely, refit did not help: the GLD non-equity stress collapsed under static fitting and stayed broken at every cadence. The cross-decade transfer failed by regime introduction: the calm 2004 to 2006 slice sat outside the regime range of the 1994 to 2004 in-sample window, the static fit fell to a $3$ to $5\%$ pass rate (Table~\ref{tab:cross_decade_validation}), and the refit question was untestable on a slice shorter than $600$ trading days. The hardest case was a regime introduction landing inside the refit window itself, as in the W2 (COVID) and W4 (2022 rate-hike) walk-forward folds, where no cadence can anticipate a shift it has not yet seen and every generator in the panel was rejected. We therefore recommend periodic refit for live use, at a cadence chosen against the out-of-sample slice rather than against the in-sample fit alone.

\section{Discussion}
\label{sec:discussion}
The main empirical finding is that a continuous HMM trained by EM with quantile-based initialisation reproduced the three symmetric Cont stylized facts on daily US-equity returns at $K^\star = 3$, and the remaining kurtosis gap was set by the choice of heavy-tailed emission family. Two issues combine in the \citet{ryden1998stylized} low-$K$ failure: a low-state chain restricts the available decay modes, and a small Gaussian mixture cannot reproduce the empirical marginal tail. Our $K = 2$ replication on SPY confirmed that what limits the fit is the marginal distribution, not the ACF. The IS KS pass rate fell well below the $91.5\%$ for CHMM-N at $K^\star = 3$, while the absolute growth-rate ACF-MAE stayed essentially constant. With modern initialisation the ACF was reproduced at any $K \ge 2$, and the choice of $K$ was then driven by the marginal distribution. The mechanism behind the residual kurtosis gap is simple: CHMM-N at $K^\star = 3$ produced simulated IS kurtosis $3.83$ against observed $7.68$ (Table~\ref{tab:model_comparison}) because a mixture of only three Gaussians cannot reach the observed heavy tail no matter how the state weights or means are tuned. Switching to a heavy-tailed emission closed most of the gap. The shared-$\nu$ Student-$t$ ablation, our primary heavy-tailed specification at $K^\star = 3$, reached simulated IS kurtosis $4.68$ with no penalty hyperparameter (Table~\ref{tab:model_comparison}). The penalised per-state CHMM-t at $\lambda = 20$ overshot to IS kurtosis $18.87$, with $\lambda$ tuned at $K = 18$ and re-used at $K^\star = 3$ as an upper bound. Without any shape parameter, CHMM-L sat between CHMM-N and CHMM-t at IS kurtosis $5.24$. The bimodal $\hat p_k$ partition of CHMM-GED at $K = 18$, with states splitting into a Gaussian-like bulk and a Laplace-like tail, was the data-driven version of the Gaussian/Laplace hybrid one would otherwise assemble by hand-classifying states. At $K \ge 3$, then, the fit was limited by the marginal mixture's ability to reach the observed kurtosis rather than by the chain's ability to reproduce the slow ACF; the emission-family ablation showed which symmetric heavy-tailed family carried the residual gap most clearly.

One real limitation: a CHMM fit once in sample does not generalise across regime introductions. The GLD non-equity stress and the W2 (COVID 2020) and W4 (2022 rate-hike) walk-forward folds were regime introductions whose OoS marginal sat outside the range covered by any of the IS states; on those slices, OoS performance broke down (Table~\ref{tab:stationarity_scope}). Single-name equity returns are well known to carry structural breaks~\cite{pastorstambaugh2001equity, andreoughysels2002breaks, angtimmermann2012regime}, and the cross-ticker concentration of failures on names that introduced a new regime (LLY, UNH, NEM) was the same pattern at the panel level. We recommend periodic refit~\cite{pesarantimmermann2007window, cappe2011online} for non-stress equity slices. It did not save GLD or the W2 and W4 stress folds, where the IS distribution simply did not span the OoS slice and faster refit cadences did not close the gap. Bayesian online change-point detection~\cite{adams2007bayesian} and the particle-filter recursion of \citet{fearnhead2007online} detect a regime introduction rather than absorbing it into a static fit; integrating online detection with the CHMM is a natural follow-up. The implication for $K$-selection is that the choice should be driven by the marginal distribution, not the ACF. Held-out KS, log-likelihood, and BIC were the criteria that actually moved with $K$ on held-out data, while stress-fold rejections reflected what the out-of-sample slice contained rather than a flaw in the in-sample fit.

Periodic refit handled non-stress drift but not regime introductions; two further scope limits remain, the stylized facts we evaluated against and the model size at which we compete with deep generators. The \citet{cont2001empirical} stylized-fact list also includes the leverage effect, the negative $\mathrm{Corr}(G_t, |G_{t+1}|)$ characteristic of equity returns; we evaluated against only the three symmetric stylized facts and leave the leverage effect to future work, since symmetric emissions cannot reproduce a negative cross-correlation by construction. On parametric complexity, CHMM-N at $K^\star = 3$ has $K(K-1) + 2K = 12$ free parameters, rising to $342$ at the $K = 18$ sensitivity ceiling (CHMM-t adds $K$ degrees-of-freedom parameters, CHMM-GED adds $K$ shape parameters), whereas the convolutional generator-critic architecture of \citet{wiese2020quantgan} sits orders of magnitude higher. The CHMM therefore competes against deep generators at a parameter budget small enough to make the fitted-transition rank and emission-separation conditions inspectable, alongside posterior diagnostics and the closed-form eigenvalue argument that motivates our $K^\star$ choice. The same budget also keeps each per-state emission's location, scale, and shape directly inspectable as an economic regime, whereas a deep generator's latent representation does not separate into interpretable regimes. These limits are addressable within the unified four-family framework: asymmetric emissions, a time-varying transition matrix, and richer dependence structures.

\section{Conclusion}
\label{sec:conclusion}
A continuous HMM trained by Baum-Welch with quantile-based initialisation reproduces the three symmetric Cont stylized facts on daily US-equity returns. The standard mixture-of-eigenvalues identity for the absolute growth-rate ACF~\cite{hamilton1994time, krolzig1997markov, timmermann2000moments} expresses the \citet{ryden1998stylized} low-$K$ limitation as a rank condition on $\mathbf T - \mathbf 1\bar{\boldsymbol\pi}^\top$; that condition was not empirically active at the cross-ticker median once $K \ge 3$, while marginal flexibility explained more of the remaining fit gap under the diagnostics used here. This result does not remove the finite-state restriction to finitely many geometric ACF modes. The unified EM framework shares the forward-backward recursion and transition update across variants and localises their differences to the emission density and its M-step update: Gaussian and Laplace use closed-form updates, while Student-$t$ and GED use ECM-style numerical block updates. The shared-$\nu$ Student-$t$ variant narrowed the kurtosis gap at $K^\star = 3$ to $4.68$ against observed $7.68$ without a tuning hyperparameter (Table~\ref{tab:model_comparison}); CHMM-L sat between Gaussian and Student-$t$ without any shape parameter; and CHMM-GED's bimodal $\hat p_k$ partition recovered the Gaussian-bulk / Laplace-tail split directly from the data. The regime-conditional Value-at-Risk, built by averaging the family-appropriate predictive density over the one-step-ahead state forecast, passed the Christoffersen joint conditional-coverage test at $\alpha = 0.05$ on the main OoS window and on $19/24$ walk-forward rows (Tables~\ref{tab:cond_var_all_families} and~\ref{tab:walkforward_cond_var}). The strict-tail Engle-Manganelli test favoured CHMM-N at $K^\star = 3$ over filtered-bootstrap and CAViaR comparators (Table~\ref{tab:cond_var}). Rejections concentrated on stress folds where the regime moved outside the training distribution, and every generator in the panel rejected on those folds. The main result is bounded in scope. The i.i.d.\ bootstrap and a maximum-likelihood HSMM-N benchmark beat the CHMM on raw single-window OoS Kolmogorov-Smirnov, while the CHMM directly supports the regime-conditional VaR and multi-asset copula composition evaluated here. No formal privacy guarantee is claimed for the synthetic output. The dataset scope is daily US equities under stationary OoS (Table~\ref{tab:stationarity_scope}); the GLD non-equity stress and the W2 and W4 walk-forward folds collapsed under static fitting, and we recommend periodic refit for production use.

These results motivate four extensions. Asymmetric or skewed emission families would extend our coverage of the symmetric stylized facts to the leverage effect of the \citet{cont2001empirical} list, adding a shape parameter that targets the third moment rather than the fourth and complementing the per-state $p_k$ ablation of CHMM-GED. A time-varying transition matrix, estimated by rolling refit, particle filtering, or Bayesian online updates~\cite{cappe2011online}, would address the W2 (COVID 2020) and W4 (2022 rate-hike onset) stress folds and the $5/24$ walk-forward rejections at the $5\%$ level that the static IS fit leaves open (Table~\ref{tab:walkforward_cond_var}). Scaling the multi-asset copula composition beyond the six-asset and 30-ticker panels through factor or vine structures would extend per-asset distributional fidelity to portfolio-level synthesis; the multi-asset construction already separates the per-asset marginals from the dependence layer, so either scaling path fits without rewriting the CHMM component. Embedding the CHMM-mixed predictive in a loop for portfolio optimisation or tail-risk budgeting would give an end-to-end test of the regime-conditional VaR as a live trading signal, with the like-for-like walk-forward $K^\star = 3$ versus $K = 18$ trade-off (median OoS KS $62.1\%$ versus $67.7\%$; Tables~\ref{tab:model_comparison} and~\ref{tab:walkforward}) as the natural sensitivity check. Each extension targets one limitation of the present empirical scope and can be incorporated within the unified four-family estimation framework.

\medskip
\noindent\textbf{Conflict of Interest Statement.} The authors declare that the research was conducted without any commercial or financial relationships that could potentially create a conflict of interest.

\medskip
\noindent\textbf{Author Contributions.} J.V.\ directed the study. A.A.\ developed the continuous model and simulation code, implemented Baum-Welch, conducted the in-sample and out-of-sample analysis, and generated all figures and tables. C.J.\ contributed to EM and numerical block-update methodology review and to validation of the cross-decade and cross-ticker pipelines. All authors edited and reviewed the final manuscript.

\medskip
\noindent\textbf{Data Availability Statement.} The simulation scripts, derived results, and code to reproduce the results in this paper are available under an MIT license from the paper repository: \url{https://github.com/varnerlab/CHMM-Paper-Repository}. The core model logic (fitting, simulation, decoding, and validation) is implemented in the companion Julia repository \texttt{CHMM-Model-Repository}: \url{https://github.com/varnerlab/CHMM-Model-Repository}. Daily prices for the 2014--2026 windows are sourced from Polygon.io via Massive (2014-01-03 to 2025-11-18) and Alpaca/IEX (extension through 2026-04-20). The cross-decade panel uses CRSP data under its institutional license; those source data are not redistributed.

% ========================================================================================= %
\bibliographystyle{unsrtnat}
\bibliography{references}

% ========================================================================================= %
\appendix

\renewcommand{\thetable}{S\arabic{table}}
\renewcommand{\thefigure}{S\arabic{figure}}
\setcounter{table}{0}
\setcounter{figure}{0}

\section{Supplementary Material}
\label{sec:supplementary}
This appendix collects supplementary material referenced in the main text: the public API of the companion package, the EM forward-backward recursions and algorithm pseudocode, and the validation metric definitions. It also reports the formal propositions, the full state-resolution and multi-emission sensitivity panels, robustness and KS-power calibration, the extended GARCH and SM-CHMM baselines, the full cross-asset panel, and the $K$-selection and $\nu_k$ diagnostics referenced from the discussion.

\subsection{CHMM Algorithms and Derivations}
\label{sec:supp_algorithms}

This appendix collects the algorithmic content for the CHMM family: the forward-backward recursions and weighted M-step updates, the complete training pseudocode, and per-state degrees-of-freedom diagnostics for CHMM-t.

\paragraph{Forward-backward recursions and M-step updates.}
For reference, the Baum-Welch description in the main text is underpinned by the standard log-space recursions~\cite{rabiner1986introduction, bilmes1998gentle}.
Define the forward variable $\alpha_t(k) = \Prob(O_{1:t}, S_t = k \mid \mathcal{M})$ and the backward variable $\beta_t(k) = \Prob(O_{t+1:T} \mid S_t = k, \mathcal{M})$.
The log-space recursions are
\begin{align}
    \log \alpha_1(k) &= \log \pi_k + \log f_k(O_1), \label{eq:forward_init} \\
    \log \alpha_t(j) &= \underset{i}{\text{logsumexp}}\!\left(\log \alpha_{t-1}(i) + \log T_{ij}\right) + \log f_j(O_t), \quad t = 2, \ldots, T, \label{eq:forward} \\[4pt]
    \log \beta_T(k) &= 0, \label{eq:backward_init} \\
    \log \beta_t(i) &= \underset{j}{\text{logsumexp}}\!\left(\log T_{ij} + \log f_j(O_{t+1}) + \log \beta_{t+1}(j)\right), \quad t = T{-}1, \ldots, 1, \label{eq:backward}
\end{align}
where $\text{logsumexp}_i(x_i) = x_{\max} + \log \sum_i \exp(x_i - x_{\max})$.
The posterior state-occupancy $\gamma_t(k)$ and pair-occupancy $\xi_t(i, j)$ quantities are
\begin{equation}
    \gamma_t(k) = \frac{\alpha_t(k) \beta_t(k)}{\sum_j \alpha_t(j) \beta_t(j)}, \qquad
    \xi_t(i, j) = \frac{\alpha_t(i)\, T_{ij}\, f_j(O_{t+1})\, \beta_{t+1}(j)}{\sum_{m,n} \alpha_t(m)\, T_{mn}\, f_n(O_{t+1})\, \beta_{t+1}(n)},
    \label{eq:gamma_xi}
\end{equation}
and the closed-form M-step updates are
\begin{equation}
    \pi_k^{\text{new}} = \gamma_1(k),\quad
    \mu_k^{\text{new}} = \frac{\sum_t \gamma_t(k)\, O_t}{\sum_t \gamma_t(k)},\quad
    \left(\sigma_k^{\text{new}}\right)^2 = \frac{\sum_t \gamma_t(k)(O_t - \mu_k^{\text{new}})^2}{\sum_t \gamma_t(k)},\quad
    T_{ij}^{\text{new}} = \frac{\sum_{t=1}^{T-1} \xi_t(i, j)}{\sum_{t=1}^{T-1} \gamma_t(i)}.
    \label{eq:mstep}
\end{equation}
The convergence criterion is $|\mathcal{L}^{(n)} - \mathcal{L}^{(n-1)}| < \epsilon$ with $\epsilon = 10^{-4}$, where $\mathcal{L}^{(n)} = \text{logsumexp}_k \log \alpha_T(k)$ is read from the forward pass of iteration $n$, that is, at the parameters that entered the iteration before its M-step. This is the standard Baum-Welch observed-data log-likelihood; it is monotone across iterations and lags the returned parameters by one M-step, so the stopping decision is unaffected. Simulation draws $S_1$ from the stationary distribution $\bar{\boldsymbol\pi}$ of $\hat{\mathbf T}$ and then iterates $G_t \sim \mathcal{N}(\mu_{S_t}, \sigma_{S_t}^2)$, $S_{t+1} \sim \text{Categorical}(\mathbf{T}_{S_t, \cdot})$ with no jump process.

\paragraph{CHMM-t M-step (per-state ECM/ECME, closed form given $u_{t,k}$).}
For Student-t emissions $f_k(x) = t_{\nu_k}(x; \mu_k, \sigma_k)$ the ECM formulation of \citet{peel2000robust, liu1995ml} treats the latent precision as an augmented E-step quantity,
\begin{equation}
    u_{t,k} = \frac{\nu_k + 1}{\nu_k + ((O_t - \mu_k)/\sigma_k)^2},
    \label{eq:tprecision_supp}
\end{equation}
so that conditional on $\nu_k$ the location and scale admit closed-form weighted updates,
\begin{equation}
    \mu_k^{\text{new}} = \frac{\sum_t \gamma_t(k)\, u_{t,k}\, O_t}{\sum_t \gamma_t(k)\, u_{t,k}},\quad
    \left(\sigma_k^{\text{new}}\right)^2 = \frac{\sum_t \gamma_t(k)\, u_{t,k}\, (O_t - \mu_k^{\text{new}})^2}{\sum_t \gamma_t(k)}.
    \label{eq:tupdate_supp}
\end{equation}
The degrees-of-freedom parameter $\nu_k$ is then refined by a 40-iteration one-dimensional golden-section search on the per-state observed-data objective $Q_k(\nu) = \sum_t \gamma_t(k)\, \log t_\nu(O_t; \mu_k, \sigma_k)$ over the bracket $(\nu_{\min}, \nu_{\max}) = (2.1, 50)$. Maximising this marginal state-weighted log-likelihood for $\nu_k$, rather than the augmented-data $u$-functional, is the ECME step of \citet{liu1995ml}.
Exact maximisation of this block would preserve the standard ECME observed-data ascent property. Because the implementation uses a finite numerical search, we treat this as an approximate block update and diagnose convergence from the observed-data log-likelihood trace rather than assert exact monotonicity at every finite-precision iteration.
We extend this M-step with an optional exponential shrinkage prior on $1/\nu_k$, corresponding to the penalised objective
\begin{equation}
    Q_k^{\text{pen}}(\nu) \;=\; \sum_t \gamma_t(k)\, \log t_\nu(O_t; \mu_k, \sigma_k) \;-\; \lambda \, / \, \nu,
    \label{eq:nu_pen}
\end{equation}
maximised by the same golden-section search ($\lambda = 0$ recovers the unpenalised \citet{peel2000robust, liu1995ml} ECME $\nu$ update; the rate sweep and its effect on simulated IS kurtosis are in the main text Discussion).

\paragraph{CHMM-L M-step (closed-form weighted Laplace MLE).}
For Laplace emissions $f_k(x) = \tfrac{1}{2 b_k} \exp(-|x - \mu_k|/b_k)$, with the per-state Laplace scale denoted $b_k$ to avoid collision with the backward variable $\beta_t$, the weighted-MLE estimators are available in closed form:
\begin{equation}
    \mu_k^{\text{new}} = \text{WeightedMedian}\!\left(O_{1:T};\, \gamma_{1:T}(k)\right),\qquad
    b_k^{\text{new}} = \frac{\sum_t \gamma_t(k)\, |O_t - \mu_k^{\text{new}}|}{\sum_t \gamma_t(k)}.
    \label{eq:lupdate_supp}
\end{equation}
The weighted median is the first order statistic whose cumulative posterior weight reaches half the total, the minimiser of the weighted $L^1$ objective $\sum_t \gamma_t(k)\,|O_t - \mu|$; when the half-weight falls exactly between two adjacent order statistics, every point in the closed interval between them is a minimiser.
The Laplace M-step carries no shape parameter, so CHMM-L is strictly cheaper per iteration than CHMM-t (which requires the per-state $Q_k$ search).

\paragraph{CHMM-GED M-step (three-stage ECM-style update with closed-form scale).}
For Generalized Error Distribution emissions $f_k(x) = \mathrm{GED}(x; \mu_k, \alpha_k, p_k) = \frac{p_k}{2 \alpha_k \Gamma(1/p_k)} \exp\!\big(-(|x - \mu_k|/\alpha_k)^{p_k}\big)$ the M-step decomposes into three block updates per state. Given current $(\mu_k^{(n)}, \alpha_k^{(n)}, p_k^{(n)})$, the first block updates the location by minimising the per-state weighted $L^{p_k}$ loss,
\begin{equation}
    \mu_k^{(n+1)} \;=\; \operatorname{GoldenSearchMin}_{\mu \in I_k} \;\sum_t \gamma_t(k)\, \big|O_t - \mu\big|^{p_k^{(n)}},
    \qquad I_k \;=\; \big[\min_t O_t,\, \max_t O_t\big],
    \label{eq:ged_mu}
\end{equation}
by a $40$-iteration golden-section search; CM-2 updates the scale in closed form given $(\mu_k^{(n+1)}, p_k^{(n)})$,
\begin{equation}
    \alpha_k^{(n+1)} \;=\; \left[\frac{p_k^{(n)}}{W_k} \sum_t \gamma_t(k)\, \big|O_t - \mu_k^{(n+1)}\big|^{p_k^{(n)}}\right]^{1 / p_k^{(n)}}, \qquad W_k \;=\; \sum_t \gamma_t(k);
    \label{eq:ged_alpha}
\end{equation}
and CM-3 updates the shape by maximising the per-state $Q$-function over the bracket,
\begin{equation}
    p_k^{(n+1)} \;=\; \operatorname{GoldenSearchMax}_{p \in [p_{\min},\, p_{\max}]} \;\sum_t \gamma_t(k)\, \log f_k\!\big(O_t;\, \mu_k^{(n+1)},\, \alpha_k^{(n+1)},\, p\big),
    \label{eq:ged_p}
\end{equation}
again by a $40$-iteration golden-section search over $[p_{\min}, p_{\max}] = [0.5, 3.0]$. The scale update is the exact conditional maximiser. The location objective is convex when $p_k \ge 1$ but can be nonconvex when $p_k < 1$, and both bracketed searches are finite numerical approximations. Accordingly, we use ``ECM-style'' for this three-block procedure and assess convergence from the observed-data log-likelihood trace; the exact ECM monotonicity result~\citep{mengrubin1993ecm} would require every block to attain a conditional maximum. At the boundary cases, $p_k = 2$ recovers Gaussian (with $\sigma = \alpha/\sqrt{2}$) and $p_k = 1$ recovers Laplace (with $b = \alpha$).

\paragraph{Algorithm descriptions.}
The unified EM procedure shared across all four emission families is given in Algorithm~\ref{alg:chmm_em}: a shared framework (quantile-based initialisation, log-space forward-backward, transition update) plus four family-specific branches at lines~\ref{line:init}, \ref{line:emit}, and \ref{line:mstep}. The simulator is given in Algorithm~\ref{alg:chmm_simulate}, and the cross-asset rank-reordering simulator for the multi-asset construction in Algorithm~\ref{alg:copula_sim}. Viterbi and SIM-resampling pseudocode are provided in the companion code repository; both are textbook constructions and we omit them from the appendix.

% ========================================================================================= %
% Algorithm 1: Baum-Welch (EM) for Continuous Gaussian HMM
% ========================================================================================= %
% Unified EM algorithm with family-specific M-step switch
% ========================================================================================= %
\begin{breakablealgorithm}
\caption{Unified EM algorithm for the CHMM family. Lines~\ref{line:init}, \ref{line:emit}, and \ref{line:mstep} branch on the emission family $F \in \{\mathcal{N}, t, L, \text{GED}\}$; every other line is identical across families. The CHMM-N and CHMM-L variants instantiate classical EM; CHMM-t and CHMM-GED use ECM-style block updates whose shape step is a finite one-dimensional golden-section search. Hyperparameter defaults used throughout this paper: convergence tolerance $\epsilon = 10^{-4}$ on $\Delta \mathcal L$, $\texttt{max\_iter} = 60$, $40$-iteration golden-section sub-step on every bracketed update, CHMM-t bracket $(\nu_{\min}, \nu_{\max}) = (2.1, 50)$ with initial $\nu^{(0)} = 6$, CHMM-GED bracket $(p_{\min}, p_{\max}) = (0.5, 3.0)$ with initial $p^{(0)} = 1.5$ and a $\mu_k$ search interval widened to cover the observed data range.}
\label{alg:chmm_em}
\footnotesize
\setlength{\baselineskip}{0.95\baselineskip}
\begin{algorithmic}[1]
\REQUIRE Observations $O_{1:T}$, number of states $K$, tolerance $\epsilon$, $\texttt{max\_iter}$, emission family $F \in \{\mathcal{N}, t, L, \text{GED}\}$, if $F = t$ bracket $(\nu_{\min}, \nu_{\max})$ and initial $\nu^{(0)}$, if $F = \text{GED}$ bracket $(p_{\min}, p_{\max})$ and initial $p^{(0)}$.
\ENSURE Transition matrix $\mathbf{T}$, per-state emission parameters $\boldsymbol\theta_{1:K}$, initial distribution $\boldsymbol{\pi}$.

\STATE \textbf{Quantile-Based Initialization.}
\STATE Sort $O^{(\text{sort})} \leftarrow \text{sort}(O_{1:T})$ and partition into $K$ equal chunks $\{C_1,\ldots,C_K\}$.
\FOR{$k = 1$ to $K$} \label{line:init}
    \STATE \textbf{switch} $F$
    \STATE \quad $F = \mathcal{N}$: $\mu_k \leftarrow \text{mean}(C_k)$, \quad $\sigma_k \leftarrow \max(\text{std}(C_k), 10^{-6})$.
    \STATE \quad $F = t$: $\mu_k \leftarrow \text{mean}(C_k)$, \quad $\sigma_k \leftarrow \max(\text{std}(C_k), 10^{-6})$, \quad $\nu_k \leftarrow \nu^{(0)}$.
    \STATE \quad $F = L$: $\mu_k \leftarrow \text{median}(C_k)$, \quad $b_k \leftarrow \max(\text{mean}(|C_k - \mu_k|), 10^{-6})$. \COMMENT{Laplace scale $b_k$.}
    \STATE \quad $F = \text{GED}$: $\mu_k \leftarrow \text{mean}(C_k)$, \quad $\alpha_k \leftarrow \max(\text{std}(C_k), 10^{-6})$, \quad $p_k \leftarrow p^{(0)}$.
\ENDFOR
\STATE $T_{ij} \leftarrow 1/K$; \quad $\pi_k \leftarrow 1/K$; \quad $\mathcal{L}_{\text{prev}} \leftarrow -\infty$.

\STATE \textbf{EM Loop.}
\FOR{$n = 1$ to $\texttt{max\_iter}$}

    \STATE \textbf{E-Step, per-family emission log-likelihood.} \label{line:emit}
    \STATE \textbf{switch} $F$
    \STATE \quad $F = \mathcal{N}$: $\log B_{t,k} \leftarrow \log \mathcal{N}(O_t \mid \mu_k, \sigma_k^2)$.
    \STATE \quad $F = t$: $\log B_{t,k} \leftarrow \log t_{\nu_k}(O_t; \mu_k, \sigma_k)$.
    \STATE \quad $F = L$: $\log B_{t,k} \leftarrow -\log(2 b_k) - |O_t - \mu_k|/b_k$.
    \STATE \quad $F = \text{GED}$: $\log B_{t,k} \leftarrow \log p_k - \log(2 \alpha_k) - \log \Gamma(1/p_k) - (|O_t - \mu_k|/\alpha_k)^{p_k}$.

    \STATE \textbf{E-Step, shared log-space forward-backward.}
    \STATE $\log \alpha_1(k) \leftarrow \log \pi_k + \log B_{1,k}$.
    \FOR{$t = 2$ to $T$, $j = 1$ to $K$}
        \STATE $\log \alpha_t(j) \leftarrow \underset{i}{\text{logsumexp}}\!\left(\log \alpha_{t-1}(i) + \log T_{ij}\right) + \log B_{t,j}$.
    \ENDFOR
    \STATE $\log \beta_T(k) \leftarrow 0$.
    \FOR{$t = T-1$ down to $1$, $i = 1$ to $K$}
        \STATE $\log \beta_t(i) \leftarrow \underset{j}{\text{logsumexp}}\!\left(\log T_{ij} + \log B_{t+1,j} + \log \beta_{t+1}(j)\right)$.
    \ENDFOR
    \STATE Compute $\gamma_t(k)$ and $\xi_t(i,j)$ from equation~\eqref{eq:gamma_xi}.

    \STATE \textbf{Shared transition update.} $T_{ij} \leftarrow \sum_t \xi_t(i,j) \;/\; \sum_t \gamma_t(i)$; $\pi_k \leftarrow \gamma_1(k)$.

    \STATE \textbf{M-Step, per-family emission update.} \label{line:mstep}
    \STATE \textbf{switch} $F$
    \STATE \quad $F = \mathcal{N}$ \COMMENT{classical Baum-Welch~\citep{baum1970maximization}.}
    \STATE \qquad $\mu_k \leftarrow \sum_t \gamma_t(k)\, O_t / \sum_t \gamma_t(k)$.
    \STATE \qquad $\sigma_k \leftarrow \max\!\left(\sqrt{\sum_t \gamma_t(k) (O_t - \mu_k)^2 / \sum_t \gamma_t(k)},\; 10^{-6}\right)$.
    \STATE \quad $F = t$ \COMMENT{ECM/ECME of \citet{peel2000robust, liu1995ml}; closed-form $(\mu_k, \sigma_k)$ given $\nu_k$, then ECME 1D search on $\nu_k$.}
    \STATE \qquad $u_{t,k} \leftarrow (\nu_k + 1) / (\nu_k + ((O_t - \mu_k)/\sigma_k)^2)$ for all $t, k$.
    \STATE \qquad $\mu_k \leftarrow \sum_t \gamma_t(k)\, u_{t,k}\, O_t / \sum_t \gamma_t(k)\, u_{t,k}$.
    \STATE \qquad $\sigma_k \leftarrow \max\!\left(\sqrt{\sum_t \gamma_t(k)\, u_{t,k}\,(O_t - \mu_k)^2 / \sum_t \gamma_t(k)},\; 10^{-6}\right)$.
    \STATE \qquad $\nu_k \leftarrow \operatorname{GoldenSearchMax}_{\nu \in [\nu_{\min}, \nu_{\max}]} \sum_t \gamma_t(k)\, \log t_\nu(O_t; \mu_k, \sigma_k)$ \COMMENT{40 iterations.}
    \STATE \quad $F = L$ \COMMENT{closed-form weighted Laplace MLE.}
    \STATE \qquad $\mu_k \leftarrow \text{WeightedMedian}(O_{1:T},\, \gamma_{1:T}(k))$ \COMMENT{first order statistic with cumulative weight $\ge W_k/2$.}
    \STATE \qquad $b_k \leftarrow \max\!\left(\sum_t \gamma_t(k)\, |O_t - \mu_k| / \sum_t \gamma_t(k),\; 10^{-6}\right)$.
    \STATE \quad $F = \text{GED}$ \COMMENT{three-stage ECM-style update: bracketed $\mu_k$, closed-form $\alpha_k$, bracketed $p_k$.}
    \STATE \qquad $\mu_k \leftarrow \operatorname{GoldenSearchMin}_{\mu \in I_k} \sum_t \gamma_t(k)\, |O_t - \mu|^{p_k}$ \COMMENT{40 iterations on $I_k = [\min_t O_t, \max_t O_t]$.}
    \STATE \qquad $\alpha_k \leftarrow \max\!\Big(\big[(p_k / W_k) \sum_t \gamma_t(k)\, |O_t - \mu_k|^{p_k}\big]^{1/p_k},\; 10^{-6}\Big)$ where $W_k = \sum_t \gamma_t(k)$ \COMMENT{closed form, unique zero of $\partial Q_k / \partial \alpha$.}
    \STATE \qquad $p_k \leftarrow \operatorname{GoldenSearchMax}_{p \in [p_{\min}, p_{\max}]} \sum_t \gamma_t(k)\, \log \mathrm{GED}(O_t; \mu_k, \alpha_k, p)$ \COMMENT{40 iterations.}

    \STATE \textbf{Convergence check.} $\mathcal{L}^{(n)} \leftarrow \underset{k}{\text{logsumexp}}(\log \alpha_T(k))$ \COMMENT{from the current E-step forward pass, the pre-M-step iterate; lags the returned parameters by one M-step.}
    \IF{$|\mathcal{L}^{(n)} - \mathcal{L}_{\text{prev}}| < \epsilon$} \STATE \textbf{break} \ENDIF
    \STATE $\mathcal{L}_{\text{prev}} \leftarrow \mathcal{L}^{(n)}$.
\ENDFOR

\RETURN $\mathbf{T}, \boldsymbol\theta_{1:K}, \boldsymbol{\pi}$ \COMMENT{$\boldsymbol\theta_k = (\mu_k, \sigma_k)$ for $F = \mathcal{N}$; $(\mu_k, \sigma_k, \nu_k)$ for $F = t$; $(\mu_k, b_k)$ for $F = L$; $(\mu_k, \alpha_k, p_k)$ for $F = \text{GED}$.}
\end{algorithmic}
\end{breakablealgorithm}

% Viterbi pseudocode omitted (textbook construction; reference implementation in companion code repository).

% ========================================================================================= %
% Algorithm 3: CHMM Simulation
% ========================================================================================= %
\begin{algorithm}[tp]
\caption{CHMM simulation of synthetic excess growth rate paths (unified four-family framework). The emission draw at line~\ref{line:simemit} branches on the family $F \in \{\mathcal{N}, t, L, \text{GED}\}$ exactly as in Algorithm~\ref{alg:chmm_em}: $\mathcal{N}(\mu_k, \sigma_k^2)$ for CHMM-N, $t_{\nu_k}(\mu_k, \sigma_k)$ for CHMM-t, $\mathrm{Laplace}(\mu_k, b_k)$ for CHMM-L, $\mathrm{GED}(\mu_k, \alpha_k, p_k)$ for CHMM-GED; the latent chain and transition step are shared.}
\label{alg:chmm_simulate}
\small
\begin{algorithmic}[1]
\REQUIRE Trained CHMM $\mathcal{M} = (\mathbf{T}, \{\boldsymbol\theta_k\}, \boldsymbol{\pi}, F)$ with $\boldsymbol\theta_k = (\mu_k, \sigma_k)$ for $F = \mathcal{N}$; $(\mu_k, \sigma_k, \nu_k)$ for $F = t$; $(\mu_k, b_k)$ for $F = L$; $(\mu_k, \alpha_k, p_k)$ for $F = \mathrm{GED}$. Number of steps $M$, Number of paths $P$
\ENSURE Simulated growth-rate paths $\{\hat{G}^{(p)}_{1:M}\}_{p=1}^{P}$

\FOR{$p = 1$ to $P$}
    \STATE Sample initial state $S_1$ from the stationary distribution of $\hat{\mathbf T}$.
    \FOR{$t = 1$ to $M$}
        \STATE Emit observation $\hat{G}_t^{(p)} \sim f_{S_t}(\,\cdot\,;\, \boldsymbol\theta_{S_t})$ with $f_k$ the $F$-family per-state density. \label{line:simemit}
        \IF{$t < M$}
            \STATE Transition to next state $S_{t+1} \sim \text{Categorical}(\mathbf{T}_{S_t, \cdot})$.
        \ENDIF
    \ENDFOR
\ENDFOR

\RETURN $\{\hat{G}^{(p)}_{1:M}\}_{p=1}^{P}$
\end{algorithmic}
\end{algorithm}

% ========================================================================================= %
% Algorithm 4: Cross-Asset Copula Rank Reordering Simulation
% ========================================================================================= %
\begin{algorithm}[tp]
\caption{Cross-asset copula simulation via rank reordering}
\label{alg:copula_sim}
\small
\begin{algorithmic}[1]
\REQUIRE Fitted per-asset CHMMs $\{\mathcal{M}_j\}_{j=1}^d$, correlation matrix $\Sigma$ (positive semi-definite (PSD), unit diagonal), copula family $\mathcal{C} \in \{\text{Gaussian}, t_\nu\}$, path length $T$, number of paths $P$
\ENSURE Cross-asset path tensor $\{\hat{G}^{(p)}_{j,t}\}$ of shape $T \times d \times P$

\FOR{$p = 1$ to $P$}

    \STATE Draw copula sample $\mathbf{U}^{(p)} \in [0,1]^{T \times d}$ from $\mathcal{C}(\Sigma)$:
    \STATE \quad $L \leftarrow \text{Cholesky}(\Sigma)$, \; $Z \in \mathbb{R}^{T \times d} \sim \mathcal{N}(0, I_d)$ i.i.d.
    \STATE \quad $Y \leftarrow Z L^\top$
    \IF{$\mathcal{C} = $ Gaussian}
        \STATE \quad $\mathbf{U}^{(p)} \leftarrow \Phi(Y)$ componentwise
    \ELSE
        \STATE \quad Draw $W \sim \chi^2_\nu / \nu$ i.i.d. for $t = 1, \dots, T$; \; $X_{t,\cdot} \leftarrow Y_{t,\cdot} / \sqrt{W_t}$
        \STATE \quad $\mathbf{U}^{(p)} \leftarrow t_\nu(X)$ componentwise
    \ENDIF

    \FOR{$j = 1$ to $d$}
        \STATE Simulate $\tilde{g}_{j,1:T}^{(p)}$ from $\mathcal{M}_j$ via Algorithm~\ref{alg:chmm_simulate}.
        \STATE Compute order statistics $\tilde{g}_{j,(1)}^{(p)} \leq \cdots \leq \tilde{g}_{j,(T)}^{(p)}$.
        \STATE Compute ranks $r_{j,t}^{(p)} \leftarrow \text{rank}(U_{j,t}^{(p)})$ for $t = 1, \dots, T$.
        \STATE $\hat{G}_{j,t}^{(p)} \leftarrow \tilde{g}_{j,(r_{j,t}^{(p)})}^{(p)}$ for $t = 1, \dots, T$.
    \ENDFOR
\ENDFOR

\RETURN $\{\hat{G}^{(p)}_{j,t}\}$
\end{algorithmic}
\end{algorithm}

% SIM pseudocode omitted (standard rank-one factor regression with bootstrap residual resampling;
% reference implementation in companion code repository).

\paragraph{Per-state shape diagnostics: CHMM-t and CHMM-GED.}
This appendix backs the main-text discussion of the per-state shape allocation under CHMM-t and CHMM-GED on SPY IS at $K = 18$ (seed = $20260420$).

\paragraph{CHMM-t per-state $\nu_k$.}
The per-state $\nu_k$ histogram for the $K = 18$ CHMM-t fit is plotted in Figure~\ref{fig:nu_hist}.
Thirteen of the eighteen states rested at the upper ECM bracket ($\nu_k = 50$), and only two states lay near the lower bracket ($\nu_2 = 2.19$, $\nu_4 = 2.10$).
The simulated IS mixture kurtosis was driven by the posterior mass routed to those two very heavy-tailed states, not by a systemic collapse to the lower bracket.

\begin{figure}[!ht]
\centering
\includegraphics[width=0.65\textwidth]{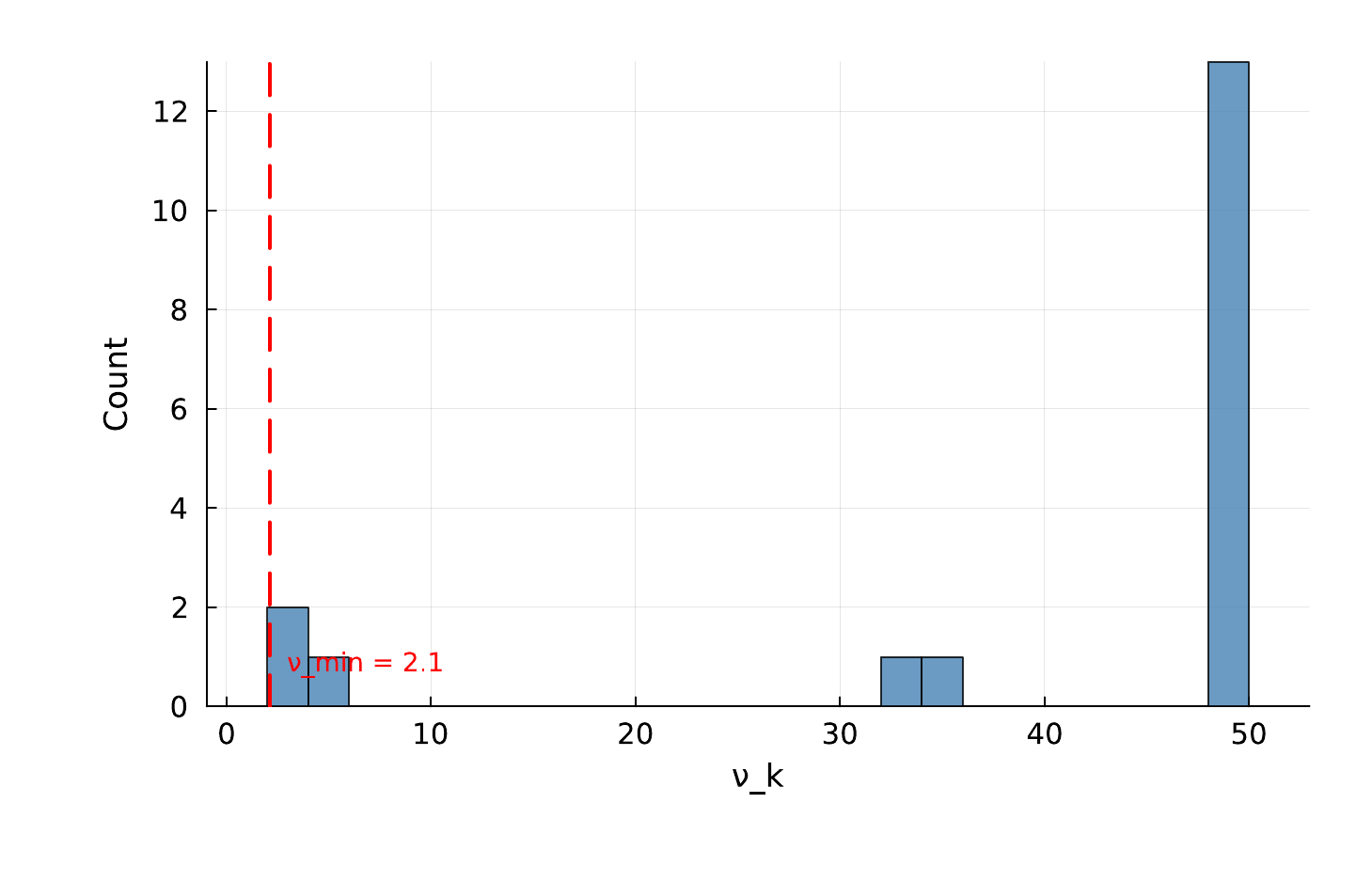}
\caption{Per-state degrees-of-freedom $\nu_k$ for the CHMM-t fit at $K = 18$ on SPY IS. The dashed red line marks the lower ECM bracket $\nu_{\min} = 2.1$. Only two states ($\nu_2 = 2.19$, $\nu_4 = 2.10$) sit near the lower bracket; the median $\nu_k$ is at the upper bracket.}
\label{fig:nu_hist}
\end{figure}

\paragraph{CHMM-GED per-state $p_k$ (Gaussian-bulk / Laplace-tail bimodality).}
The generator comparison reports and discusses the $\hat p_k$ partition (Figure~\ref{fig:p_hist}). Briefly, the $K = 18$ CHMM-GED fit on SPY IS was bimodal: eleven states attained exactly $p_k = 3.0$ (thirteen were Gaussian-like with $p_k \ge 1.85$), one state landed at intermediate $p_k = 1.6$, and four states clustered in the Laplace-shape regime $p_k \in [0.86, 1.24]$, concentrated in the high-volatility ranks. The smallest fitted $p_k$ on SPY was $0.86$, well clear of every $p_{\min}$ tested in the bracket sweep, so the lower bracket was non-binding. The bimodal partition replicated across $10$ Monte Carlo seeds and across the cross-asset universe (full robustness panel in the cross-asset robustness appendix).

\paragraph{Bracket-sensitivity sweeps.}
A CHMM-t bracket sweep on $\nu_{\min} \in \{2.1, 2.5, 3.0, 4.0, \infty\}$ at $K = 18$ showed that the kurtosis overshoot is a feature of the Student-$t$ emission at $K = 18$, not a lower-bracket artefact. The simulated IS kurtosis stayed in a heavy-tailed band across finite brackets and collapsed only in the full Gaussian limit. The analogous CHMM-GED sweep on $p_{\max} \in \{2.0, 2.5, 3.0, 3.5, 4.0\}$ left the non-Gaussian state count invariant, with Laplace-like states in $\{3,4,5\}$; tightening $p_{\max}$ from $4.0$ to $3.0$ cost only a few nats of LL. The $p_{\min}$ sweep on $\{0.5, 0.7, 0.85\}$ was fully degenerate: every floor produced an identical fit because the smallest observed $\hat p_k$ on SPY was $0.86$.

\paragraph{Evaluation summary.}
Observed growth rates $G_t$ are fit by Baum-Welch CHMM at the default $K^\star = 3$, with $K = 18$ retained as the sensitivity reference for kurtosis fidelity; $1{,}000$ paths are then simulated and per-asset metrics computed (Table~\ref{tab:model_comparison}). For the multi-asset analysis we reuse the per-asset $K^\star = 3$ fits as marginals, draw a copula sample $\mathbf U \sim \mathcal C(\Sigma)$ from the fitted Student-$t$ copula on ranks, rank-reorder the per-asset single-asset paths to match $\mathbf U$, and report cross-asset metrics (Table~\ref{tab:cross_asset}). The single-asset CHMM framework is shared across both; the CHMM training and simulation steps are Algorithms~\ref{alg:chmm_em} and~\ref{alg:chmm_simulate}, the cross-asset rank-reordering step is Algorithm~\ref{alg:copula_sim}, and the architecture diagram for the single-asset framework is Figure~\ref{fig:chmm_architecture}.

\subsection{Validation Metric Details}
\label{sec:supp_metrics}
This appendix documents the seven per-path metrics used in the evaluation protocol.
For each model, 1{,}000 independent paths of length $T$ are simulated (200 paths for the cross-asset extension) and each metric is computed per path before aggregation, except the ACF-MAE, which compares the observed ACF against the across-path mean simulated ACF.

\paragraph{Kolmogorov--Smirnov (KS) pass rate.}
The two-sample KS statistic~\cite{kolmogorov1933, smirnov1948} measures the maximum vertical distance between two empirical CDFs,
\begin{equation}
    D_{n,m} = \sup_x |F_n(x) - G_m(x)|,
    \label{eq:ks}
\end{equation}
where $F_n$ and $G_m$ are the empirical CDFs of the two samples and $n$ and $m$ their sizes.
The KS pass rate is the fraction of paths whose marginal is not rejected against the reference data at $\alpha = 0.05$.
We read it as a descriptive marginal-fidelity score, not a calibrated test: the asymptotic two-sample KS null assumes i.i.d.\ observations, whereas observed and simulated paths carry serial dependence, so the nominal $p$-values are not calibrated here.
The stationary block-bootstrap KS recalibration provides the temporally aware inferential check.

\paragraph{Anderson--Darling (AD) pass rate.}
The AD test~\citep{scholz1987ksample} is the integrated squared CDF deviation with a quadratic tail weighting, making it more sensitive than KS to tail discrepancies.
The AD pass rate is the fraction of paths for which the two-sample AD test fails to reject at $\alpha = 0.05$.

\paragraph{Excess kurtosis.}
Mean simulated excess kurtosis across paths, compared to the observed value.
Given Gaussian emissions, the CHMM is expected to underestimate kurtosis relative to the heavy-tailed empirical distribution, and the shortfall is one of the reported diagnostics.

\paragraph{Autocorrelation-function MAE (ACF-MAE).}
Mean absolute error between observed and mean-simulated autocorrelation functions of $|G_t|$ over $L = 252$ lags,
\begin{equation}
    \text{ACF-MAE} = \frac{1}{L}\sum_{\tau=1}^{L} \left| \hat\rho^{\text{obs}}_{|G|}(\tau) - \bar{\hat\rho}^{\,\text{sim}}_{|G|}(\tau) \right|.
    \label{eq:acf_mae_supp}
\end{equation}
This is the direct diagnostic for the \citet{ryden1998stylized} limitation.

\paragraph{Wasserstein-1 distance.}
The mean absolute difference of sorted empirical quantiles, equivalent to the $L^1$ distance between the empirical CDFs~\cite{glasserman2003monte},
\begin{equation}
    W_1(F, G) = \int_0^1 |F^{-1}(u) - G^{-1}(u)|\, du.
    \label{eq:wasserstein}
\end{equation}

\paragraph{Hellinger distance.}
A histogram-based overlap distance in $[0, 1]$,
\begin{equation}
    H(P, Q) = \frac{1}{\sqrt{2}}\sqrt{\sum_{i=1}^{B}\left(\sqrt{p_i} - \sqrt{q_i}\right)^2},
    \label{eq:hellinger}
\end{equation}
with $B$ equal-width bins spanning the joint support and $p_i, q_i$ the observed and simulated bin probabilities.

\paragraph{Quantile coverage.}
The fraction of 99 equally-spaced observed quantiles (1st through 99th percentile) falling inside the 5th--95th percentile envelope of the corresponding simulated quantile across 1{,}000 paths.
A coverage of 100\% indicates that the simulation envelope fully brackets the observed distribution at every percentile.

\paragraph{Cross-asset metrics.}
For the cross-asset extension we additionally track the mean Frobenius norm $\|\hat\Sigma_{\text{sim}}^{(p)} - \hat\Sigma_{\text{obs}}\|_F$ and the off-diagonal Pearson correlation MAE between simulated and observed growth-rate matrices, averaged over paths. The off-diagonal MAE is averaged over the $d(d-1)/2$ unique upper-triangular entries of the symmetric correlation matrix (so $15$ entries for the six-asset cross-section of the main text cross-asset analysis), not double-counted across both triangles:
\begin{equation*}
    \text{MAE}_{\text{off-diag}} \;=\; \frac{1}{P}\sum_{p=1}^{P} \frac{2}{d(d-1)} \sum_{1 \leq i < j \leq d} \left| \hat\Sigma^{(p)}_{\text{sim},\,ij} - \hat\Sigma_{\text{obs},\,ij} \right|.
\end{equation*}

\subsubsection{Continuous Ranked Probability Score and Diebold-Mariano}
\label{sec:crps_methods}

The OoS CRPS column of Table~\ref{tab:model_comparison} reports a proper scoring rule complementing the KS pass rate~\citep{gneiting2007strictly}. We use the unbiased sample CRPS estimator: for an ensemble of $N$ simulated paths $\{x_i\}_{i=1}^N$ at time $t$ and observed value $y_t$,
\begin{equation}
    \widehat{\text{CRPS}}_t = \frac{1}{N} \sum_{i=1}^N |x_i - y_t| - \frac{1}{N(N-1)} \sum_{1 \le i < j \le N} |x_i - x_j|.
\end{equation}
The double sum is computed in $O(N \log N)$ via the sorted-ensemble identity $\sum_{i<j} (x_{(j)} - x_{(i)}) = \sum_i x_{(i)} (2i - N - 1)$. The ensemble at each $t$ is the cross-section of the $N = 1{,}000$ unconditional simulated paths used throughout the main text; this yields a marginal-predictive CRPS suitable for an unconditional generative-fidelity assessment, distinct from the conditional one-step-ahead CRPS used in forecasting.

For pairwise comparison we report a Diebold-Mariano test~\citep{diebold1995comparing} on the per-$t$ loss differential $d_t = \widehat{\text{CRPS}}_t^A - \widehat{\text{CRPS}}_t^B$ with a Newey-West HAC variance under a Bartlett kernel and bandwidth $h = \lfloor T_{\text{OoS}}^{1/3} \rfloor = 8$. Two-sided $p$-values are reported under the standard-normal null. The main-text claim is that the CHMM tied or beat every benchmark on mean OoS CRPS, with Gaussian i.i.d.\ the only significantly worse comparator.

\subsubsection{Christoffersen Conditional-Coverage Panel}
\label{sec:christoffersen_supp}

The main-text VaR back-test reports only the unconditional Kupiec $\text{LR}_{\text{uc}}$ statistic~\citep{kupiec1995techniques}. This appendix reports the full Christoffersen~\cite{christoffersen1998evaluating} panel: $\text{LR}_{\text{uc}}$ (unconditional coverage), $\text{LR}_{\text{ind}}$ (breach independence), and $\text{LR}_{\text{cc}} = \text{LR}_{\text{uc}} + \text{LR}_{\text{ind}}$ (joint conditional coverage), at $\alpha \in \{0.01, 0.05\}$ on the SPY IS and OoS windows. The breach series is constructed from the pooled-archive VaR threshold (the $\alpha$-quantile of $\mathrm{vec}(\text{sim archive})$ across paths and time), then $\text{br}_t = (R_t \le \widehat{\text{VaR}}_\alpha)$.

\begin{table}[!ht]
\centering
\small
\caption{\textbf{Christoffersen conditional-coverage VaR back-test} (SPY, seed = $20260420$, $K = 18$, $1{,}000$ simulated paths). Critical values: $\chi^2_1(0.05) = 3.841$, $\chi^2_2(0.05) = 5.991$. ``br rate'' is the empirical breach rate; $\text{LR}_{\text{uc}}$ is the Kupiec statistic ($\sim\chi^2_1$ under correct unconditional coverage); $\text{LR}_{\text{ind}}$ is the Christoffersen independence statistic ($\sim\chi^2_1$ under independent breaches); $\text{LR}_{\text{cc}} = \text{LR}_{\text{uc}} + \text{LR}_{\text{ind}}$ is the joint conditional-coverage statistic ($\sim\chi^2_2$). PASS if statistic is below critical. Every unconditional generator (i.i.d.\ bootstrap, GARCH, and all four CHMM variants) rejects independence on OoS at $\alpha = 0.05$, consistent with breach clustering driven by volatility-clustering in $R_{\text{oos}}$ rather than a CHMM-specific mis-specification of conditional dynamics.}
\label{tab:christoffersen_var}
\begin{tabular}{l l c r r r r r r r r}
\toprule
& & & & & \multicolumn{2}{c}{Kupiec} & \multicolumn{2}{c}{Christ.\ ind} & \multicolumn{2}{c}{Christ.\ cc} \\
\cmidrule(lr){6-7}\cmidrule(lr){8-9}\cmidrule(lr){10-11}
Model & Win & $\alpha$ & breaches & br rate & $\text{LR}_{\text{uc}}$ & $p$ & $\text{LR}_{\text{ind}}$ & $p$ & $\text{LR}_{\text{cc}}$ & $p$ \\
\midrule
Bootstrap & IS  & $0.01$ & $26$  & $1.03\%$ & $0.03$ & $0.87$ & $15.28$ & $0.00$ & $15.31$ & $0.00$ \\
Bootstrap & IS  & $0.05$ & $126$ & $5.01\%$ & $0.00$ & $0.99$ & $56.46$ & $0.00$ & $56.46$ & $0.00$ \\
Bootstrap & OoS & $0.01$ & $5$   & $0.87\%$ & $0.10$ & $0.76$ & $13.18$ & $0.00$ & $13.28$ & $0.00$ \\
Bootstrap & OoS & $0.05$ & $19$  & $3.32\%$ & $3.83$ & $0.05$ & $5.26$  & $0.02$ & $9.08$  & $0.01$ \\
\midrule
GARCH(1,1) & IS  & $0.01$ & $31$  & $1.23\%$ & $1.28$ & $0.26$ & $12.45$ & $0.00$ & $13.72$ & $0.00$ \\
GARCH(1,1) & IS  & $0.05$ & $139$ & $5.52\%$ & $1.41$ & $0.24$ & $45.38$ & $0.00$ & $46.79$ & $0.00$ \\
GARCH(1,1) & OoS & $0.01$ & $6$   & $1.05\%$ & $0.01$ & $0.91$ & $20.87$ & $0.00$ & $20.89$ & $0.00$ \\
GARCH(1,1) & OoS & $0.05$ & $24$  & $4.20\%$ & $0.82$ & $0.37$ & $5.87$  & $0.02$ & $6.69$  & $0.04$ \\
\midrule
CHMM-N & IS  & $0.01$ & $20$  & $0.79\%$ & $1.15$ & $0.28$ & $12.80$ & $0.00$ & $13.95$ & $0.00$ \\
CHMM-N & IS  & $0.05$ & $123$ & $4.89\%$ & $0.07$ & $0.80$ & $47.33$ & $0.00$ & $47.40$ & $0.00$ \\
CHMM-N & OoS & $0.01$ & $3$   & $0.52\%$ & $1.58$ & $0.21$ & $18.98$ & $0.00$ & $20.56$ & $0.00$ \\
CHMM-N & OoS & $0.05$ & $19$  & $3.32\%$ & $3.83$ & $0.05$ & $5.26$  & $0.02$ & $9.08$  & $0.01$ \\
\midrule
CHMM-t & IS  & $0.01$ & $22$  & $0.87\%$ & $0.42$ & $0.52$ & $11.62$ & $0.00$ & $12.04$ & $0.00$ \\
CHMM-t & IS  & $0.05$ & $127$ & $5.05\%$ & $0.01$ & $0.91$ & $55.54$ & $0.00$ & $55.55$ & $0.00$ \\
CHMM-t & OoS & $0.01$ & $4$   & $0.70\%$ & $0.58$ & $0.45$ & $15.53$ & $0.00$ & $16.12$ & $0.00$ \\
CHMM-t & OoS & $0.05$ & $20$  & $3.50\%$ & $3.03$ & $0.08$ & $4.71$  & $0.03$ & $7.74$  & $0.02$ \\
\midrule
CHMM-L & IS  & $0.01$ & $19$  & $0.76\%$ & $1.66$ & $0.20$ & $13.44$ & $0.00$ & $15.11$ & $0.00$ \\
CHMM-L & IS  & $0.05$ & $134$ & $5.33\%$ & $0.55$ & $0.46$ & $49.43$ & $0.00$ & $49.98$ & $0.00$ \\
CHMM-L & OoS & $0.01$ & $2$   & $0.35\%$ & $3.26$ & $0.07$ & $9.15$  & $0.00$ & $12.41$ & $0.00$ \\
CHMM-L & OoS & $0.05$ & $20$  & $3.50\%$ & $3.03$ & $0.08$ & $4.71$  & $0.03$ & $7.74$  & $0.02$ \\
\midrule
CHMM-GED & IS  & $0.01$ & $20$  & $0.79\%$ & $1.15$ & $0.28$ & $12.80$ & $0.00$ & $13.95$ & $0.00$ \\
CHMM-GED & IS  & $0.05$ & $126$ & $5.01\%$ & $0.00$ & $0.99$ & $56.46$ & $0.00$ & $56.46$ & $0.00$ \\
CHMM-GED & OoS & $0.01$ & $2$   & $0.35\%$ & $3.26$ & $0.07$ & $9.15$  & $0.00$ & $12.41$ & $0.00$ \\
CHMM-GED & OoS & $0.05$ & $19$  & $3.32\%$ & $3.83$ & $0.05$ & $5.26$  & $0.02$ & $9.08$  & $0.01$ \\
\bottomrule
\end{tabular}
\end{table}

Two patterns stand out (Table~\ref{tab:christoffersen_var}). On the unconditional (Kupiec) leg every CHMM variant passed at both levels on both windows; the $\text{LR}_{\text{uc}}$ statistic for the four CHMM rows at $\alpha = 0.05$ on OoS clustered at $3.03$--$3.83$ (just below the $\chi^2_1$ critical value $3.841$), reflecting the integer-breach grid at $T_{\text{OoS}} = 572$ noted in the main text table caption. On the independence leg every unconditional generator in the panel (bootstrap, GARCH, and all four CHMM variants) rejected: $\text{LR}_{\text{ind}}$ ranged over $4.71$--$56.46$ across rows. This is structural rather than CHMM-specific. Volatility-clustering in $R_{\text{oos}}$ produces persistent regimes of high and low realised volatility. The OoS-window breaches concentrate inside the high-volatility regime, and a constant-across-$t$ VaR threshold (the construction every unconditional generator uses) cannot track that timing. The natural fix is a regime-conditional VaR that propagates the CHMM latent-state forecast $\Prob(s_{t+1}\,|\,\text{history}_t)$ through the per-state emission-mixture; this is operationally the conditional-VaR companion-paper direction noted in the main text Discussion.

\subsubsection{Value-at-Risk and Expected-Shortfall Envelope Visualization (Companion to Table~\ref{tab:var_es})}
\label{sec:var_es_supp}

The median and $[5\%, 95\%]$ envelopes across simulated paths are visualised in Figure~\ref{fig:var_es}; the numeric values are reported in Table~\ref{tab:var_es}.

\begin{figure}[H]
\centering
\begin{subfigure}[b]{0.49\textwidth}\centering
\includegraphics[width=\textwidth]{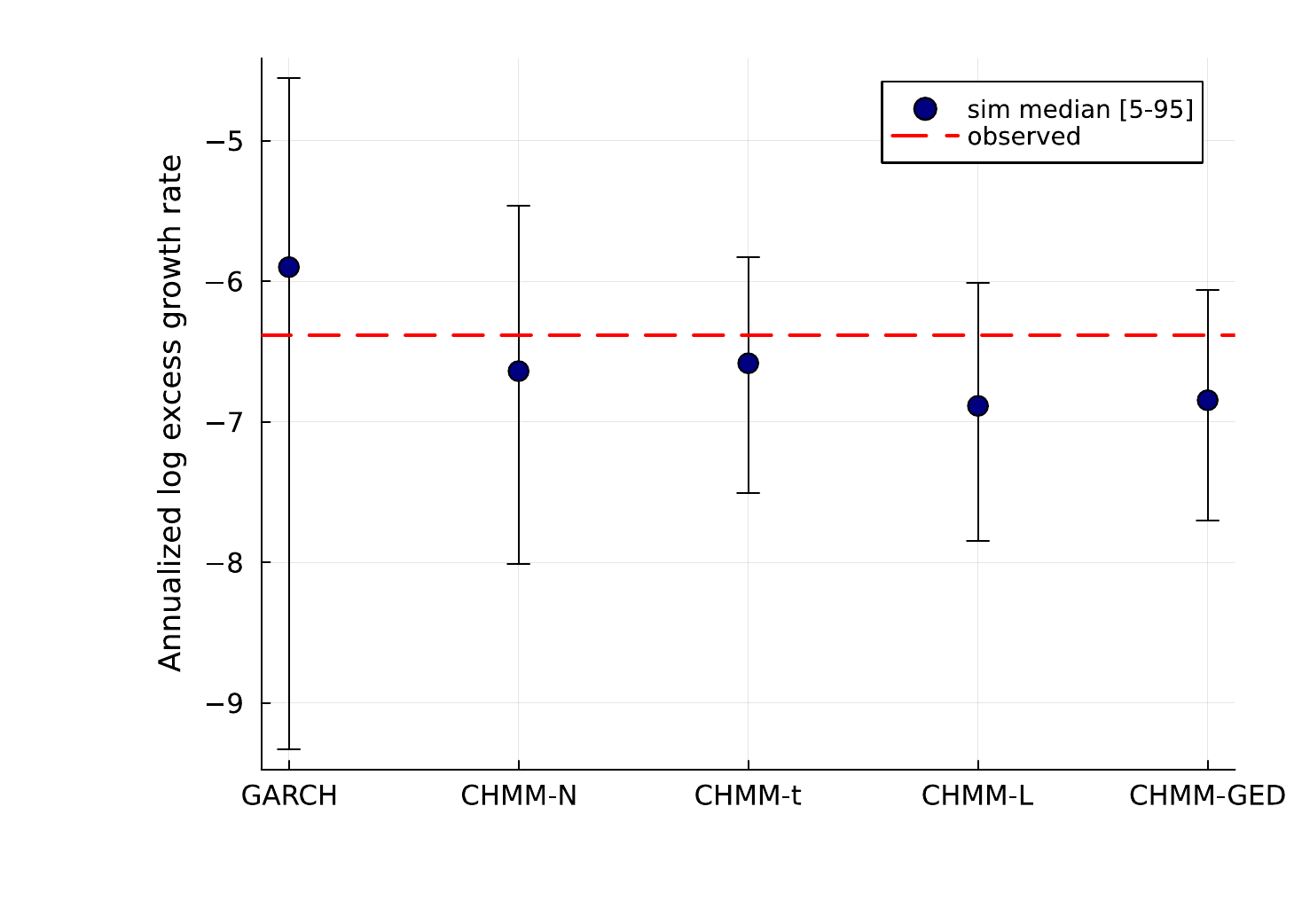}
\caption{IS VaR ($\alpha = 0.01$).}
\end{subfigure}\hfill
\begin{subfigure}[b]{0.49\textwidth}\centering
\includegraphics[width=\textwidth]{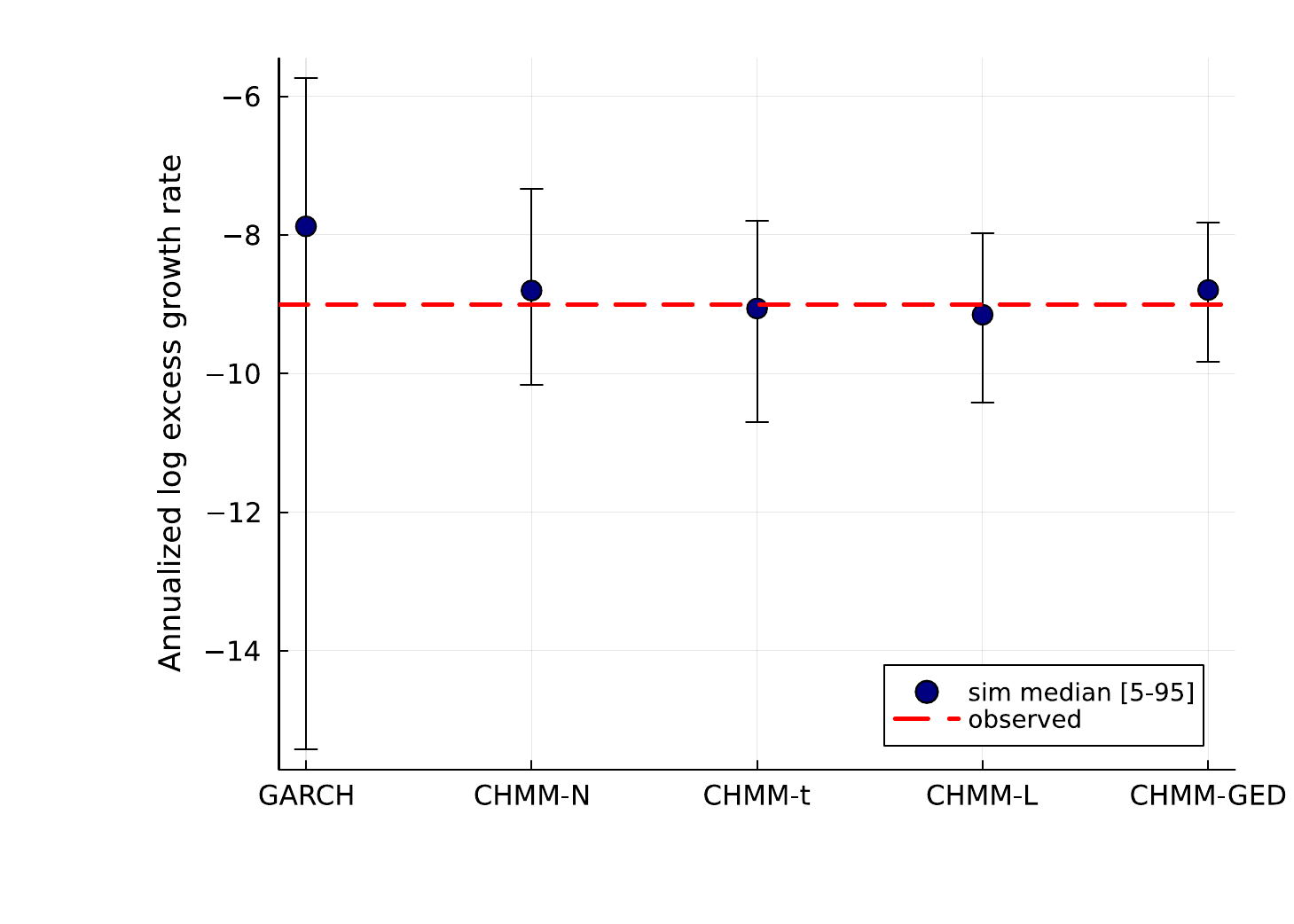}
\caption{IS ES ($\alpha = 0.01$).}
\end{subfigure}\\[0.5em]
\begin{subfigure}[b]{0.49\textwidth}\centering
\includegraphics[width=\textwidth]{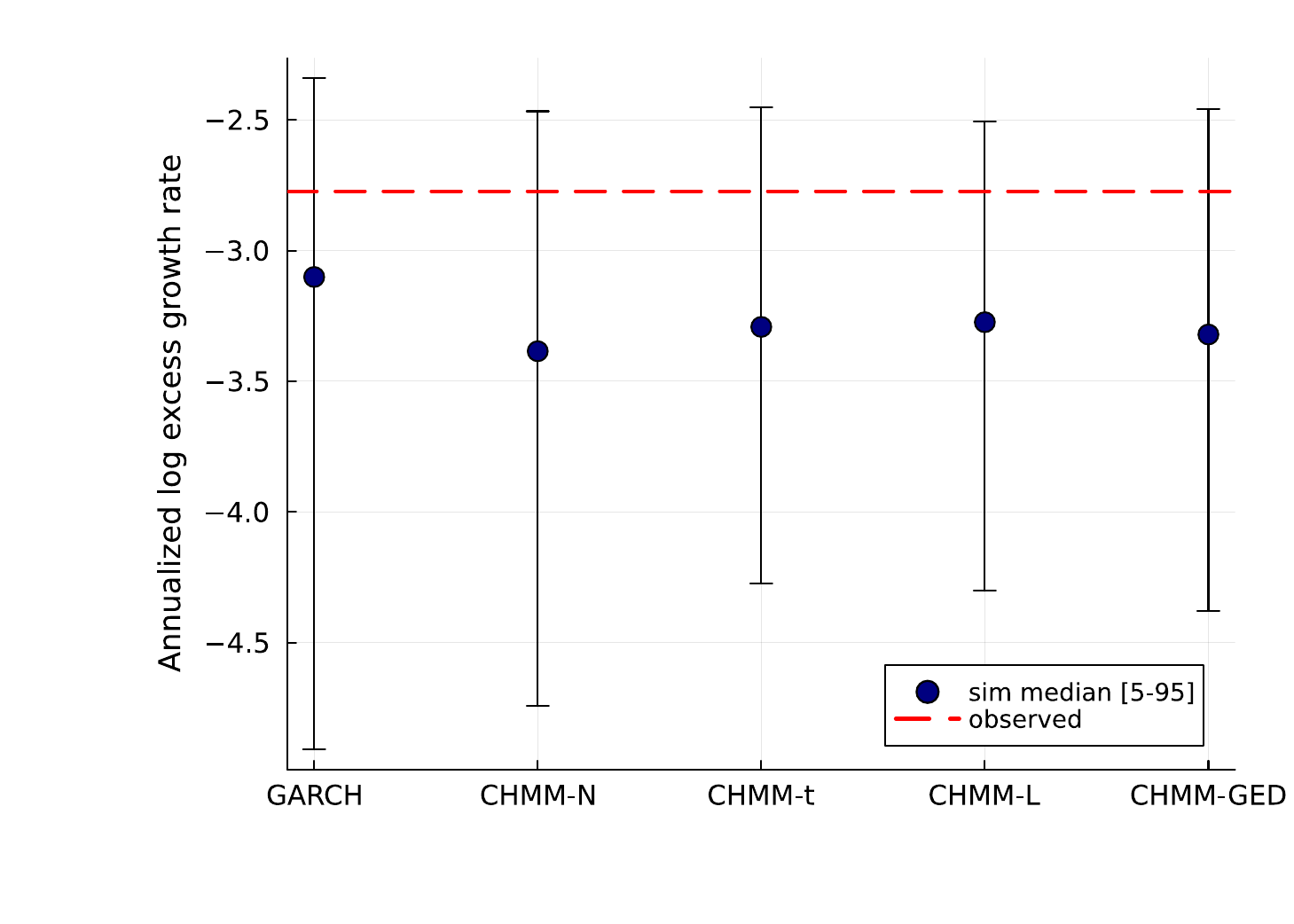}
\caption{OoS VaR ($\alpha = 0.05$).}
\end{subfigure}\hfill
\begin{subfigure}[b]{0.49\textwidth}\centering
\includegraphics[width=\textwidth]{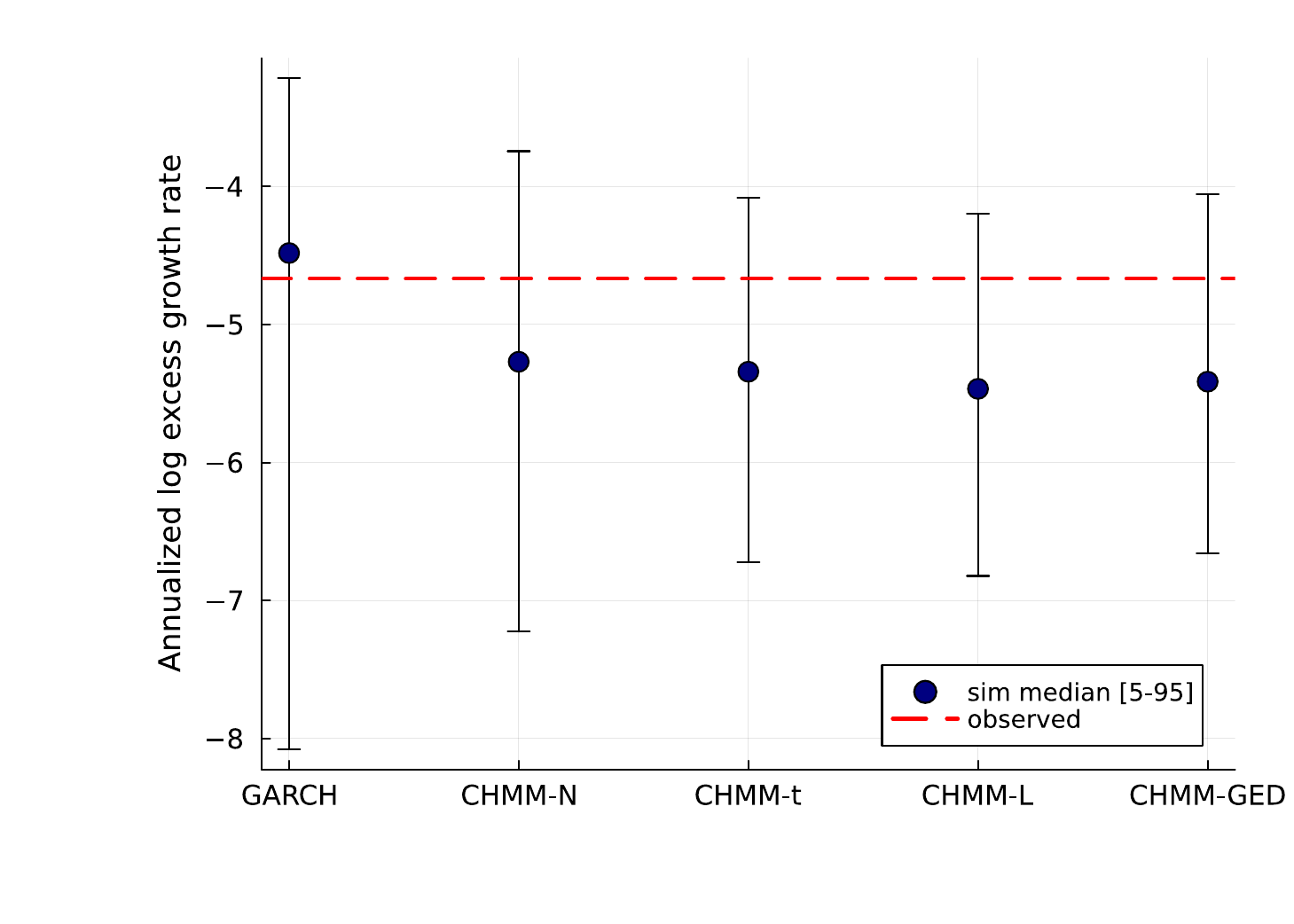}
\caption{OoS ES ($\alpha = 0.05$).}
\end{subfigure}
\caption{\textbf{VaR / ES envelope diagnostic} (SPY, \NPaths{1{,}000}; global seed policy under the evaluation protocol). Each panel shows the simulated median (dot) and $[5\%, 95\%]$ envelope (error bar) across paths against the observed historical value (dashed red line); panels (a)/(b) are the IS window at $\alpha = 0.01$, panels (c)/(d) are the OoS window at $\alpha = 0.05$, and the x-axis in each panel lists the five generators kept for the conditional-coverage comparison among volatility-aware models (GARCH(1,1), CHMM-N, CHMM-t, CHMM-L, CHMM-GED). The i.i.d.\ bootstrap is omitted here because its envelope is by construction matched to the empirical marginal and does not discriminate among volatility-aware models. The y-axis is the annualized log excess growth rate (daily log return scaled by $1/\Delta t$ with $\Delta t = 1/252$), consistent with the $R_{\text{IS}} / R_{\text{OoS}}$ convention used throughout the paper. Every observed red line falls inside the corresponding simulated envelope, confirming that the CHMM family passes the coverage leg of the back-test complementary to the Kupiec likelihood-ratio leg in the main text VaR back-test.}
\label{fig:var_es}
\end{figure}

\subsection{Formal Theoretical Statements}
\label{sec:supp_propositions}

The following assumptions support the spectral identity. We also state the scope of the numerical-optimisation, identifiability, and consistency claims so that those issues are separated from the empirical fitting procedure.

\begin{assumption}[Irreducibility and aperiodicity]
\label{ass:irred}
The transition matrix $\mathbf T \in \R^{K \times K}$ is irreducible and aperiodic on the state space $\{1, \ldots, K\}$, so the chain admits a unique stationary distribution $\bar{\boldsymbol\pi}$ satisfying $\bar{\boldsymbol\pi}\,\mathbf T = \bar{\boldsymbol\pi}$ and $\lambda_1 = 1$ is a simple eigenvalue with all other eigenvalues strictly inside the unit disk~\citep{levinperes2017markov}.
\end{assumption}

\begin{assumption}[Finite second moment and emission contrast]
\label{ass:moments}
Each per-state emission density $f_k(\cdot;\boldsymbol\theta_k)$ has finite second moment, $\E_{f_k}[G^2] < \infty$. For a non-degenerate absolute growth-rate ACF we further require the state-conditional absolute growth-rate means $m_k = \E[|G_t| \mid s_t = k]$ to differ across at least two states; this rules out only the trivial flat-ACF case and is not needed for the ACF identity itself.
\end{assumption}

\paragraph{Numerical optimisation scope.}
The CHMM-N and CHMM-L M-steps are exact conditional maximisations. For CHMM-t and CHMM-GED, exact maximisation of every numerical block would give the standard ECM/ECME ascent property~\cite{wu1983convergence, mengrubin1993ecm}. The implementation instead uses fixed-iteration bounded searches for $\nu_k$, $\mu_k$, and $p_k$. These are approximate block updates, and the GED location objective can be nonconvex when $p_k < 1$. We therefore do not claim exact likelihood monotonicity for every finite-precision iteration. Numerical convergence is assessed from the observed-data log-likelihood trace using the tolerance and iteration cap stated in Algorithm~\ref{alg:chmm_em}. Even when an ascent condition is enforced, bounded monotone likelihood values imply convergence of the likelihood sequence, not necessarily convergence of the full parameter sequence to a unique maximiser.

\paragraph{Numerical full-rank check on $\hat{\mathbf T}$.}
Full rank of $\mathbf T$ is a separate sufficient condition in the HMM identifiability result of \citet{allman2009identifiability}; it does not follow from irreducibility. As an empirical diagnostic, the fitted $\hat{\mathbf T}$ at $K = 18$ across the four emission families on SPY IS was full rank in the numerical sense (smallest singular value $\sigma_{\min} \in [0.0007, 0.0170]$, condition numbers $63$ to $1{,}620$), and the deflated matrix $\hat{\mathbf T} - \mathbf 1\bar{\boldsymbol\pi}^\top$ had numerical rank exactly $K - 1 = 17$ across all four families. This fitted-matrix check is not by itself a proof of population identifiability.

\begin{proposition}[Conditional HMM identifiability]
Suppose that $\mathbf T$ has full rank and the $K$ per-state emission densities are linearly independent. Under the remaining regularity conditions of \citet{allman2009identifiability}, the HMM parameters are identifiable from a sufficiently long observation block up to permutation of the latent-state labels.
\end{proposition}
\noindent\emph{Scope.} Distinct Gaussian or Laplace location-scale components satisfy standard finite-mixture identifiability conditions~\cite{yakowitzspragins1968}, but this paper does not prove the required linear independence for the full varying-shape Student-$t$ and GED families. In particular, the Student-$t$ family has no moment-generating function at finite degrees of freedom, so an MGF argument is unavailable. The proposition is therefore a conditional statement, not a claim that pairwise-distinct fitted parameters alone establish identifiability for every emission family.

\paragraph{Consistency scope.}
Standard HMM maximum-likelihood consistency results require more than Assumptions~\ref{ass:irred}--\ref{ass:moments}: typically identifiability of the population model, continuity and integrability of the log-density, a compact or otherwise controlled parameter space, and a uniform mixing or forgetting condition~\cite{bickel1998asymptotic,doucmoulinesryden2004}. We do not assert that the two displayed assumptions alone verify all of those conditions for every varying-shape emission family. The empirical fitting and validation procedures in this paper do not rely on an asymptotic consistency theorem.

\begin{proposition}[Marginal preservation under rank reordering]
\label{prop:rank_marginals}
Let $\tilde g_{j, 1:T}$ be a sample of size $T$ from the per-asset CHMM marginal of asset $j$, and for simulated path $p$ let the rank-reordered output be $\hat g_{j,t}^{(p)} = \tilde g_{j, (r_{j,t}^{(p)})}^{(p)}$, where $r_{j,t}^{(p)}$ is the rank of the path-$p$ copula sample in column $j$ at time $t$. Then for every asset $j$ the empirical CDF of $\{\hat g_{j,t}^{(p)}\}_{t=1}^T$ equals that of $\{\tilde g_{j,t}^{(p)}\}_{t=1}^T$.
\end{proposition}
\noindent\emph{Proof sketch.} The map $t \mapsto r_{j,t}^{(p)}$ is a permutation almost surely, so the multiset of values is preserved; the empirical CDF depends only on the multiset.

\subsection{Spectral ACF Identity: Extended Derivation}
\label{sec:supp_spectral_acf}

The closed-form mixture-of-eigenvalues form for the absolute growth-rate ACF, given in Eq.~\eqref{eq:acf_normalised}, is used throughout the main text. The bilinear cross-moment identity $\E[|G_t|\,|G_{t+\tau}|] = \mathbf m^\top \mathrm{diag}(\bar{\boldsymbol\pi})\,\mathbf T^\tau\,\mathbf m$ and its eigendecomposition over the non-unit eigenvalues of $\mathbf T$ are textbook material in the regime-switching literature (\citealp{hamilton1994time,krolzig1997markov}; the moment-generating-function form for general Markov-switching specifications with conditional means and variances depending on the latent state is given by \citealp{timmermann2000moments}); the underlying Markov-chain spectral decomposition is standard~\cite{levinperes2017markov}. We restate the absolute growth-rate specialisation here as a self-contained theorem with a step-by-step proof, and record the lag-zero behaviour and the complex and non-diagonalisable extensions that the main text does not. Our contribution is the explicit application of the spectral form to recast the empirical \citet{ryden1998stylized} low-$K$ failure as a rank statement on $\mathbf T - \mathbf 1\bar{\boldsymbol\pi}^\top$, and the empirical demonstration that the rank constraint was not what limited the fit at moderate $K$ on equity-return data.

\begin{theorem}[Mixture-of-eigenvalues identity for the absolute growth-rate ACF]
\label{thm:supp_spectral_acf}
Let $\{(s_t, G_t)\}_{t \in \mathbb{Z}}$ be a stationary CHMM on $\{1, \ldots, K\}$ with irreducible aperiodic transition matrix $\mathbf T$ (Assumption~\ref{ass:irred}) and per-state emissions of finite second moment (Assumption~\ref{ass:moments}). Write $\bar{\boldsymbol\pi}$ for the unique stationary distribution, $\mathbf D = \mathrm{diag}(\bar{\boldsymbol\pi})$, $m_k = \E[|G_t| \mid s_t = k]$, $\mathbf m = (m_1, \ldots, m_K)^\top$, $\mu = \E|G_t| = \bar{\boldsymbol\pi}^\top \mathbf m$, and $\sigma_{|G|}^2 = \mathrm{Var}(|G_t|)$. Assume $\mathbf T$ is diagonalisable with right eigenvectors $\mathbf v_k$ and left eigenvectors $\mathbf w_k$ biorthonormal, $\mathbf w_k^\top \mathbf v_l = \delta_{kl}$, indexed so that $\lambda_1 = 1 > |\lambda_2| \ge \cdots \ge |\lambda_K|$. Then for every integer lag $\tau \ge 1$,
\begin{equation*}
    \rho_{|G|}(\tau)
    \;=\; \sum_{k=2}^{K} a_k\, \lambda_k^{\tau},
    \qquad
    a_k \;=\; \frac{(\mathbf m^\top \mathbf D\, \mathbf v_k)\,(\mathbf w_k^\top \mathbf m)}{\sigma_{|G|}^2}.
\end{equation*}
\end{theorem}

\begin{proof}
By stationarity, $\rho_{|G|}(\tau) = (\E[|G_t|\,|G_{t+\tau}|] - \mu^2) / \sigma_{|G|}^2$, so it suffices to evaluate the cross-moment for $\tau \ge 1$. Conditioning on the latent state pair $(s_t, s_{t+\tau})$ and using HMM conditional independence ($G_t \perp G_{t+\tau} \mid (s_t, s_{t+\tau})$ for $\tau \ge 1$) plus the Markov factorisation $\Prob(s_t = i, s_{t+\tau} = j) = \bar\pi_i (\mathbf T^\tau)_{ij}$ gives the bilinear form
\begin{equation}
    \E[|G_t|\,|G_{t+\tau}|] \;=\; \mathbf m^\top \mathbf D\, \mathbf T^\tau\, \mathbf m, \qquad \tau \ge 1.
    \label{eq:supp_bilinear}
\end{equation}
Perron--Frobenius gives $\lambda_1 = 1$ simple with $\mathbf v_1 = \mathbf 1$, $\mathbf w_1 = \bar{\boldsymbol\pi}$~\citep{levinperes2017markov}; the dyadic expansion $\mathbf T^\tau = \mathbf 1\bar{\boldsymbol\pi}^\top + \sum_{k \ge 2} \lambda_k^\tau \mathbf v_k \mathbf w_k^\top$ separates the dominant projector. The $k = 1$ term contributes $\mu^2$ and cancels in the autocovariance, leaving $\mathrm{Cov}(|G_t|, |G_{t+\tau}|) = \sum_{k \ge 2} c_k \lambda_k^\tau$ with $c_k = (\mathbf m^\top \mathbf D \mathbf v_k)(\mathbf w_k^\top \mathbf m)$; dividing by $\sigma_{|G|}^2$ gives Equation~\eqref{eq:acf_normalised}.
\end{proof}

\paragraph{Lag zero and within-state variance.}
The identity~\eqref{eq:supp_bilinear} fails at $\tau = 0$ since $\E[G_t^2 \mid s_t = i] \ne m_i^2$. Eigenvalue completeness gives $\sum_{k \ge 2} c_k = \mathrm{Var}(m_{s_t})$, the between-state variance, so $\sum_{k \ge 2} a_k = \mathrm{Var}(m_{s_t}) / \sigma_{|G|}^2 \in [0, 1]$. The normalized within-state variance $\E[\mathrm{Var}(|G_t| \mid s_t)] / \sigma_{|G|}^2 = 1 - \sum_{k \ge 2} a_k$ explains the difference between $\rho_{|G|}(0)=1$ and the spectral expression evaluated at lag zero. The further change to lag one, $\rho_{|G|}(1)=\sum_{k \ge 2}a_k\lambda_k$, also reflects transition decay through the factors $\lambda_k$.

\paragraph{Squared growth-rate, complex-eigenvalue, and Jordan-block extensions.}
If every emission has a finite fourth moment, replacing $|G_t|$ with $G_t^2$ and $\mathbf m$ with $\mathbf M = (\E[G_t^2 \mid s_t = i])_i$ leaves the proof termwise valid and gives $\rho_{G^2}(\tau) = \sum_{k \ge 2} \tilde a_k \lambda_k^\tau$ on the same non-unit spectrum. This additional condition is not automatic for CHMM-t when some $\nu_k \le 4$. Complex-conjugate eigenvalue pairs $(re^{\pm i\theta}, \cdot)$ contribute real damped oscillations $A r^\tau \cos(\theta\tau) + B r^\tau \sin(\theta\tau)$. For non-diagonalisable $\mathbf T$, a Jordan expansion instead adds polynomial-times-geometric prefactors, with exponential rate controlled by $\max_{k \ge 2} |\lambda_k|$; Theorem~\ref{thm:supp_spectral_acf}, which assumes diagonalizability, does not cover that case as stated. The dominant timescale used in the main text spectral mechanism corresponded to a real positive $\lambda_2$ on the fits of this paper.

\subsection{Vendor-Stitch Sanity Check}
\label{sec:vendor_stitch}

The OoS series spans January~4, 2024 through April~20, 2026 ($T_{\text{OoS}} = 572$) and is sourced from two vendors: Polygon.io via Massive (early portion) and Alpaca Markets / IEX (held-out extension). To verify that the stitch introduces no boundary artefact we compared the two vendors on every shared trading day. Of the $323$ overlap days, the per-day VWAP differential, the per-day annualised log-return differential, and a rolling 30-day kurtosis differential were all identically zero (mean, standard deviation, and maximum absolute value all $0.000$ to floating-point precision); the two-sample Kolmogorov-Smirnov test on the overlap returns gave $D = 0.000$ and $p = 1.00$. The two vendors were exact-match on the overlap, so the stitched OoS series is single-source-equivalent and no boundary artefact appeared at the November 19, 2025 stitch date. The diagnostic is reproducible from \texttt{runners/diagnostics/run\_vendor\_stitch\_check.jl}.

\subsection{Variant Decision Guide}
\label{sec:variant_decision_guide}

This appendix recommends a CHMM emission family by use-case priority (Table~\ref{tab:variant_choice}). All four variants share the forward-backward framework, the transition update, and quantile-based initialisation; they differ in the emission density and its M-step update. The shared-$\nu$ Student-$t$ row is the main heavy-tail recommendation: per-state heavy tails without a penalty hyperparameter at the highest IS KS pass rate of any CHMM variant in Table~\ref{tab:model_comparison}. The penalised CHMM-t at $\lambda = 20$ is retained as a sensitivity reference; the $\lambda$ value was tuned at $K = 18$ and is reported at $K^\star = 3$ as an upper bound on the unpenalised heavy tail.

\begin{table}[!ht]
\centering
\small
\caption{\textbf{Variant decision guide at $K^\star = 3$.}}
\label{tab:variant_choice}
\begin{tabular}{l l}
\toprule
Use-case priority & Recommended variant \\
\midrule
Simplest closed-form fit; kurtosis fidelity secondary              & CHMM-N \\
Closest IS kurtosis match without a shape parameter                & CHMM-L \\
Per-state heavy tails, no penalty hyperparameter \emph{(main)} & CHMM-t shared-$\nu$ \\
Per-state heavy tails, IS-calibrated $\lambda$-shrinkage           & CHMM-t pen.\ ($\lambda = 20$) \\
Adaptive per-state shape; data-chosen Gaussian-bulk / Laplace-tail mixture & CHMM-GED \\
\bottomrule
\end{tabular}
\end{table}

\subsection{State-Resolution and Multi-Emission Sensitivity}
\label{sec:supp_sensitivity}

This appendix reports the full $K$-sweep underlying the main state count, the multi-emission family extension confirming that $K = 18$ is stable under Student-t and Laplace emissions, the EM convergence curves at representative $K$, the IS and OoS visual comparison panels at non-selected $K \in \{3, 6, 12, 21\}$, the multi-emission figure panels at $K = 18$, and the Ryd\'{e}n $K = 2$ replication. The cross-ticker generalisation distribution at the main state count is summarised in Table~\ref{tab:cross_ticker}.

\begin{table}[!ht]
\centering
\small
\caption{\textbf{Cross-ticker generalisation, sector-balanced 30-ticker panel} (penalised CHMM-t at $\lambda = 20$, $1{,}000$ paths, seed $= 20260420$, main-paper $K^\star = 3$). The OoS distribution is wider than the IS distribution and is concentrated at the tickers that introduced a new regime (LLY, UNH, NEM). The full per-ticker panel, per-sector rollup, $K = 6$ and $K = 18$ sensitivity rebuilds, and quarterly-refit approach are reported in this appendix.}
\label{tab:cross_ticker}
\begin{tabular}{l c}
\toprule
Metric & $K^\star = 3$ \\
\midrule
IS KS pass rate (\%) median           & $96.8$ \\
IS KS pass rate (\%) mean $\pm$ s.d.  & $95.1 \pm 4.9$ \\
OoS KS pass rate (\%) median          & $69.1$ \\
OoS KS pass rate (\%) mean $\pm$ s.d. & $66.2 \pm 28.2$ \\
$|G_t|$ ACF-MAE median                & $0.0399$ \\
Kurtosis residual (sim $-$ obs) median & $7.18$ \\
Tickers OoS KS $< 60\%$ (IS-fixed)    & $11 / 30$ \\
\bottomrule
\end{tabular}
\end{table}

\subsubsection{Full State-Resolution Sensitivity Table (companion artefact)}
\label{sec:sensitivity_table}

The full $K$-sweep table for SPY CHMM-N referenced in the main text state-count selection (KS, AD, kurtosis, $W_1$, Hellinger, OoS coverage at $K \in \{3, 6, 9, 12, 15, 18, 21\}$, $1{,}000$ paths, seed $20260420$) is provided as a companion artefact in the code repository. The qualitative pattern is that distributional pass rates plateau in the $K \in \{12, 15, 18, 21\}$ band and ACF-MAE is essentially flat across the entire sweep, consistent with the single-dominant-mode spectral mechanism.

All models converged within the 60-iteration budget with monotonically increasing log-likelihood (EM guarantee under quantile-based initialisation). Lower $K$ converged in 15--25 iterations, $K = 18$ in 25--40, with a final $|\Delta \mathcal L| < 10^{-4}$ at every state count. Convergence-trace figures and per-iteration log-likelihood histories are provided in the companion code repository.

\subsubsection{Sector-Balanced 30-Ticker Panel: Full Per-Ticker Rollup}
\label{sec:cross_ticker_sector_panel}

The aggregate distribution and per-sector medians for the sector-balanced 30-ticker panel at the main text state count $K^\star = 3$ are reported in Table~\ref{tab:cross_ticker}.
The full per-ticker rollup is given at $K = 18$, the sensitivity reference for kurtosis fidelity, rather than at $K^\star = 3$ (Table~\ref{tab:sector_panel}). The per-ticker pattern was essentially $K$-robust across $\{3, 6, 18\}$, with identical $11/30$ failure counts and an overlapping ticker list (Table~\ref{tab:cross_ticker}), so the rollup at $K = 18$ is representative of the $K^\star = 3$ panel up to small per-ticker shifts. Each ticker is fit independently under the penalised CHMM-t framework ($K = 18$, $\lambda = 20$, $1{,}000$ simulated paths, seed $20260420$) on its own IS slice covering 2014-01-03 to 2024-01-03 and validated on the same 2024-01-04 to 2026-04-20 OoS window used for SPY in Table~\ref{tab:model_comparison}.
The per-ticker $\nu_k$ medians sat at the upper bracket ($50$) on every ticker, indicating that the $1/\nu_k$ shrinkage was operative across the panel; minimum and maximum $\nu_k$ values are in the per-ticker CSV in the companion code repository.

\begin{table}[!ht]
\centering
\scriptsize
\caption{\textbf{Sector-balanced 30-ticker panel, per-ticker rollup} (penalised CHMM-t at $K = 18$, $\lambda = 20$). Sectors ordered by GICS classification; tickers within sector ordered by the panel construction (top three large-cap representatives at IS-window-median market cap). Kurt resid is simulated minus observed IS excess kurtosis. Aggregate distribution and per-sector medians are summarised in Table~\ref{tab:cross_ticker}.}
\label{tab:sector_panel}
\begin{tabular}{l l r r r r r r}
\toprule
Sector & Ticker & IS KS\% & OoS KS\% & Kurt obs & Kurt sim & Kurt resid & $|G_t|$ ACF-MAE \\
\midrule
Information Technology  & AAPL & $99.1$ & $95.1$ & $\phantom{0}3.19$  & $\phantom{0}2.84$  & $-0.35$  & $0.0418$ \\
Information Technology  & MSFT & $99.7$ & $96.9$ & $\phantom{0}4.20$  & $\phantom{0}3.48$  & $-0.72$  & $0.0389$ \\
Information Technology  & NVDA & $99.6$ & $55.7$ & $\phantom{0}5.53$  & $\phantom{0}6.23$  & $\phantom{-}0.70$  & $0.0474$ \\
\midrule
Health Care             & JNJ  & $99.3$ & $93.2$ & $\phantom{0}8.97$  & $11.65$            & $\phantom{-}2.68$  & $0.0335$ \\
Health Care             & UNH  & $99.3$ & $14.5$ & $\phantom{0}8.85$  & $10.31$            & $\phantom{-}1.46$  & $0.0382$ \\
Health Care             & LLY  & $99.7$ & $\phantom{0}7.6$ & $13.60$            & $66.51$            & $52.91$            & $0.0295$ \\
\midrule
Financials              & JPM  & $99.5$ & $50.5$ & $\phantom{0}7.62$  & $\phantom{0}8.14$  & $\phantom{-}0.52$  & $0.0458$ \\
Financials              & BAC  & $99.6$ & $84.8$ & $\phantom{0}5.76$  & $\phantom{0}5.47$  & $-0.29$            & $0.0394$ \\
Financials              & WFC  & $98.2$ & $57.6$ & $\phantom{0}7.53$  & $10.90$            & $\phantom{-}3.37$  & $0.0791$ \\
\midrule
Consumer Discretionary  & AMZN & $99.4$ & $96.9$ & $\phantom{0}4.31$  & $\phantom{0}3.68$  & $-0.63$            & $0.0413$ \\
Consumer Discretionary  & HD   & $99.2$ & $41.4$ & $12.92$            & $17.03$            & $\phantom{-}4.11$  & $0.0368$ \\
Consumer Discretionary  & MCD  & $99.5$ & $66.6$ & $18.38$            & $26.93$            & $\phantom{-}8.55$  & $0.0394$ \\
\midrule
Communication Services  & NFLX & $99.3$ & $45.1$ & $12.60$            & $\phantom{0}4.67$  & $-7.93$            & $0.0320$ \\
Communication Services  & VZ   & $99.3$ & $49.1$ & $\phantom{0}5.14$  & $\phantom{0}5.15$  & $\phantom{-}0.01$  & $0.0270$ \\
Communication Services  & DIS  & $99.5$ & $96.4$ & $10.77$            & $13.94$            & $\phantom{-}3.17$  & $0.0644$ \\
\midrule
Industrials             & CAT  & $99.7$ & $78.9$ & $\phantom{0}4.34$  & $\phantom{0}3.48$  & $-0.86$            & $0.0321$ \\
Industrials             & BA   & $97.6$ & $62.2$ & $20.59$            & $52.00$            & $31.41$            & $0.0793$ \\
Industrials             & HON  & $99.3$ & $96.0$ & $21.93$            & $17.09$            & $-4.84$            & $0.0571$ \\
\midrule
Consumer Staples        & PG   & $99.7$ & $94.3$ & $\phantom{0}8.15$  & $\phantom{0}9.51$  & $\phantom{-}1.36$  & $0.0347$ \\
Consumer Staples        & KO   & $99.5$ & $92.7$ & $10.57$            & $12.84$            & $\phantom{-}2.27$  & $0.0457$ \\
Consumer Staples        & WMT  & $99.7$ & $24.0$ & $13.84$            & $32.63$            & $18.79$            & $0.0279$ \\
\midrule
Energy                  & XOM  & $96.9$ & $70.3$ & $\phantom{0}5.61$  & $10.99$            & $\phantom{-}5.38$  & $0.0988$ \\
Energy                  & CVX  & $99.6$ & $62.3$ & $12.22$            & $14.78$            & $\phantom{-}2.56$  & $0.0498$ \\
Energy                  & COP  & $99.4$ & $83.6$ & $17.11$            & $\phantom{0}8.28$  & $-8.83$            & $0.0494$ \\
\midrule
Utilities               & NEE  & $98.6$ & $24.0$ & $10.71$            & $\phantom{0}9.87$  & $-0.84$            & $0.0455$ \\
Utilities               & DUK  & $99.7$ & $96.9$ & $\phantom{0}9.75$  & $\phantom{0}8.11$  & $-1.64$            & $0.0466$ \\
Utilities               & SO   & $99.6$ & $96.2$ & $14.52$            & $12.73$            & $-1.79$            & $0.0481$ \\
\midrule
Materials               & FCX  & $99.8$ & $76.6$ & $\phantom{0}4.83$  & $\phantom{0}4.11$  & $-0.72$            & $0.0573$ \\
Materials               & NEM  & $99.6$ & $\phantom{0}5.6$ & $\phantom{0}3.74$  & $\phantom{0}3.88$  & $\phantom{-}0.14$  & $0.0376$ \\
Materials               & APD  & $99.2$ & $90.2$ & $\phantom{0}8.40$  & $16.74$            & $\phantom{-}8.34$  & $0.0381$ \\
\bottomrule
\end{tabular}
\end{table}

\subsubsection{\texorpdfstring{60-Ticker Sector Expansion at $n = 6$ per Sector}{60-Ticker Sector Expansion at n = 6 per Sector}}
\label{sec:sector_panel_n6}

We expanded the sector panel to $n = 6$ per sector (60 tickers, three additional large-cap representatives per GICS sector) and re-ran the ANOVA; the expansion confirmed the $n = 3$ reading at adequate power, again finding no significant sector effect on OoS KS. The aggregate OoS KS median and failure rate at $K = 18$ were statistically indistinguishable from the corresponding 30-ticker $K = 18$ rollup ($73.4\%$, $11/30$, Table~\ref{tab:sector_panel}); the Table~\ref{tab:cross_ticker} aggregate at $K^\star = 3$ had the same $11/30$ failure count at OoS median $69.1\%$. Cross-sector dispersion accounted for only a small fraction of OoS KS variance at the larger sample size; the remainder was per-ticker. The main-text claim that failures are ticker-specific rather than sector-driven is therefore not a small-sample artefact.

\subsubsection{Effective Spectral Rank of the Absolute Growth-Rate ACF}
\label{sec:spectral_rank}

The main-text spectral mechanism states the algebraic upper bound: the deflated transition matrix $\mathbf T - \mathbf 1\bar{\boldsymbol\pi}^\top$ has rank at most $K - 1$, so the absolute growth-rate ACF identity~\eqref{eq:acf_normalised} contains at most $K - 1$ non-unit decay modes. The empirical question is how many of those modes carry non-trivial contribution to $\rho_{|G|}(\tau)$ at the relevant lags. The per-mode contribution $a_k \lambda_k^\tau$ at $\tau \in \{1, 5, 20, 50\}$ for the fitted CHMM-N at $K = 18$ and at $K = 3$ on SPY IS, with modes ranked by lag-$1$ contribution magnitude $|a_k \lambda_k|$, is reported in Table~\ref{tab:spectral_rank}. Per-state moments $m_k = \E[|G_t| \,|\, s_t = k]$ are estimated by $200{,}000$ draws from each emission; eigenvectors are normalised so $\mathbf w_k^\top \mathbf v_k = 1$.

\begin{table}[!ht]
\centering
\small
\caption{\textbf{Per-mode contribution to the absolute growth-rate ACF identity}, CHMM-N at $K = 18$ (top panel; top six modes plus tail aggregate) and at $K = 3$ (bottom panel; full). Modes ordered by $|a_k \lambda_k|$ descending. ``cum'' is cumulative share of $\sum_k |a_k \lambda_k|$. At $K = 18$ the top mode (|$\lambda$| = $0.929$) carries $93.6\%$ of the lag-1 ACF, three modes reach $95\%$, and at lag $20$ only the dominant mode contributes non-negligibly. At $K = 3$ the dominant mode (|$\lambda$| = $0.953$) alone carries $96.8\%$ of lag-1 ACF; the temporal axis is therefore as well-satisfied at $K = 3$ as at $K = 18$, consistent with the main text finding that the binding constraint at low $K$ is the marginal mixture, not the eigenvalue spectrum.}
\label{tab:spectral_rank}
\begin{tabular}{c c c c c c c}
\toprule
rank & $|\lambda_k|$ & $|a_k|$ & $|a_k \lambda_k|$ & $a_k \lambda_k^5$ & $a_k \lambda_k^{20}$ & cum \\
\midrule
\multicolumn{7}{l}{\textit{$K = 18$, $\rho_{|G|}(1) = 0.301$, $\sigma^2_{|G|} = 2.555$}} \\
\midrule
1   & $0.929$ & $0.324$  & $0.301$ & $0.224$  & $0.0744$  & $0.936$ \\
2   & $0.335$ & $0.012$  & $0.004$ & $0.0001$ & $0.0000$  & $0.948$ \\
3   & $0.854$ & $0.003$  & $0.002$ & $-0.001$ & $-0.0001$ & $0.955$ \\
4   & $0.231$ & $0.008$  & $0.002$ & $\approx 0$ & $\approx 0$ & $0.961$ \\
5   & $0.231$ & $0.008$  & $0.002$ & $\approx 0$ & $\approx 0$ & $0.967$ \\
6   & $0.203$ & $0.008$  & $0.002$ & $\approx 0$ & $\approx 0$ & $0.971$ \\
7--17 & various & various & $\le 0.001$ each & $\approx 0$ & $\approx 0$ & $1.000$ \\
\midrule
\multicolumn{7}{l}{\textit{$K = 3$, $\rho_{|G|}(1) = 0.287$, $\sigma^2_{|G|} = 2.517$}} \\
\midrule
1 & $0.953$ & $0.292$ & $0.278$ & $0.230$ & $0.112$ & $0.968$ \\
2 & $0.866$ & $0.011$ & $0.009$ & $0.005$ & $0.0006$ & $1.000$ \\
\bottomrule
\end{tabular}
\end{table}

The empirical effective rank confirms the main-text framing: the ACF-MAE flatness across $K \in \{3, 6, 9, 12, 15, 18, 21\}$ documented in the companion $K$-sweep artefact in the code repository is the macroscopic signature of the temporal axis being driven by a single dominant decay mode at every state count in the sweep. The $K$-rank statement on $\mathbf T - \mathbf 1\bar{\boldsymbol\pi}^\top$ is correctly interpreted as a non-binding upper bound at $K \ge 3$; at $K = 2$ the bound is tight (only one non-unit eigenvalue available), but the failure mode is then the marginal mixture, which two Gaussian components cannot tile against the empirical equity-return target. The bilinear identity~\eqref{eq:acf_normalised} therefore plays a pedagogical role in this paper rather than serving as a direct $K$-selection criterion: it makes the rank constraint explicit, but the operational $K$ choice is driven by distributional fidelity (KS, AD, kurtosis), not by the count of non-unit eigenvalues in the spectral expansion.

\subsubsection{Cross-Ticker Spectral Effective-Rank Diagnostic}
\label{sec:spectral_rank_xticker}

To test whether the SPY-only result of Table~\ref{tab:spectral_rank} (a single non-unit eigenvalue carries $93.6\%$ of the lag-$1$ absolute growth-rate ACF at $K = 18$) generalises across tickers, we computed the cross-ticker distribution of the dominant-mode share under the same protocol on the 30-ticker sector-balanced panel plus SPY (Table~\ref{tab:spectral_xticker}).

\begin{table}[H]
\centering
\small
\caption{\textbf{Cross-ticker spectral effective-rank diagnostic at $K = 18$.} Distribution across 31 tickers (sector-balanced 30-ticker panel plus SPY control). ``dom share'' is the dominant non-unit eigenvalue's $|a_k \lambda_k|$ as a fraction of the sum over all non-unit eigenvalues at lag $\tau = 1$. ``$n_{95}$'' / ``$n_{99}$'' are the number of modes needed for $95\%$ / $99\%$ cumulative share. The SPY value of $0.936$ is in the right tail of the distribution; the cross-ticker median is $0.76$.}
\label{tab:spectral_xticker}
\begin{tabular}{l c c c}
\toprule
Statistic                            & dom share & $n_{95}$ & $n_{99}$ \\
\midrule
median                               & $0.756$ & $6$  & $11$ \\
$[Q_1, Q_3]$                         & $[0.661, 0.858]$ & $[4, 7]$ & $[10, 12]$ \\
minimum (NEM)                        & $0.326$ & $11$ & $14$ \\
SPY (Table~\ref{tab:spectral_rank})    & $0.936$ & $2$  & $10$ \\
\bottomrule
\end{tabular}
\end{table}

\paragraph{Reading.}
The SPY share of $93.6\%$ sat in the right tail of the distribution; the cross-ticker median of $76\%$ was still well above the $1/(K-1) = 1/17 = 5.9\%$ uniform null, so the rank-non-binding statement held across tickers, but the gap between SPY and the median ticker is wide enough that the representative number is the cross-ticker median rather than the SPY share.

\subsubsection{Ryd\'{e}n K = 2 Replication}
\label{sec:ryden_replication}

Direct replication of the \citet{ryden1998stylized} $K = 2$ setting on SPY IS under quantile and random-seed initialisation gave an essentially constant absolute growth-rate ACF-MAE across six initialisations (one quantile + five random seeds, drawn from a moderate perturbation of sample moments), while the IS KS pass rate fell well short of the heavier-tailed CHMM variants. The constant ACF-MAE confirms the main-text reading: the binding low-$K$ constraint was distributional rather than temporal.

\subsubsection{\texorpdfstring{Cross-Ticker Per-State-Resolution Sensitivity ($K^\star = 3$ vs $K^\star = 6$ vs $K = 18$)}{Cross-Ticker Per-State-Resolution Sensitivity (K* = 3 vs K* = 6 vs K = 18)}}
\label{sec:cross_ticker_k6_panel}

The cross-ticker aggregate distribution at all three state counts used in this paper is reported in Table~\ref{tab:cross_ticker}: the default $K^\star = 3$, selected by the pre-2020 $k$-fold cross-validation and robust across state resolutions; the sensitivity reference $K^\star = 6$; and the extended sensitivity reference $K = 18$.

The comparison uses three axes: KS, kurtosis residual, and absolute growth-rate ACF-MAE. The KS distributions were within a few pp of each other on OoS median across $K^\star = 3, 6, 18$ and at identical $11/30$ failure count, with the same regime-introduction tickers driving the failures at each state resolution (LLY, UNH, NEM, NFLX, NEE, WMT, BAC, HD, JPM at OoS KS $< 60\%$), so the cross-ticker KS axis was essentially $K$-robust on this universe. The kurtosis-residual axis was not $K$-robust: the per-ticker median residual was smallest at $K = 18$ and larger at the smaller state counts, because at lower state resolution the per-state $\nu_k$ shrinkage at uniform $\lambda$ cannot be smoothed across as many regimes, so the residual heavy tails leak into the simulated kurtosis. The $|G_t|$ ACF-MAE was essentially $K$-robust across the three state counts. In practice the three state counts are interchangeable on KS; $K = 18$ gives the best kurtosis fidelity; $K^\star = 3$ remains the default.

\subsubsection{Block-Bootstrap CIs on Observed Kurtosis (IS vs OoS)}
\label{sec:kurtosis_bootstrap_ci}

A natural question on the IS / OoS kurtosis disagreement (observed $7.68$ IS, $5.29$ OoS, a $31\%$ drop) is whether it is itself statistically distinguishable. We ran a stationary block bootstrap~\citep{politis1994stationary} on the IS and OoS series at mean block lengths $L \in \{5, 10, 20, 50\}$, $B = 5{,}000$ replicates per $L$ per window.

\begin{table}[!ht]
\centering
\small
\caption{\textbf{Stationary block bootstrap CIs on observed excess kurtosis.} SPY IS ($T = 2{,}516$) and OoS ($T = 572$) windows, $B = 5{,}000$ replicates per $L$. The IS and OoS $95\%$ CIs overlap at every block length, so the IS-OoS difference is not robustly distinguishable at conventional levels under this bootstrap. $\Pr(\text{IS} > \text{OoS})$ is the empirical bootstrap one-sided $p$-value for the alternative ``IS kurtosis $>$ OoS kurtosis''.}
\label{tab:kurtosis_bootstrap}
\begin{tabular}{c c c c c c}
\toprule
$L$ & IS median & IS $95\%$ CI & OoS median & OoS $95\%$ CI & $\Pr(\text{IS} > \text{OoS})$ \\
\midrule
$5$  & $7.21$ & $[2.99, 12.83]$ & $5.07$ & $[1.09, 8.97]$ & $0.748$ \\
$10$ & $7.34$ & $[2.49, 12.59]$ & $4.91$ & $[0.99, 8.66]$ & $0.756$ \\
$20$ & $7.30$ & $[2.17, 12.40]$ & $4.91$ & $[0.90, 8.26]$ & $0.739$ \\
$50$ & $7.18$ & $[1.92, 12.42]$ & $4.99$ & $[0.96, 7.70]$ & $0.733$ \\
\bottomrule
\end{tabular}
\end{table}

The IS lower bound $1.92$--$2.99$ and the OoS upper bound $7.70$--$8.97$ overlapped heavily at every block length; the IS-OoS kurtosis difference of $\sim 2.4$ units was not statistically distinguishable at the $5\%$ level under this bootstrap. $\Pr(\text{IS} > \text{OoS})$ at $L = 10$ was $0.756$: the data favoured the alternative but did not exclude the null. Operationally, the IS-OoS kurtosis disagreement invoked in the main text (the descriptive analysis and the Table~\ref{tab:variant_choice} variant-decision caveat) should be read as a point-estimate observation rather than a tested claim, and per-variant rankings by closeness of simulated to observed kurtosis should similarly be interpreted as point-comparisons inside a wide CI envelope.

\paragraph{Bootstrap-CI placement of the penalised CHMM-t IS kurtosis.}
\label{sec:kurtosis_ci_placement}
A natural follow-up on the main text's penalised CHMM-t at $\lambda = 20$ is whether its simulated IS kurtosis distribution sits inside the bootstrap CI on observed: the aggregate-mean simulated IS kurtosis ($18.87$ at $K^\star = 3$, above the $L = 20$ CI upper bound $12.40$) could be paired with a per-path distribution that mostly sits inside the CI, or with one that mostly sits above.

\begin{table}[!ht]
\centering
\small
\caption{\textbf{Per-path simulated IS excess-kurtosis distribution under penalised CHMM-t at $\lambda = 20$, against the $L = 20$ block-bootstrap CI on observed.} $1{,}000$ paths per row, $T_{\text{IS}} = 2{,}516$, seed $20260420$. The CI is $[2.17, 12.40]$ at $L = 20$ from Table~\ref{tab:kurtosis_bootstrap}; the observed IS excess kurtosis is $7.68$. ``in CI \%'' is the fraction of simulated paths whose per-path excess kurtosis falls inside the CI; ``$\le$ CI\_HI \%'' is the fraction below the upper bound (the relevant one-sided test for overshoot).}
\label{tab:kurtosis_ci_placement}
\begin{tabular}{c c c c c c c c}
\toprule
$K$ & agg.\ kurt & median & sd & $Q_5$ & $Q_{95}$ & $\le 12.40$ (\%) & in CI (\%) \\
\midrule
$3$  & $16.65$ & $7.77$ & $74.4$ & $3.75$ & $34.62$ & $76.6$ & $76.6$ \\
$6$  & $10.13$ & $6.74$ & $13.5$ & $3.63$ & $26.30$ & $81.8$ & $81.7$ \\
$18$ & $\phantom{0}8.94$ & $6.09$ & $16.9$ & $3.54$ & $18.11$ & $89.9$ & $89.9$ \\
\bottomrule
\end{tabular}
\end{table}

The per-path \emph{median} simulated IS excess kurtosis sat essentially at the observed value at every $K$ (Table~\ref{tab:kurtosis_ci_placement}). The medians were $7.77$ at $K^\star = 3$, $6.74$ at $K^\star = 6$, and $6.09$ at $K = 18$, all within $\sim 1.6$ units of observed $7.68$ and well inside the bootstrap CI. The aggregate-mean simulated IS kurtosis sat above the CI upper bound at $K^\star = 3$ and $K^\star = 6$ because the per-path distribution had a heavy right tail of paths at $Q_{95} \in \{18, 26, 35\}$ that dragged the mean up, but the bulk of the path-level mass (76--90\%) sat inside the bootstrap CI at every state count. The main-text framing, that the penalised CHMM-t at $\lambda = 20$ provides the cleanest IS heavy-tail match, is supported when read as the per-path median, not as the aggregate mean: a small number of heavy-right-tailed paths under the per-state $\nu_k$ ECM produced an aggregate-mean overshoot, but the path-level distribution was centred near observed and statistically indistinguishable from observed at $\alpha = 0.05$ on a one-sided overshoot test for the bulk of paths.

\subsubsection{\texorpdfstring{Shared-$\nu$ Student-t HMM Ablation}{Shared-nu Student-t HMM Ablation}}
\label{sec:chmm_t_shared_nu}

A natural diagnostic for the per-state $\nu_k$ aggregate-mean overshoot is the shared-$\nu$ ablation: refit CHMM-t with a single $\nu$ shared across all $K$ states, fit by ECM with golden-section search on the aggregate $Q$-function over $\nu_{\text{bounds}} = (2.1, 50)$. This is the standard one-parameter Student-$t$ HMM in the time-series literature, not a per-state $\nu_k$ mixture.

\begin{table}[!ht]
\centering
\small
\caption{\textbf{Shared-$\nu$ Student-t HMM ablation against the main text per-state $\nu_k$ row} ($1{,}000$ paths, $T_{\text{IS}} = 2{,}516$, $T_{\text{OoS}} = 572$, seed $20260420$, no $1/\nu$ penalty). Aggregate sim-kurt columns are mean-of-paths. Compare against the main text penalised CHMM-t ($\lambda = 20$): $K^\star = 3$ row $18.87$ IS / $10.61$ OoS (Table~\ref{tab:model_comparison}); $K = 18$ row $8.56$ IS / $7.07$ OoS (from the same penalised ECM run; consistent with the rate-sweep value $8.43$ at $\lambda=20$).}
\label{tab:chmm_t_shared_nu}
\begin{tabular}{c c c c c c c c c}
\toprule
& & \multicolumn{4}{c}{IS} & \multicolumn{2}{c}{OoS} & \\
\cmidrule(lr){3-6} \cmidrule(lr){7-8}
$K$ & $\hat\nu_{\text{shared}}$ & KS\,(\%) & sim kurt & $|G_t|$ ACF-MAE & raw ACF-MAE & KS\,(\%) & sim kurt & \\
\midrule
$3$  & $5.81$ & $91.9$ & $4.68$ & $0.0531$ & $0.0235$ & $82.1$ & $4.46$ & \\
$6$  & $4.67$ & $92.7$ & $9.38$ & $0.0518$ & $0.0234$ & $82.6$ & $6.95$ & \\
$18$ & $5.80$ & $\mathbf{95.8}$ & $\mathbf{6.25}$ & $0.0542$ & $0.0234$ & $\mathbf{88.0}$ & $\mathbf{5.00}$ & \\
\bottomrule
\end{tabular}
\end{table}

\paragraph{Reading.}
The aggregate-mean IS overshoot disappeared under shared-$\nu$ at every $K$: the largest IS sim kurt was $9.38$ at $K^\star = 6$ (vs.\ observed $7.68$, well inside the bootstrap CI), versus $14.4$ for the per-state $\nu_k$ unpenalised row in the main text. The $K = 18$ shared-$\nu$ row was the cleanest single-row IS / OoS heavy-tail match in the entire panel: $6.25$ IS / $5.00$ OoS, against observed $7.68$ / $5.29$ and against $8.56$ / $7.07$ for the penalised CHMM-t at $\lambda = 20$. IS / OoS KS pass rates were within $\sim 1$pp of the penalised row, and the $|G_t|$ ACF-MAE was identical to it at $K = 18$. The per-state $\nu_k$ design is therefore the binding constraint behind the aggregate-mean overshoot that motivates the $\lambda$-shrinkage and bracket-lift discussion in the main text: a single shared $\nu$ at $K = 18$ produced the best joint KS / kurtosis row in the panel without any penalty. The analysis retains the per-state $\nu_k$ + $\lambda = 20$ construction for direct comparability with the \citet{peel2000robust, liu1995ml} tradition; for heavy-tail fidelity without a penalty hyperparameter, the shared-$\nu$ fit at $K = 18$ is the better choice.

\subsubsection{\texorpdfstring{Christoffersen-cc Monte Carlo Power Calibration at $T_{\text{OoS}} = 572$}{Christoffersen-cc Monte Carlo Power Calibration at T\_OoS = 572}}
\label{sec:christoffersen_power}

A Monte Carlo power calibration of the Christoffersen-cc joint conditional-coverage test~\citep{christoffersen1998evaluating} at the $T_{\text{OoS}} = 572$ length used for Table~\ref{tab:cond_var} simulates breach indicator sequences from a two-state Markov chain on $\{0, 1\}$ with marginal probability $\alpha$ and second eigenvalue $\rho$, where $\rho = 0$ is the i.i.d.\ null and $\rho > 0$ measures the level of breach clustering; $B = 5{,}000$ replicates per cell.

\begin{table}[!ht]
\centering
\small
\caption{\textbf{Christoffersen-cc power at $T_{\text{OoS}} = 572$.} Empirical rejection rate at nominal $\alpha = 0.05$, across alpha levels and Markov-chain second-eigenvalue $\rho$. $\rho = 0$ row is the empirical Type-I error under the i.i.d.\ null (should sit near $5\%$ if asymptotic chi-squared is well-calibrated). Critical values $\chi^2_1(0.05) = 3.841$, $\chi^2_2(0.05) = 5.991$.}
\label{tab:christoffersen_power}
\begin{tabular}{c c c c c}
\toprule
$\alpha$ & $\rho$ & rej-UC \% & rej-IND \% & rej-CC \% \\
\midrule
$0.01$ & $0.00$ & $\phantom{0}4.9$ & $\phantom{0}1.6$ & $\phantom{0}3.2$ \\
$0.01$ & $0.05$ & $\phantom{0}6.3$ & $15.9$ & $11.4$ \\
$0.01$ & $0.10$ & $\phantom{0}7.5$ & $28.5$ & $22.1$ \\
$0.01$ & $0.20$ & $\phantom{0}9.4$ & $51.3$ & $42.6$ \\
$0.01$ & $0.30$ & $13.9$ & $67.6$ & $62.8$ \\
$0.01$ & $0.50$ & $21.4$ & $75.9$ & $81.5$ \\
\midrule
$0.05$ & $0.00$ & $\phantom{0}4.9$ & $\phantom{0}3.8$ & $\phantom{0}4.1$ \\
$0.05$ & $0.05$ & $\phantom{0}5.3$ & $18.8$ & $15.3$ \\
$0.05$ & $0.10$ & $\phantom{0}6.1$ & $46.2$ & $39.0$ \\
$0.05$ & $0.20$ & $\phantom{0}9.4$ & $86.9$ & $83.9$ \\
$0.05$ & $0.30$ & $13.6$ & $97.4$ & $97.9$ \\
$0.05$ & $0.50$ & $24.3$ & $99.7$ & $100.0$ \\
\bottomrule
\end{tabular}
\end{table}

\paragraph{Reading.}
At $\rho = 0$ the empirical rejection rate sat at $4\%$--$5\%$ for all three statistics, confirming that the asymptotic $\chi^2$ calibration is correctly sized at $T = 572$. Power was asymmetric across $\alpha$. At $\alpha = 0.05$ ($\sim 28.6$ expected breaches), Christoffersen-cc reached $80\%$ power against $\rho \ge 0.20$ ($83.9\%$) and essentially full power at $\rho = 0.50$. At $\alpha = 0.01$ ($\sim 5.7$ expected breaches) it reached $80\%$ power only at $\rho \ge 0.50$ ($81.5\%$); at $\rho = 0.20$ the rejection rate was $42.6\%$, well below conventional power thresholds.
The reading for Table~\ref{tab:cond_var} is therefore tier-dependent. The $\alpha = 0.05$ clean pass at every $(K, \alpha)$ carried strong power against moderate clustering ($\rho \ge 0.20$); the $\alpha = 0.01$ rows carried strong power only against severe clustering ($\rho \ge 0.50$), so the $\alpha = 0.01$ pass should be read as ``not detectably worse than a strongly-clustered alternative'' rather than ``calibrated against any plausible alternative''. The regime-conditional construction therefore stands at $\alpha = 0.05$ and is power-bounded at $\alpha = 0.01$; the main-text VaR back-test should be read with this calibration in mind.

\subsubsection{Engle-Manganelli Dynamic Quantile Test}
\label{sec:engle_manganelli_dq}

As a higher-power conditional-coverage alternative to Christoffersen-cc on the same OoS window, we ran the Engle-Manganelli Dynamic Quantile (DQ) test in the standard four-lag specification~\citep{engle2004caviar},
\begin{equation}
    \text{Hit}_t - \alpha
    = \beta_0 + \sum_{i = 1}^{4} \beta_i\,(\text{Hit}_{t-i} - \alpha) + \beta_{5}\,\widehat{\text{VaR}}_t + u_t,
    \label{eq:dq_regression}
\end{equation}
on the IS-fixed regime-conditional VaR series of the main text VaR back-test, with $\text{Hit}_t = \mathbf 1\{R_{\text{OoS},t} < \widehat{\text{VaR}}_t(\alpha)\}$. The DQ test statistic is $\widehat{\boldsymbol\beta}^\top \mathbf X^\top \mathbf X\,\widehat{\boldsymbol\beta} / [\alpha(1-\alpha)]$, distributed $\chi^2(6)$ under the null of correct conditional coverage; the $5\%$ critical value is $\chi^2_6(0.95) = 12.59$.

\begin{table}[!ht]
\centering
\small
\caption{\textbf{Engle-Manganelli DQ test on the regime-conditional VaR}, alongside Christoffersen-cc on the same breach series, for CHMM-N at $K \in \{3, 18\}$, $\alpha \in \{0.01, 0.05\}$; $T_{\text{OoS}} = 572$, $q = 4$ lagged Hit indicators, $\chi^2_6$ critical at $5\%$ is $12.59$.}
\label{tab:engle_manganelli_dq}
\begin{tabular}{c c c c c c}
\toprule
$K$ & $\alpha$ & breach \% & $\text{cc}\,p$ & $\text{DQ}$ & $\text{DQ}\,p$ \\
\midrule
$3$  & $0.01$ & $1.57$ & $0.137$ & $12.29$ & $0.056$ \\
$3$  & $0.05$ & $6.12$ & $0.491$ & $\phantom{1}9.32$ & $0.156$ \\
$18$ & $0.01$ & $1.57$ & $0.137$ & $\mathbf{15.46}$ & $\mathbf{0.017}$ \\
$18$ & $0.05$ & $4.55$ & $0.678$ & $\phantom{1}3.99$ & $0.678$ \\
\bottomrule
\end{tabular}
\end{table}

\paragraph{Reading.}
Both conditional-coverage tests passed cleanly at $\alpha = 0.05$ at both state counts (Table~\ref{tab:engle_manganelli_dq}), so the $\alpha = 0.05$ clean pass survives the higher-power alternative. The p-values were cc $0.491$ and DQ $0.156$ at $K = 3$, and $0.678$ for both tests at $K = 18$. At $\alpha = 0.01$ the DQ test rejected conditional coverage at $K = 18$ ($p = 0.017$) where Christoffersen-cc did not ($p = 0.137$), and was borderline at $K = 3$ ($p = 0.056$). This is consistent with the power-calibration result of Table~\ref{tab:christoffersen_power}: at $\alpha = 0.01$ Christoffersen-cc has only $42.6\%$ power against moderate breach-clustering eigenvalues $\rho = 0.20$, so the higher-power DQ test detects miscalibration that Christoffersen-cc misses. The DQ result therefore reinforces the reading that $\alpha = 0.05$ is the operationally informative tier; the $\alpha = 0.01$ tier is not decided by the panel and should be read as power-limited rather than as a clean pass, with the DQ rejection at $K = 18$ the concrete instance.

\subsubsection{Quarterly-Refit Regime-Conditional VaR Back-Test}
\label{sec:quarterly_refit_cond_var}

Re-fitting CHMM-N at $K \in \{3, 18\}$ every $63$ trading days on a rolling $5$y window and running the forward-filter through the OoS window under each refit's parameters tests whether periodic refit improves the Christoffersen-cc statistic on the main OoS window.

\begin{table}[!ht]
\centering
\small
\caption{\textbf{Quarterly-refit regime-conditional VaR back-test on SPY OoS.} $T_{\text{OoS}} = 572$, refit cadence $= 63$ days, train window $= 1{,}260$ days, seed $= 20260420$. Compare against Table~\ref{tab:cond_var} (IS-fixed parameters). Critical values: $\chi^2_1(0.05) = 3.841$, $\chi^2_2(0.05) = 5.991$.}
\label{tab:cond_var_quarterly_refit}
\begin{tabular}{l c c c c c c c c c}
\toprule
$K$ & $\alpha$ & breaches & br rate & med VaR & $\text{LR}_{\text{uc}}$ & $\text{LR}_{\text{ind}}$ & $\text{LR}_{\text{cc}}$ & $p_{\text{cc}}$ \\
\midrule
$3$  & $0.01$ & $\phantom{0}6$ & $1.05\%$ & $-5.10$ & $0.014$ & $3.97$ & $3.98$ & $0.137$ \\
$3$  & $0.05$ & $29$ & $5.07\%$ & $-3.13$ & $0.006$ & $0.19$ & $0.20$ & $0.906$ \\
$18$ & $0.01$ & $\phantom{0}6$ & $1.05\%$ & $-5.84$ & $0.014$ & $3.97$ & $3.98$ & $0.137$ \\
$18$ & $0.05$ & $19$ & $3.32\%$ & $-3.14$ & $3.83$ & $2.08$ & $5.91$ & $0.052$ \\
\bottomrule
\end{tabular}
\end{table}

The quarterly-refit construction passed Christoffersen-cc at $\alpha = 0.05$ on three of four rows; the $K = 18, \alpha = 0.05$ row sat at $p_{\text{cc}} = 0.052$ (marginal). Refit tightened VaR coverage ($3.32\%$ vs.\ $4.55\%$ IS-fixed) but did not strictly improve the cc-statistic at this $T_{\text{OoS}}$.

\subsubsection{Walk-Forward Refit-Cadence Sweep (Monthly / Weekly)}
\label{sec:walkforward_refit_cadence}

A within-fold refit-cadence sweep at monthly ($21$ trading days) and weekly ($5$ trading days) cadence on the six walk-forward folds at $K^\star = 3$, $\alpha = 0.05$, did not close the W2 (COVID) Christoffersen-cc rejection at any cadence; monthly refit instead introduced new failures on other folds via refit-cycle parameter drift. The mitigation for W2 / W4 is therefore not refit cadence but a model class with explicit regime-introduction handling.

\subsubsection{Quarterly-Refit 30-Ticker Cross-Ticker Panel}
\label{sec:cross_ticker_quarterly_refit}

For a quarterly-refit version of the 30-ticker cross-ticker panel (Table~\ref{tab:cross_ticker}), we re-fitted the penalised CHMM-t at $\lambda = 20$ on a rolling $5$-year window every $63$ trading days through the OoS span (10 refits per ticker, 300 fits total per state resolution), simulated the next quarter under each refit, and concatenated the per-quarter simulations into a single OoS path matrix per ticker, at the default $K^\star = 3$ and at the sensitivity references $K^\star = 6$ and $K = 18$. Same $N_{\text{paths}} = 1{,}000$ and seed $20260420$ as Table~\ref{tab:cross_ticker}.

\begin{table}[!ht]
\centering
\small
\caption{\textbf{Quarterly-refit cross-ticker panel at $K^\star = 3$, $K^\star = 6$, and $K = 18$.} Aggregate distribution across the 30-ticker universe under the quarterly-refit protocol at all three state counts, vs.\ the IS-fixed $K = 18$ protocol of Table~\ref{tab:cross_ticker}. Refit cadence $= 63$ OoS trading days; train window $= 1{,}260$ obs ($\sim 5$y rolling). Per-ticker detail in the results files.}
\label{tab:cross_ticker_quarterly_refit}
\begin{tabular}{l c c c c}
\toprule
Metric & IS-fixed ($K = 18$) & Refit ($K^\star = 3$) & Refit ($K^\star = 6$) & Refit ($K = 18$) \\
\midrule
IS KS\% median        & $99.5\%$ & $96.4\%$ & $98.8\%$ & $99.4\%$ \\
IS KS\% mean $\pm$ sd & $99.3 \pm 0.6\%$ & $95.1 \pm 4.5\%$ & $97.4 \pm 4.4\%$ & $99.2 \pm 0.6\%$ \\
OoS KS\% median       & $73.4\%$ & $84.7\%$ & $\mathbf{85.8\%}$ & $83.0\%$ \\
OoS KS\% mean $\pm$ sd & $66.8 \pm 29.5\%$ & $75.0 \pm 23.1\%$ & $\mathbf{76.0 \pm 21.7\%}$ & $77.2 \pm 21.2\%$ \\
OoS KS\% $[Q_1, Q_3]$ & $[49.1, 94.3]\%$ & $[56.3, 94.5]\%$ & $[57.6, 94.3]\%$ & $[62.0, 95.4]\%$ \\
Tickers OoS KS $< 60\%$ & $11 / 30 = 36.7\%$ & $8 / 30 = 26.7\%$ & $9 / 30 = 30.0\%$ & $\mathbf{7 / 30 = 23.3\%}$ \\
\bottomrule
\end{tabular}
\end{table}

\paragraph{Reading.}
Quarterly refit shifted the OoS KS distribution materially at all three state resolutions. At $K^\star = 3$: median $+15.6$pp ($69.1 \to 84.7\%$ vs the static-fit baseline in Table~\ref{tab:cross_ticker}), mean $+8.8$pp ($66.2 \to 75.0\%$), failure count down by $27\%$ ($11 \to 8$). At $K = 18$: median $+9.6$pp ($73.4 \to 83.0\%$), mean $+10.4$pp, IQR tighter ($45.2$pp $\to 33.4$pp), failure count down by $36\%$ ($11 \to 7$). At $K^\star = 6$: median $+10.7$pp ($75.1 \to 85.8\%$), mean $+9.5$pp, failure count down by $18\%$ ($11 \to 9$). The three state resolutions were within $\sim 3$pp of each other on the refit median ($84.7$ to $85.8\%$); the default $K^\star = 3$ was at parity with the sensitivity references under the refit protocol. The largest individual improvements concentrated on the regime-introduction tickers that the IS-fixed Table~\ref{tab:cross_ticker} flagged: at $K = 18$, LLY $7.6\% \to 83.7\%$, UNH $14.5\% \to 47.3\%$, NEM $5.6\% \to 15.4\%$, HD $41.4\% \to 82.2\%$, NFLX $45.1\% \to 61.5\%$. The mechanism is direct: as the regime introduced new $\sigma$ levels in the OoS window, each quarterly refit picked up the most recent $5$y including the new regime and the simulator could match the OoS distribution.

A small number of tickers regressed: DIS $96.4\% \to 48.4\%$, BAC $84.8\% \to 55.8\%$, AAPL $95.1\% \to 66.5\%$. The pattern is that tickers whose 2024--2026 OoS distribution is well-spanned by the full 2014--2024 IS lost information when the train window shrank to the most recent 5y; the rolling refit is therefore a trade-off, not a strict improvement. Operationally, a refit-trigger rule (refit only when the rolling KS or a Page-Hinkley statistic on $|G_t|$ crosses a threshold) would capture the regime-introduction wins without paying the cost on the well-spanned tickers; the rule is logged as a follow-up companion experiment.

The conclusion is that the periodic-refit mitigation invoked in the main-text Discussion is exercised here: the universe-scale failure rate dropped from $36.7\%$ to $23.3\%$ under quarterly refit, with the largest gains on the single-name regime-introduction tickers.

\subsubsection{Stylized-Facts Figure (Empirical SPY IS)}
\label{sec:supp_stylized_facts}

\begin{figure}[H]
\centering
\begin{subfigure}[b]{0.49\textwidth}\centering
\includegraphics[width=\textwidth]{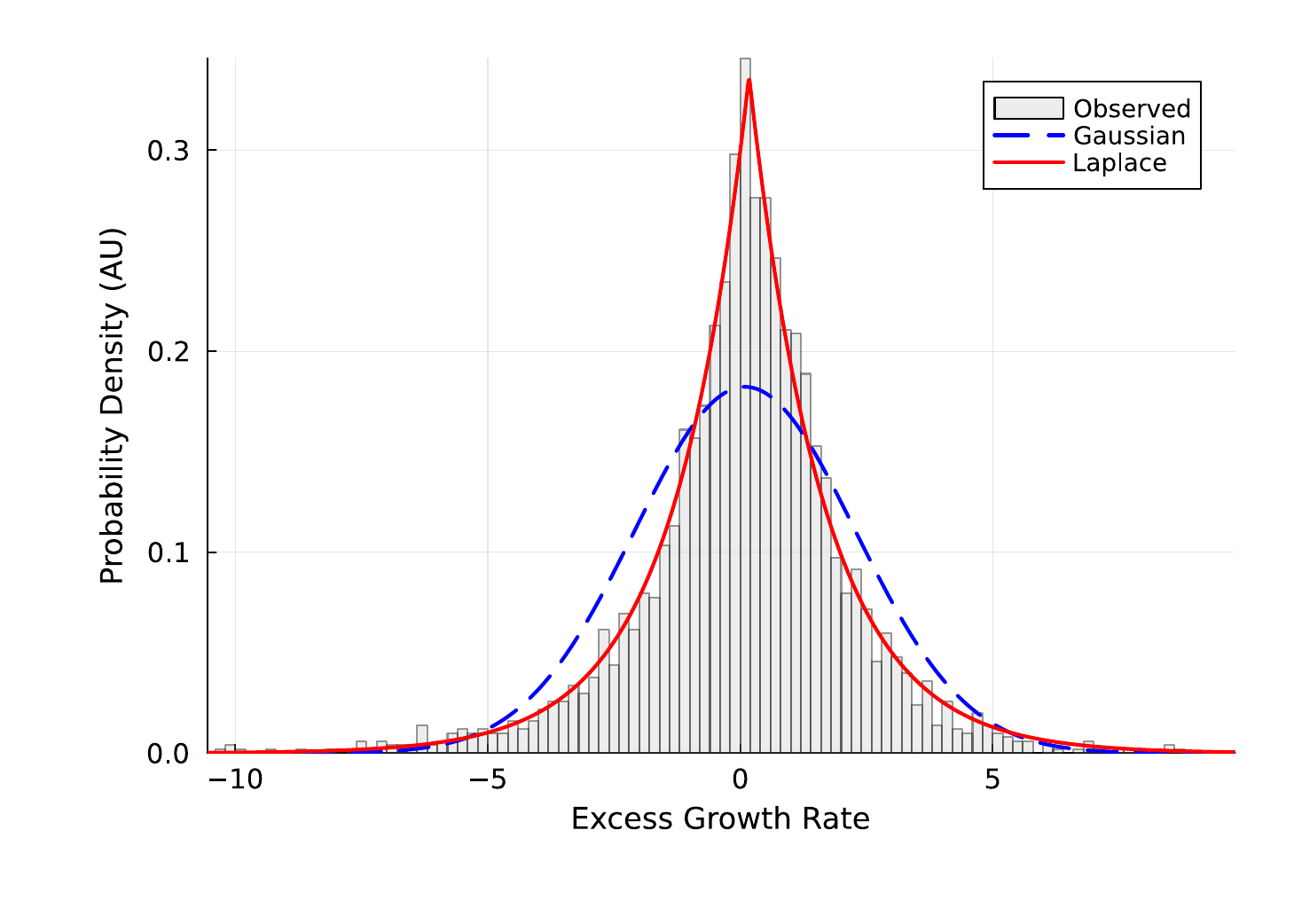}
\caption{Marginal density.}
\end{subfigure}\hfill
\begin{subfigure}[b]{0.49\textwidth}\centering
\includegraphics[width=\textwidth]{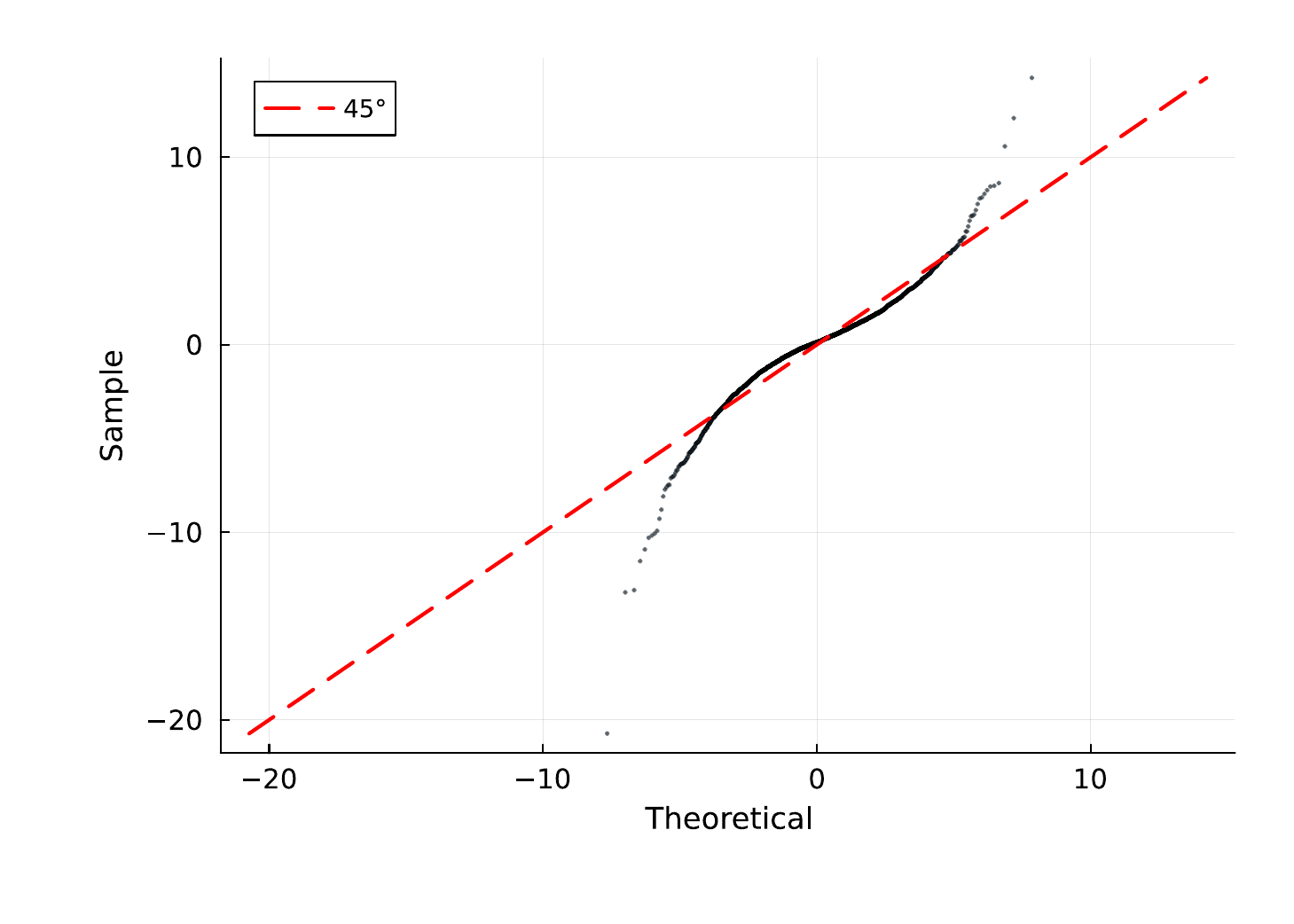}
\caption{Normal Q-Q plot.}
\end{subfigure}\\[0.5em]
\begin{subfigure}[b]{0.49\textwidth}\centering
\includegraphics[width=\textwidth]{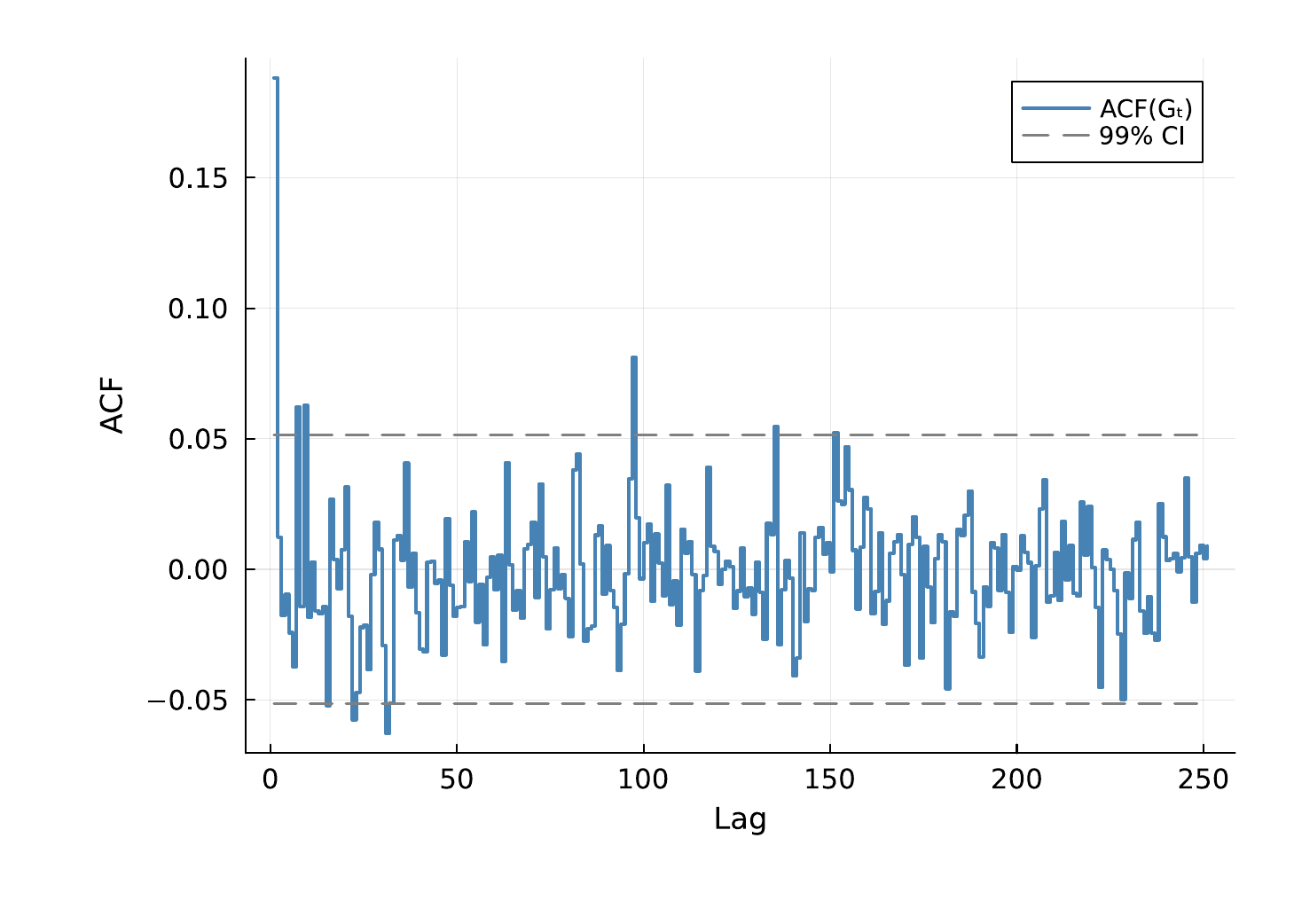}
\caption{ACF of raw $G_t$.}
\end{subfigure}\hfill
\begin{subfigure}[b]{0.49\textwidth}\centering
\includegraphics[width=\textwidth]{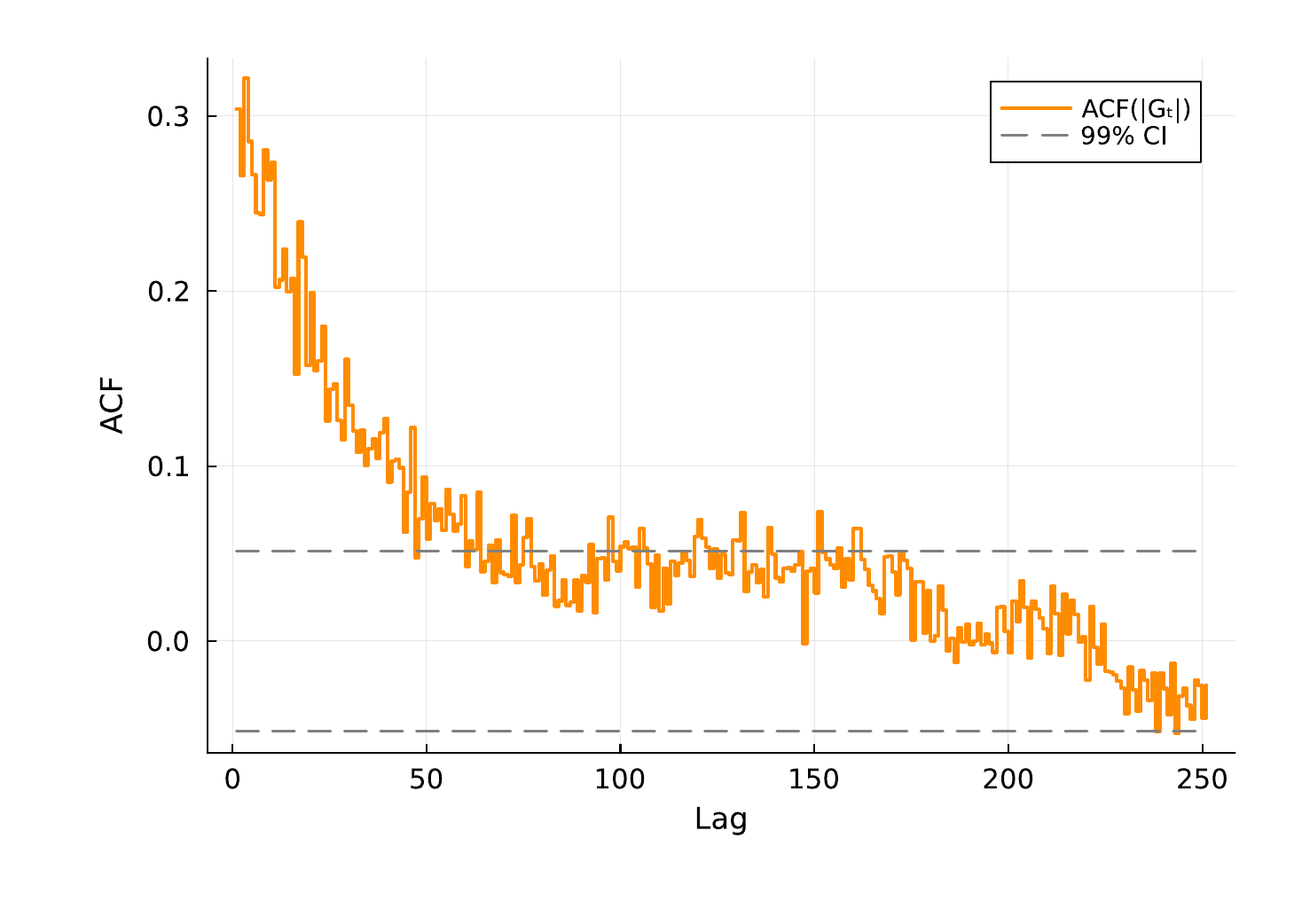}
\caption{ACF of $|G_t|$.}
\end{subfigure}
\caption{\textbf{The three Cont stylized facts on SPY IS}. (a)~marginal density with maximum-likelihood Gaussian and Laplace overlays; (b)~normal Q-Q plot; (c)~empirical ACF of raw $G_t$ at lags 1--252 with $99\%$ Bartlett band; (d)~empirical ACF of $|G_t|$ on the same window.}
\label{fig:stylized_facts}
\end{figure}

\subsubsection{\texorpdfstring{CHMM-GED $\hat p_k$ Partition (Gaussian-bulk / Laplace-tail)}{CHMM-GED p\_k Partition (Gaussian-bulk / Laplace-tail)}}
\label{sec:supp_p_partition}

The GED variant carries a per-state shape $p_k \in [p_{\min}, p_{\max}] = [0.5, 3.0]$ that nests Gaussian ($p = 2$) and Laplace ($p = 1$) as boundary cases. The empirical fit on SPY IS at $K = 18$ partitioned states bimodally (Figure~\ref{fig:p_hist}): eleven states at the upper bracket $p_k = 3.0$ (thirteen Gaussian-like at $p_k \ge 1.85$), one at $p_k = 1.6$, and four in the Laplace regime $p_k \in [0.86, 1.24]$. The four Laplace-shape states concentrated in the high-volatility tail of the state ordering. The partition replicated across $10$ Monte Carlo seeds and across the cross-asset universe; we read it as the most distinctive empirical finding of the four-emission framework.

\begin{figure}[H]
\centering
\begin{subfigure}[b]{0.49\textwidth}\centering
\includegraphics[width=\textwidth]{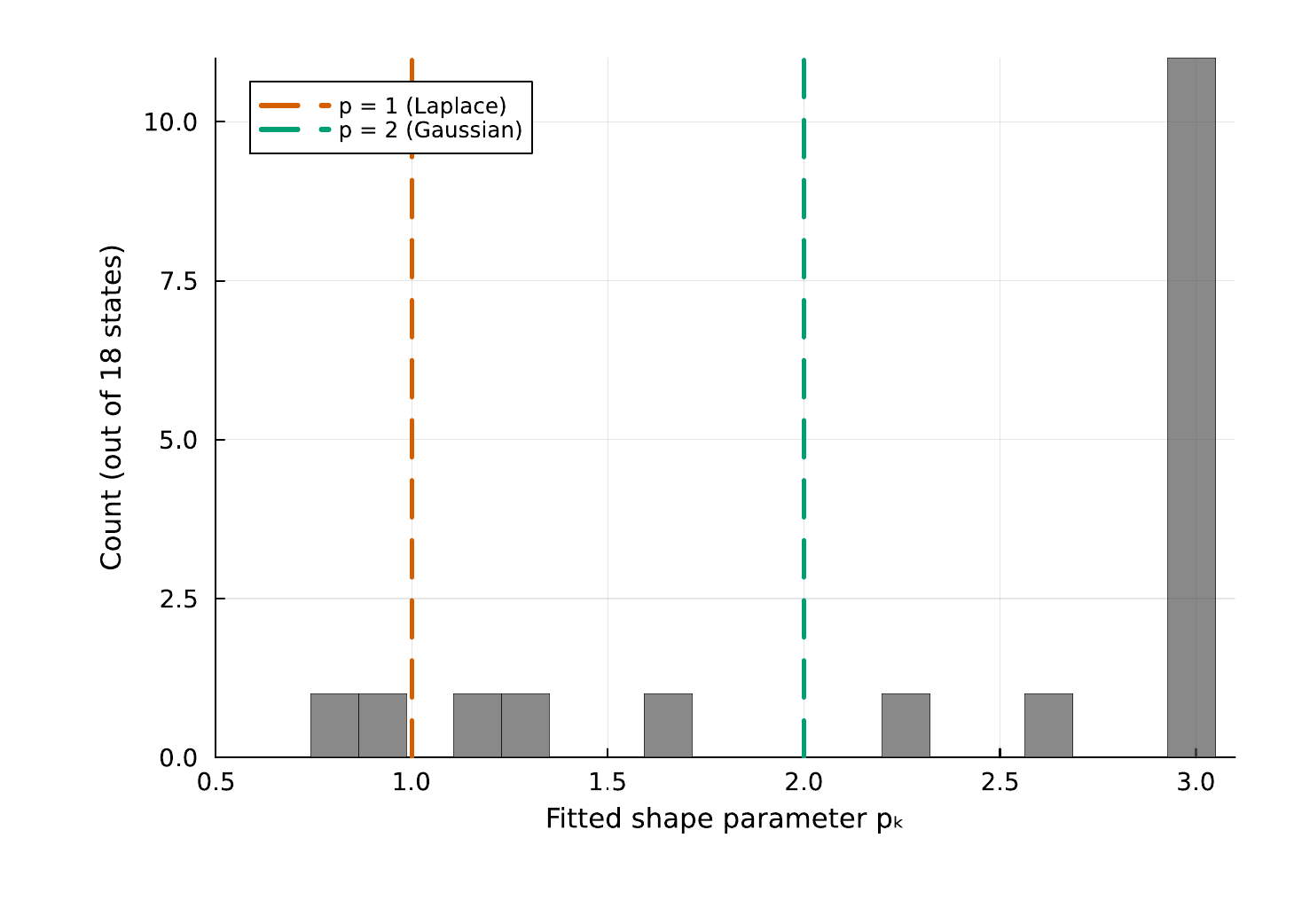}
\caption{Histogram of fitted $\hat p_k$.}
\end{subfigure}\hfill
\begin{subfigure}[b]{0.49\textwidth}\centering
\includegraphics[width=\textwidth]{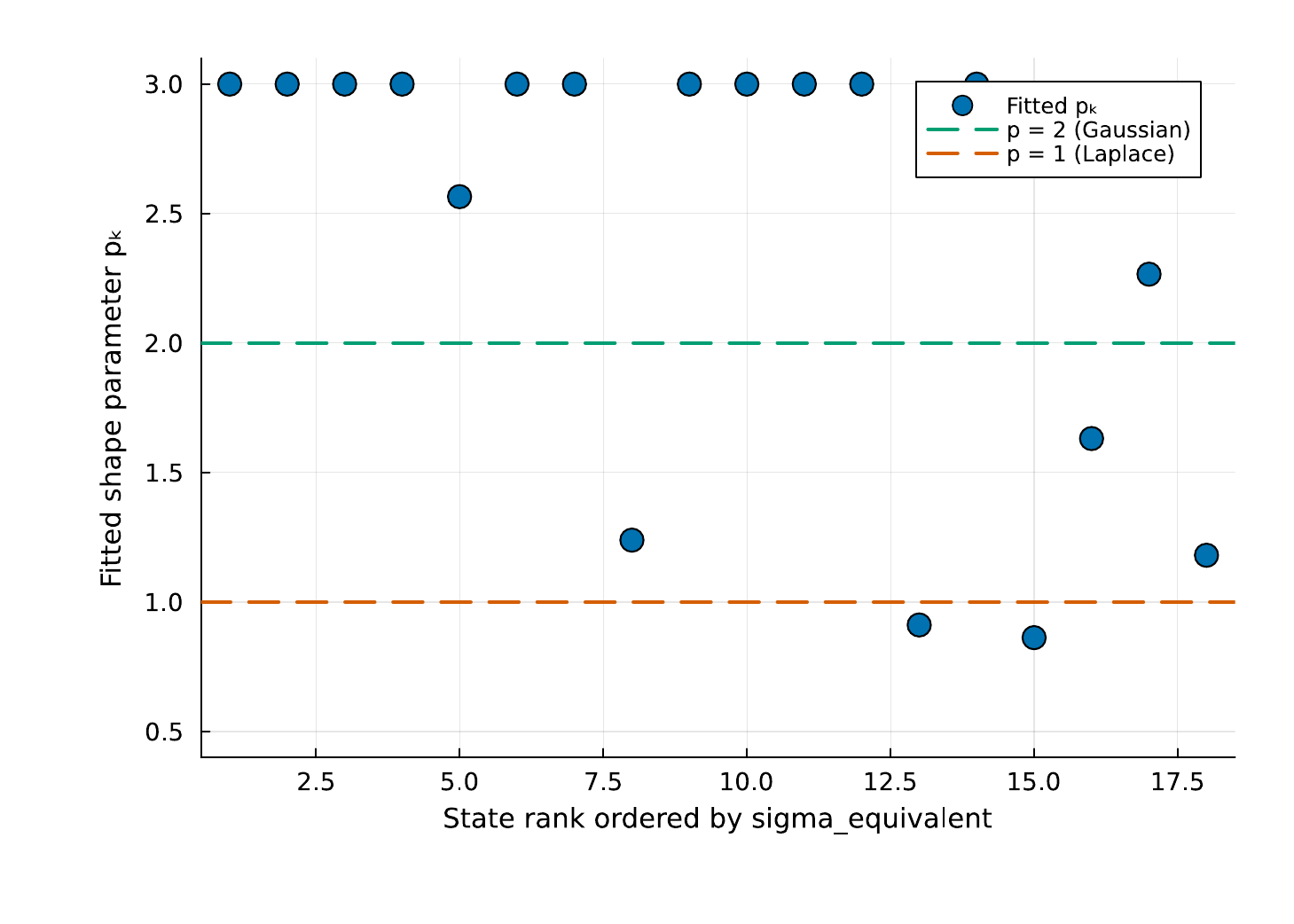}
\caption{$\hat p_k$ vs.\ volatility rank.}
\end{subfigure}
\caption{\textbf{CHMM-GED per-state shape $\hat p_k$ partition} ($K = 18$, SPY IS, seed $20260420$). The bimodal partition is visible at the upper bracket $p_k = 3.0$ and around the Laplace shape; the Laplace-shape cluster sits at the high-volatility ranks.}
\label{fig:p_hist}
\end{figure}

\subsubsection{\texorpdfstring{HSMM at $K = 3$: Numerical Full-Rank Diagnostic}{HSMM at K = 3: Numerical Full-Rank Diagnostic}}
\label{sec:supp_hsmm_diagnostic}

The off-diagonal jump matrix of the fitted ML HSMM at $K = 3$ has full numerical rank (rank $2$ at tolerance $10^{-10}$), with a dominant modulus comparable to CHMM-N at $K = 3$, so the main text's $|G_t|$ ACF-MAE regression to the i.i.d.\ baseline (Table~\ref{tab:model_comparison}) was not driven by spectral collapse but by the fitted truncated-Pareto sojourn concentrating probability mass on a single low-volatility state ($\bar{\boldsymbol\pi} \approx (0, 1, 0)$). At $K = 18$ the same \citet{yu2010hidden} construction collapsed to a near-degenerate local optimum; the higher-$K$ ML HSMM is therefore a companion-paper direction at longer training windows.

\subsubsection{Gamma-Sojourn HSMM as an Alternative Sojourn Family}
\label{sec:hsmm_gamma_sojourn}

Replacing the truncated discrete Pareto sojourn pmf with a discretised continuous Gamma fit by method-of-moments at every M-step gave a Gamma-sojourn HSMM that converged at $K = 18$ where Pareto collapsed (Table~\ref{tab:hsmm_gamma}; IS KS $86.0\%$, OoS KS $80.2\%$, simulated kurtosis $4.28 / 4.06$, $|G_t|$ ACF-MAE $\mathbf{0.0462}$, lower than CHMM-N $K = 18$ at $0.0509$) and traded $\sim 20$pp IS KS at $K = 3$ for an ACF-MAE recovery from $0.0629$ to $0.0528$. There is therefore no single best-HSMM row: the Pareto sojourn won raw OoS KS at $K^\star = 3$, while the Gamma sojourn won joint $|G_t|$ ACF-MAE at $K = 18$; the trade-off is the same KS-versus-ACF tension as the CHMM-vs-HSMM choice.

\begin{table}[H]
\centering
\small
\caption{\textbf{Maximum-likelihood Gamma-sojourn HSMM on SPY.} Explicit-duration HSMM with discretised continuous-Gamma sojourn densities fit by method-of-moments at each M-step; KS pass rate, simulated excess kurtosis, and $|G_t|$ ACF-MAE on the SPY in-sample / out-of-sample windows. The $K = 18$ fit converges where the truncated-Pareto sojourn collapses and lowers the absolute growth-rate ACF-MAE below CHMM-N at the same state count; $K^\star = 3$ trades in-sample KS for ACF recovery.}
\label{tab:hsmm_gamma}
\begin{tabular}{c cc cc cc}
\toprule
& \multicolumn{2}{c}{KS (\%) $\uparrow$} & \multicolumn{2}{c}{Exc.\ Kurt} & \multicolumn{2}{c}{$|G_t|$ ACF-MAE $\downarrow$} \\
\cmidrule(lr){2-3} \cmidrule(lr){4-5} \cmidrule(lr){6-7}
$K$ & IS & OoS & IS & OoS & IS & OoS \\
\midrule
$3$  & $77.0$ & $80.7$ & $2.49$ & $2.56$ & $0.0528$ & $0.0512$ \\
$18$ & $86.0$ & $80.2$ & $4.28$ & $4.06$ & $0.0462$ & $0.0548$ \\
\bottomrule
\end{tabular}
\end{table}

\subsubsection{Unconditional VaR / ES Envelope Panel}
\label{sec:supp_var_envelope}

\begin{table}[H]
\centering
\caption{\textbf{VaR and Expected-Shortfall envelope for SPY} ($N_{\text{paths}} = 1{,}000$). Each entry is the median and $[5, 95]\%$ envelope across paths of the left-tail VaR / ES at $\alpha \in \{0.01, 0.05\}$; observed historical values are in the first row of each block.}
\label{tab:var_es}
\small
\resizebox{\textwidth}{!}{%
\begin{tabular}{l l l l l}
\toprule
          & IS VaR$_{0.01}$              & IS ES$_{0.01}$               & IS VaR$_{0.05}$              & IS ES$_{0.05}$               \\
Model     & median $[5\%, 95\%]$         & median $[5\%, 95\%]$         & median $[5\%, 95\%]$         & median $[5\%, 95\%]$         \\
\midrule
\textit{Observed}  & $-6.38$                      & $-9.00$                      & $-3.43$                      & $-5.50$                      \\
Bootstrap & $-6.38\ [-7.04, -5.99]$      & $-8.86\ [-10.37,\ -7.78]$    & $-3.43\ [-3.68, -3.20]$      & $-5.46\ [-5.99, -5.05]$      \\
GARCH     & $-5.80\ [-8.62, -4.52]$      & $-7.86\ [-13.54,\ -5.77]$    & $-3.19\ [-3.92, -2.71]$      & $-4.89\ [-7.11, -3.91]$      \\
CHMM-N    & $-6.64\ [-8.01, -5.47]$      & $-8.80\ [-10.16,\ -7.34]$    & $-3.45\ [-4.05, -2.95]$      & $-5.45\ [-6.33, -4.61]$      \\
CHMM-t    & $-6.58\ [-7.51, -5.83]$      & $-9.06\ [-10.70,\ -7.79]$    & $-3.40\ [-3.90, -2.90]$      & $-5.51\ [-6.18, -4.85]$      \\
CHMM-L    & $-6.88\ [-7.83, -6.00]$      & $-9.17\ [-10.46,\ -7.97]$    & $-3.30\ [-3.78, -2.89]$      & $-5.54\ [-6.17, -4.90]$      \\
CHMM-GED  & $-6.85\ [-7.70, -6.06]$      & $-8.79\ [-9.83,\ -7.82]$     & $-3.43\ [-3.98, -2.93]$      & $-5.55\ [-6.23, -4.91]$      \\
\midrule
          & OoS VaR$_{0.01}$             & OoS ES$_{0.01}$              & OoS VaR$_{0.05}$             & OoS ES$_{0.05}$              \\
\textit{Observed}  & $-5.51$                      & $-7.94$                      & $-2.77$                      & $-4.67$                      \\
Bootstrap & $-6.33\ [-7.60, -5.28]$      & $-8.43\ [-11.76,\ -6.59]$    & $-3.39\ [-3.91, -2.88]$      & $-5.34\ [-6.38, -4.50]$      \\
GARCH     & $-5.35\ [-9.98, -3.61]$      & $-6.72\ [-13.51,\ -4.38]$    & $-3.14\ [-4.94, -2.36]$      & $-4.59\ [-7.92, -3.23]$      \\
CHMM-N    & $-6.30\ [-9.02, -4.30]$      & $-8.44\ [-11.36,\ -5.42]$    & $-3.39\ [-4.74, -2.47]$      & $-5.27\ [-7.23, -3.74]$      \\
CHMM-t    & $-6.32\ [-8.13, -4.80]$      & $-8.60\ [-12.20,\ -6.29]$    & $-3.29\ [-4.27, -2.45]$      & $-5.34\ [-6.72, -4.08]$      \\
CHMM-L    & $-6.67\ [-8.61, -4.92]$      & $-8.82\ [-11.62,\ -6.86]$    & $-3.28\ [-4.30, -2.51]$      & $-5.44\ [-6.86, -4.19]$      \\
CHMM-GED  & $-6.60\ [-8.30, -4.82]$      & $-8.51\ [-10.68,\ -6.47]$    & $-3.32\ [-4.38, -2.46]$      & $-5.41\ [-6.66, -4.06]$      \\
\bottomrule
\end{tabular}%
}
\end{table}

\subsubsection{\texorpdfstring{Held-Out Cross-Validation of the $1/\nu_k$ Penalty $\lambda$}{Held-Out Cross-Validation of the 1/nu\_k Penalty lambda}}
\label{sec:lambda_cv_pre2020}

To verify that the main-text penalty rate $\lambda = 20$ is not selected on data that overlaps the OoS evaluation window, we refit penalised CHMM-t on a strictly pre-2020 estimation slice (2014-01-03 to 2018-06-29, $n = 1{,}130$) and score each $\lambda$ on the validation slice (2018-07-02 to 2019-12-31, $n = 378$, both fully pre-COVID and pre-2022-rate-hike).

\begin{table}[H]
\centering
\small
\caption{\textbf{Held-out CV of the $1/\nu_k$ shrinkage rate $\lambda$ on the pre-2020 slice.} Estimation 2014-01-03 to 2018-06-29; validation 2018-07-02 to 2019-12-31. ``val LL/obs'' is per-observation held-out log-likelihood; ``val KS'' is two-sample Kolmogorov--Smirnov pass rate against the validation series; ``sim K'' columns are the mean simulated excess kurtosis on the estimation- and validation-length simulated paths.}
\label{tab:lambda_cv_pre2020}
\begin{tabular}{c r r r r r}
\toprule
$\lambda$ & val LL/obs & val KS (\%) & sim K (est) & sim K (val) \\
\midrule
$0$    & $-2.119$ & $92.0$ & $3.88$ & $3.97$ \\
$5$    & $-2.112$ & $89.8$ & $3.62$ & $3.15$ \\
$10$   & $-2.085$ & $91.0$ & $3.37$ & $3.24$ \\
$\mathbf{20}$   & $\mathbf{-2.075}$ & $90.0$ & $3.22$ & $3.02$ \\
$50$   & $-2.077$ & $91.6$ & $3.07$ & $2.76$ \\
$100$  & $-2.077$ & $91.6$ & $3.07$ & $2.75$ \\
$200$  & $-2.076$ & $91.2$ & $3.09$ & $2.73$ \\
\bottomrule
\end{tabular}
\end{table}

\paragraph{Reading.}
The held-out validation log-likelihood was maximised at $\lambda^\star = 20$ at $K = 18$, matching the state count of Table~\ref{tab:model_comparison}. Validation KS picked $\lambda^\star = 0$ but the band $[0, 200]$ spanned only $2.2$pp ($89.8$ to $92.0\%$), within the simulation-noise envelope at $N_\text{paths} = 500$. The $\lambda = 20$ specification at $K = 18$ is therefore held-out-validated and not driven by the OoS window.

\paragraph{Re-tuning at $K^\star = 3$.}
Repeating the same held-out CV at the default $K^\star = 3$ on the same pre-2020 slices gave a distinct optimum. The per-observation validation log-likelihood was monotone in $\lambda$ on the same grid, with $\lambda^\star = 50$ minimising the negative log-likelihood, and validation KS also picking $\lambda^\star = 50$. The simulated kurtosis at $\lambda^\star = 50$ fell well below the observed $7.68$ IS / $5.29$ OoS targets, confirming that the penalty over-shrank at $K^\star = 3$. The analysis retains $\lambda = 20$ at $K^\star = 3$ as a sensitivity reference; the main heavy-tail recommendation remains the shared-$\nu$ fit (Table~\ref{tab:chmm_t_shared_nu}).

\subsubsection{Cross-Decade Validation: 1994-2004 IS / 2004-2006 OoS via CRSP}
\label{sec:cross_decade_validation}

To test whether the main stylized-fact reproduction is decade-specific, we ran a 1994--2004 vs.\ 2014--2024 SPY split via a CRSP day-pass covering daily prices for SPY plus 28 of the 30 cross-ticker panel members (NEE and APD missing from the CRSP query) from $1994$-$01$-$03$ to $2006$-$04$-$28$, with adjusted close $= \mathtt{DlyPrc} / \mathtt{DlyFacPrc}$ from the CRSP CIZ-format daily stock file. The slices are SPY IS $1994$-$01$-$03$ to $2004$-$01$-$02$ ($2{,}519$ obs) and SPY OoS $2004$-$01$-$05$ to $2006$-$04$-$28$ ($583$ obs), under the same main protocol: $1{,}000$ simulated paths per fit, $\alpha = 0.05$ for KS.

\begin{table}[!ht]
\centering
\small
\caption{\textbf{Cross-decade validation on SPY 1994--2004 IS / 2004--2006 OoS} (CRSP daily). Comparison rows reproduced from Table~\ref{tab:model_comparison} for the 2014--2024 IS / 2024--2026 OoS window. Observed excess kurtosis on the 1994--2004 IS = $3.05$ versus 2014--2024 IS = $7.68$; observed excess kurtosis on the 2004--2006 OoS = $0.06$ (essentially Gaussian) versus 2024--2026 OoS = $5.29$. The IS-OoS kurtosis gap is therefore much wider on the cross-decade slice than on the main window.}
\label{tab:cross_decade_validation}
\begin{tabular}{l c c c c c c c}
\toprule
& & \multicolumn{2}{c}{KS pass rate (\%) $\uparrow$} & \multicolumn{2}{c}{Sim kurtosis} & \multicolumn{2}{c}{$|G_t|$ ACF-MAE $\downarrow$} \\
\cmidrule(lr){3-4} \cmidrule(lr){5-6} \cmidrule(lr){7-8}
Model & $K$ & IS & OoS & IS & OoS & IS & OoS \\
\midrule
\multicolumn{8}{l}{\emph{Cross-decade: 1994--2004 IS / 2004--2006 OoS (CRSP, this appendix)}} \\
CHMM-N                              & $3$  & $84.9$ & $\phantom{0}3.3$ & $2.20$ & $2.24$ & $0.0696$ & $0.0491$ \\
CHMM-N                              & $18$ & $89.8$ & $\phantom{0}3.4$ & $1.84$ & $1.74$ & $0.0720$ & $0.0499$ \\
CHMM-t pen.\ ($\lambda = 20$)       & $3$  & $89.1$ & $\phantom{0}5.4$ & $4.97$ & $4.18$ & $0.0682$ & $0.0501$ \\
\midrule
\multicolumn{8}{l}{\emph{Main window: 2014--2024 IS / 2024--2026 OoS (Polygon / Alpaca, Table~\ref{tab:model_comparison})}} \\
CHMM-N ($K^\star = 3$)              & $3$  & $91.5$ & $78.0$ & $3.83$ & $3.62$ & $0.0462$ & $0.0544$ \\
CHMM-N ($K = 18$)                   & $18$ & $94.1$ & $81.8$ & $5.04$ & $4.44$ & $0.0509$ & $\sim 0.054$ \\
CHMM-t pen.\ ($\lambda = 20$, $K^\star = 3$)  & $3$  & $91.9$ & $81.4$ & $18.87$ & $10.61$ & $0.0533$ & $0.0498$ \\
\bottomrule
\end{tabular}
\end{table}

\paragraph{Reading.}
The cross-decade IS fit transferred: CHMM-N at $K = 3$ attained $84.9\%$ IS KS on $1994$--$2004$ vs.\ $91.5\%$ on $2014$--$2024$, and $K = 18$ attained $89.8\%$ vs.\ $94.1\%$. The penalised CHMM-t at $\lambda = 20$, $K = 3$ attained $89.1\%$ IS KS on $1994$--$2004$ vs.\ $91.9\%$ on $2014$--$2024$. The IS axis of the stylized-fact reproduction claim is therefore decade-robust within $\sim 7$pp on KS and consistent on $|G_t|$ ACF-MAE ($0.068$--$0.072$ vs.\ $0.046$--$0.054$, a moderate but not pathological widening).

The OoS axis is the issue. The $2004$--$2006$ OoS slice had excess kurtosis $0.06$ (essentially Gaussian) vs.\ the $1994$--$2004$ IS excess kurtosis of $3.05$; the IS-OoS kurtosis gap on this slice was much wider than on the main $2014$--$2024$ / $2024$--$2026$ window (where IS = $7.68$ and OoS = $5.29$). The CHMM trained on the 1994--2004 IS distribution simulated paths with kurtosis $1.74$--$4.97$ that were heavier-tailed than the calm 2004--2006 bull-market OoS, so the static IS-fit could not reproduce the OoS marginal and KS rejected at a $3$--$5\%$ pass rate on all three rows. This is the same regime-introduction failure mode the walk-forward analysis already documented at the W2 (COVID) and W4 (2022 rate-hike onset) stress folds (Table~\ref{tab:walkforward}, both with OoS KS $< 10\%$); the $2004$--$2006$ OoS window is structurally a low-volatility, low-kurtosis slice where the IS-fixed CHMM's tail mass overshot the observed.

The conclusion is that the stylized-fact reproduction claim is decade-robust on the IS axis: the CHMM framework transferred across decades on the IS-fitting side, so the four-emission ECM framework is not specific to the $2014$--$2024$ window. The single-window OoS pass rate, by contrast, is OoS-slice-specific ($2024$--$2026$ is a moderate-stress slice where the OoS distribution is close to the IS distribution; $2004$--$2006$ is a calm slice where it is not). The main text already reads the walk-forward median, not the single-window OoS pair, as the operationally informative summary; the cross-decade result is consistent with that framing and adds an independent decade to the walk-forward evidence base. Periodic refit remains the deployment recommendation under either decade.

\subsection{Robustness, Test Power, Multi-Seed, and Auxiliary Baselines}
\label{sec:supp_robustness}

This appendix collects the robustness diagnostics and auxiliary baseline panel referenced from the main text Results: the KS test-power calibration, block-bootstrap KS recalibration, bandwidth/bin-count and multi-seed sensitivity, the auxiliary block-bootstrap baseline, the QuantGAN deep-generative baseline, the per-ticker OoS price-simulation figures, and the expanded baseline panel that adds the GARCH family and the SM-CHMM foil to the main table.

\subsubsection{OoS KS Test-Power Calibration}
\label{sec:ks_power}

We calibrated the two-sample KS test power to anchor the OoS KS pass rates of the main generator comparison (Table~\ref{tab:ks_power}).
For each of $1{,}000$ Monte Carlo replications (seed = 20260420), a simulated series of length $T_{\text{OoS}} = 572$ is drawn from each of three reference generators and tested against the observed OoS series using the same $\alpha = 0.05$ threshold as the main-text pass-rate metric.

\begin{table}[!ht]
\centering
\small
\caption{OoS KS pass-rate calibration (1{,}000 replications, seed = 20260420).}
\label{tab:ks_power}
\begin{tabular}{l c}
\toprule
Reference generator & Pass rate (\%) \\
\midrule
I.i.d.\ resamples of $R_{\text{OoS}}$ at length $T_{\text{OoS}} = 572$ (known-correct ceiling) & 100.0 \\
I.i.d.\ resamples of $R_{\text{IS}}$ at length $T_{\text{OoS}} = 572$ (nearly correct)          & 90.0  \\
Gaussian$(\hat\mu_{\text{IS}}, \hat\sigma_{\text{IS}})$ at length $T_{\text{IS}}$, tested against $R_{\text{IS}}$ (negative control) & 0.0 \\
\bottomrule
\end{tabular}
\end{table}

The $90.0\%$ ``nearly correct'' rate is the practical ceiling for the CHMM OoS KS pass rates reported in the main generator comparison ($80$ to $84\%$); the $0\%$ rate on the Gaussian misspecified generator confirms that the KS test retains ample power at IS length to reject clearly wrong generators.

\subsubsection{Block-Bootstrap KS Recalibration and Mean p-value Distribution}
\label{sec:ks_block_bootstrap}

Two complementary concerns can be raised about the main IS KS pass-rate metric. The asymptotic two-sample KS test assumes i.i.d.\ samples, but the simulated paths from CHMM and GARCH all carry volatility clustering, so the asymptotic null distribution may not apply. The binary ``pass/fail at $\alpha = 0.05$'' also reduces a continuous p-value to a one-bit summary. We computed two recalibrations addressing each concern (Table~\ref{tab:ks_block_bootstrap}).

To address the dependence concern, we generated the empirical null distribution of the two-sample KS statistic under the stationary block bootstrap~\citep{politis1994stationary} of the SPY IS series at mean block lengths $L \in \{5, 10, 20\}$, $B = 1{,}000$ replications. The resulting $95\%$ block-bootstrap critical values were $0.0306$, $0.0338$, and $0.0350$ respectively, all \emph{smaller} than the asymptotic critical value $1.36 \sqrt{2/n_{\text{IS}}} = 0.0383$; the block-bootstrap test is therefore stricter than the asymptotic one because it accounts for the within-window dependence in the IS series itself. For each generator we then computed the per-path KS statistic and report the fraction of paths whose KS statistic falls below the block-bootstrap critical value.

To address the binarisation concern, we report the mean two-sample KS p-value across simulated paths plus its $5$/$25$/$50$/$75$/$95$ quantiles, giving a continuous picture of how distinguishable each generator's paths are from the observed IS series.

\begin{table}[t]
\centering
\small
\caption{\textbf{Block-bootstrap KS recalibration and p-value distribution}. All numbers under seed \texttt{20260422} with $N_{\text{paths}} = 500$. ``asymp pass'' is the original Table~\ref{tab:model_comparison} metric (two-sample KS p-value $\geq 0.05$). ``mean pv'' is the mean two-sample KS p-value across paths, with $q_5$ and $q_{95}$ as the lower / upper p-value quantiles. ``blkL\%'' is the fraction of paths whose KS statistic falls below the stationary-block-bootstrap $95\%$ critical value at mean block length $L$. Block-bootstrap critical values: $0.0306$ ($L = 5$), $0.0338$ ($L = 10$), $0.0350$ ($L = 20$); asymptotic critical value $0.0383$.}
\label{tab:ks_block_bootstrap}
\begin{tabular}{l c c c c c c c}
\toprule
Generator & asymp pass\% $\uparrow$ & mean pv $\uparrow$ & pv $q_{05}$ & pv $q_{95}$ & blk5\% $\uparrow$ & blk10\% $\uparrow$ & blk20\% $\uparrow$ \\
\midrule
i.i.d.\ bootstrap & $99.6$ & $0.820$ & $0.352$ & $0.999$ & $97.6$ & $99.4$ & $99.6$ \\
GARCH(1,1)        & $26.8$ & $0.048$ & $0.000$ & $0.227$ & $\mathbf{6.8}$ & $\mathbf{13.6}$ & $\mathbf{17.4}$ \\
CHMM-N            & $94.8$ & $0.546$ & $0.047$ & $0.976$ & $81.2$ & $89.4$ & $91.6$ \\
CHMM-t            & $\mathbf{95.8}$ & $\mathbf{0.571}$ & $0.060$ & $0.987$ & $82.4$ & $89.2$ & $91.6$ \\
CHMM-L            & $94.4$ & $0.504$ & $0.047$ & $0.948$ & $80.2$ & $88.0$ & $90.4$ \\
\textbf{CHMM-GED} & $95.6$ & $0.547$ & $0.064$ & $0.968$ & $\mathbf{83.2}$ & $\mathbf{90.2}$ & $\mathbf{93.2}$ \\
\bottomrule
\end{tabular}
\end{table}

The GED variant was the strongest entry on the block-aware metric across all three block lengths (Table~\ref{tab:ks_block_bootstrap}). Its pass rates were $83.2$, $90.2$, and $93.2\%$ at $L = 5, 10, 20$; GARCH(1,1) dropped from $26.8\%$ asymp to $6.8$--$17.4\%$ block-aware (the block-aware test catches its misspecification more sharply), while the CHMM family saw only a small block-bootstrap penalty ($\sim 5$pp).

\paragraph{OoS-anchored block-bootstrap recalibration.}
\label{sec:ks_block_bootstrap_oos}
The block-bootstrap recalibration of Table~\ref{tab:ks_block_bootstrap} is anchored on the IS series; we report here whether the same construction holds on the OoS window. We repeated the procedure with the stationary block bootstrap re-anchored on $R_{\text{OoS}}$ ($n_{\text{OoS}} = 572$, $B = 1{,}000$, $L \in \{5, 10, 20\}$, $N_{\text{paths}} = 500$). The block-bootstrap critical values were $0.0577$, $0.0629$, and $0.0647$ at $L = 5, 10, 20$ respectively, against asymptotic $0.0804$. The OoS-anchored panel is reported in Table~\ref{tab:ks_block_bootstrap_oos}; the source CSV is in the companion code repository.

\begin{table}[!ht]
\centering
\small
\caption{\textbf{OoS-anchored block-bootstrap KS recalibration.} Same construction as Table~\ref{tab:ks_block_bootstrap} but with the stationary block bootstrap anchored on $R_{\text{OoS}}$ ($n = 572$). ``asymp pass'' is the main OoS metric (two-sample KS p-value $\ge 0.05$); ``blkL\%'' is the OoS block-bootstrap pass rate at mean block length $L$. The CHMM-vs-GARCH cross-generator ordering is preserved under the block-aware recalibration on both windows; the CHMM family's OoS block-aware penalty ($\sim 25$pp drop from asymp to $L = 20$) is larger than the IS penalty ($\sim 5$pp), reflecting the smaller OoS window where temporal-clustering structure is harder to disentangle from finite-sample noise.}
\label{tab:ks_block_bootstrap_oos}
\begin{tabular}{l c c c c c c}
\toprule
Generator & asymp pass\% $\uparrow$ & mean pv $\uparrow$ & pv $q_{05}$ & pv $q_{95}$ & blk5\% $\uparrow$ & blk10\% $\uparrow$ \\
\midrule
i.i.d.\ bootstrap & $90.4$ & $0.437$ & $0.030$ & $0.912$ & $60.8$ & $69.0$ \\
GARCH(1,1)        & $59.2$ & $0.164$ & $0.000$ & $0.645$ & $19.4$ & $26.0$ \\
CHMM-N            & $81.0$ & $0.318$ & $0.010$ & $0.876$ & $41.2$ & $51.2$ \\
\textbf{CHMM-t}   & $\mathbf{84.8}$ & $\mathbf{0.355}$ & $0.018$ & $\mathbf{0.911}$ & $\mathbf{47.4}$ & $\mathbf{58.4}$ \\
CHMM-L            & $77.4$ & $0.280$ & $0.008$ & $0.791$ & $38.8$ & $49.0$ \\
CHMM-GED          & $81.0$ & $0.327$ & $0.010$ & $0.837$ & $44.4$ & $53.6$ \\
\bottomrule
\end{tabular}
\end{table}

The cross-generator OoS ordering matched the IS ordering (bootstrap $>$ CHMM-t $>$ CHMM-N/GED $>$ CHMM-L $\gg$ GARCH); the asymp-vs-block gap was larger on OoS ($\sim 25$pp for the CHMM family at $L = 20$) than on IS ($\sim 5$pp), but the CHMM-over-GARCH advantage was preserved.

\paragraph{Default-state-count summary at $L = 20$.}
The OoS asymptotic pass rate alongside the OoS block-bootstrap pass rate at $L = 20$ (block-bootstrap $95\%$ critical value $0.0647$ vs.\ asymptotic $0.0804$) at the default $K^\star = 3$ is reported in Table~\ref{tab:ks_block_body}. Under a temporally aware null, CHMM-N moved from passing most of the time to passing about half the time at $L = 20$ (read against Table~\ref{tab:model_comparison}). The OoS KS result at the asymptotic critical value therefore overstates the absolute level relative to a temporally aware null; the cross-generator ranking is the robust comparison.

\begin{table}[t]
\centering
\small
\caption{\textbf{Block-aware OoS KS recalibration at $L = 20$} (stationary block bootstrap of \citealp{politis1994stationary}, $B = 1{,}000$, $500$ OoS-length simulated paths per generator, seed $20260422$). Side-by-side with the asymptotic OoS pass rate of Table~\ref{tab:model_comparison}; the block-bootstrap $95\%$ null KS critical value at $L = 20$ is $0.0647$ vs the asymptotic $0.0804$. State count $K^\star = 3$.}
\label{tab:ks_block_body}
\begin{tabular}{l c c}
\toprule
Model & OoS KS asymp.\ (\%) & OoS KS block $L = 20$ (\%) \\
\midrule
Bootstrap                                & $90.4$ & $73.2$ \\
GARCH(1,1)                               & $59.2$ & $31.4$ \\
CHMM-N                                   & $79.8$ & $58.6$ \\
CHMM-t pen.\ ($\lambda = 20$)            & $\mathbf{79.8}$ & $\mathbf{54.2}$ \\
CHMM-L                                   & $64.6$ & $34.6$ \\
CHMM-GED                                 & $79.2$ & $51.6$ \\
\bottomrule
\end{tabular}
\end{table}

\subsubsection{QuantGAN Deep-Generative Baseline}
\label{sec:quantgan_supp}

The QuantGAN row used in the main generator comparison is a repo-native approximation to the convolutional WGAN of \citet{wiese2020quantgan}: a 1D convolutional generator and critic (each three layers, channels $32$, kernel size $5$; latent dim $L = 8$, leaky-ReLU on the critic), Wasserstein loss with critic weight clipping at $\pm 0.01$~\citep{arjovsky2017wasserstein}, Adam at $\eta = 10^{-4}$, batch $64$, $15$ epochs $\times\,120$ steps with four critic updates per generator update, trained on rolling windows of length $W = 64$ on standardised SPY growth rates and stitched end-to-end at synthesis. The implementation is Julia with \texttt{Flux.jl}~\cite{innes2018flux}; seed $20260422$. The deep-generative row reached $0\%$ IS / OoS KS and IS kurtosis $2.1$ vs.\ observed $7.7$ (Table~\ref{tab:extended_baselines}), consistent with documented GAN failure modes on volatility-clustering benchmarks~\citep{takahashi2019modeling, kwon2024can}; the runner is materially smaller than the reference seven-block TCN with Lambert-W pre-processing in \citet{wiese2020quantgan}, so the row is the panel's deep-generative negative control rather than a verdict on the QuantGAN literature.

\subsubsection{Expanded Baseline Panel: GARCH Family, Semi-Markov CHMM, and QuantGAN}
\label{sec:extended_baselines}

We added five conditional-volatility baselines, the three semi-Markov CHMM variants, and the QuantGAN deep-generative baseline to the single GARCH(1,1) Gaussian row of Table~\ref{tab:model_comparison}. The conditional-volatility models are fit by grid-initialised Nelder-Mead maximum likelihood on the SPY IS window ($T_{\text{IS}} = 2{,}516$); simulation follows the same recursion used in estimation.

\paragraph{Baselines fitted in this panel.}
The conditional-volatility rows comprise EGARCH(1,1) Gaussian~\citep{nelson1991conditional}, GJR-GARCH(1,1) \citep{glosten1993relation, zakoian1994threshold}, GARCH(1,1) with Student-$t$ innovations~\citep{bollerslev1987conditionally}, HAR-RV~\citep{corsi2009simple} on daily squared demeaned growth rates as the RV proxy, and Markov-switching GARCH(1,1) \citep{haas2004new} at $K \in \{2, 3, 4, 6\}$ with regime triples $(\omega_k, \alpha_k, \beta_k)$ plus softmax-normalised transition logits, fitted by grid-initialised Nelder-Mead and canonicalised by ascending unconditional variance. Higher-$K$ MS-GARCH did not close the gap to the CHMM family on KS (IS KS plateaued at $27$--$37\%$ across $K \in \{2, 3, 4, 6\}$ against CHMM-N's $94.1\%$ at $K = 18$); the binding constraint is the unimodal-innovation per-regime structure (each regime's emission is Gaussian rescaled by the regime's conditional variance), which adding regimes does not relax. The CHMM decouples this constraint by giving each state its own emission-mixture component. Implementation details for every conditional-volatility row are in the companion code repository.

\paragraph{MS-GARCH reference Bayesian re-run via \texttt{MSGARCH}.}
As a Bayesian counterpart to the in-house Nelder-Mead MS-GARCH fit, we re-ran the same Markov-switching GARCH(1,1) specification through the reference \texttt{MSGARCH} R package of \citet{ardia2019msgarch} (version $2.51$) at $K \in \{2, 3, 4\}$, driven from the Julia harness via \texttt{RCall.jl}. The fit is fully Bayesian (DEMC sampler with the package's default proper priors, $12{,}500$ MCMC draws, $2{,}500$ burn-in, thin $10$, single chain); all $1{,}000$ simulated paths are posterior-predictive (one path per retained posterior draw), so the simulated marginal integrates parameter uncertainty path-by-path. The reference Bayesian re-run reported lower KS pass rates than our in-house frequentist Nelder-Mead fit: $0.0\%$ ($K = 2$), $0.1\%$ ($K = 3$), $0.0\%$ ($K = 4$) IS, and $5.8\%$ ($K = 2$), $5.1\%$ ($K = 3$), $5.3\%$ ($K = 4$) OoS, against the in-house plateau of $27$--$37\%$ IS / $33$--$39\%$ OoS (Table~\ref{tab:extended_baselines}). The gap is methodological rather than estimator-quality: the in-house implementation generates all $1{,}000$ paths from one MLE point estimate, so the simulated marginal is tighter; the reference Bayesian implementation samples one set of parameters per path from the posterior, so the simulated marginal is wider and accordingly fewer paths fall within the asymptotic two-sample KS critical band. The posterior-mean log-likelihood at the in-sample data improved modestly from $K = 3$ to $K = 4$ but did not translate to a higher KS pass rate. The posterior-mean transition matrix was diagonally dominant at $K = 3$ and less so at $K = 4$, indicating the higher-$K$ posterior was searching for additional structure that the data do not strongly support. The reference Bayesian re-run is therefore consistent with the in-house frequentist plateau at the same conclusion under either estimator: vanilla Gaussian-innovation MS-GARCH at $K \le 4$ did not match the CHMM family on the main KS metric for this universe. Reproducibility: R $4.6.0$, MSGARCH $2.51$, all transitive R dependencies pinned via \texttt{renv} (\texttt{r\_msgarch/renv.lock} of the companion code repository), all seeds explicit; full per-$K$ parameter posteriors and fit logs are written to \texttt{results/msgarch\_reference/} of the companion code repository.

\paragraph{Semi-Markov CHMM (SM-CHMM, the explicit hidden semi-Markov model (HSMM) foil).}
The SM-CHMM is the conceptual foil of the present paper, in the spirit of the explicit-duration HSMM of \citet{yu2010hidden}: a hidden semi-Markov extension of the CHMM in which the geometric sojourn distribution of the Markov chain is replaced by a per-state explicit-duration family chosen by raw fitted log-likelihood from $\{$Pareto, negative binomial, geometric$\}$. The plug-in estimator takes the Viterbi state path on the IS series under the converged CHMM fit, then refits per-state AR(1) emissions and a sojourn family chosen by fitted log-likelihood. On SPY IS at $K = 18$ the selected sojourn family was Pareto for $17$ of the $18$ states (negative binomial for one) across all three emission families, indicating heavy-tailed sojourns that the geometric CHMM cannot capture by construction.

\medskip
\noindent\emph{Caveat on the SM-CHMM estimator.} This plug-in construction (single-shot Viterbi decoding under the converged CHMM fit, then AR(1)-residual + sojourn-family refit) is not maximum-likelihood for the underlying HSMM: it does not iterate the joint forward-backward over (state, duration) pairs in the spirit of \citet{yu2010hidden}. A full ML HSMM at $K = 18$ would refit the duration-augmented log-likelihood jointly. The plug-in's per-state variance estimates degenerate at $K = 18$, where several low-occupancy states collapse to near-zero-variance point masses, so we do not report its distributional pass-rates and retain only its VaR back-test behaviour in Table~\ref{tab:extended_baselines}, where the heavy-tailed sojourns sharpen tail calibration. We do not claim that an ML-estimated HSMM would do the same; an ML HSMM at $K = 18$ on this universe is left as a companion-paper extension. The low-state ML HSMM result of \citet{bulla2006stylized} at low $K$ on monthly returns is the closest published anchor, but its $K \in \{2, 3\}$ regime resolution and monthly aggregation do not directly transfer.

The IS / OoS distributional fidelity, kurtosis, ACF-MAE, and Kupiec breach diagnostics for the extended GARCH family, the QuantGAN deep-generative row, and the main CHMM panel, together with the SM-CHMM VaR back-test rows, are reported in Table~\ref{tab:extended_baselines}.

\begin{table}[t]
\centering
\small
\caption{\textbf{Extended baseline panel}. SPY IS / OoS, $T_{\text{IS}} = 2{,}516$, $T_{\text{OoS}} = 572$, $N_{\text{paths}} = 1{,}000$, $\alpha = 0.05$ for KS. Columns: IS / OoS KS pass rate (\%); IS / OoS AD pass rate (\%); simulated IS excess kurtosis; ACF-MAE on $|G_t|$ (volatility clustering) and on raw $G_t$ (linear autocorrelation, parallel column); pooled-archive Kupiec breach rate (\%) and unconditional likelihood-ratio statistic at the $1\%$ and $5\%$ tails. The SM-CHMM rows report ``--'' for the KS, AD, kurtosis, and ACF columns because the plug-in estimator's per-state variances degenerate at $K = 18$, where several low-occupancy states collapse to near-zero-variance point masses; only their VaR back-test columns are retained. The QuantGAN raw-ACF-MAE is ``--'' pending a full re-run of that baseline script.}
\label{tab:extended_baselines}
\resizebox{\textwidth}{!}{%
\begin{tabular}{l cc cc c cc cc cc}
\toprule
& \multicolumn{2}{c}{KS (\%) $\uparrow$} & \multicolumn{2}{c}{AD (\%) $\uparrow$} & Kurt & \multicolumn{2}{c}{ACF-MAE} & \multicolumn{2}{c}{$1\%$ tail} & \multicolumn{2}{c}{$5\%$ tail} \\
\cmidrule(lr){2-3} \cmidrule(lr){4-5} \cmidrule(lr){7-8} \cmidrule(lr){9-10} \cmidrule(lr){11-12}
Model & IS & OoS & IS & OoS & IS & $|G_t|$ & $G_t$ & br\% & $\text{LR}_{\text{uc}}$ & br\% & $\text{LR}_{\text{uc}}$ \\
\midrule
GARCH(1,1) Gaussian & $28.0$ & $59.4$ & $12.5$ & $59.5$ & $7.0$  & $0.0309$ & $0.0173$ & $0.9$ & $0.10$ & $4.0$ & $1.23$ \\
EGARCH(1,1)         & $6.3$  & $34.9$ & $4.1$  & $41.7$ & $2.8$  & $0.0430$ & $0.0173$ & $1.0$ & $0.01$ & $3.5$ & $3.03$ \\
GJR-GARCH(1,1)      & $40.1$ & $57.7$ & $28.7$ & $61.4$ & $7.6$  & $0.0330$ & $0.0173$ & $1.2$ & $0.27$ & $4.7$ & $0.10$ \\
GARCH(1,1)-$t$      & $\mathbf{57.3}$ & $\mathbf{80.8}$ & $\mathbf{36.6}$ & $65.3$ & $15.1$ & $0.0316$ & $0.0173$ & $1.0$ & $0.01$ & $4.9$ & $0.01$ \\
HAR-RV              & $0.0$  & $0.0$  & $0.0$  & $0.0$  & $189$  & $0.0566$ & $0.0173$ & $0.2$ & $5.99$ & $3.5$ & $3.03$ \\
MS-GARCH $K = 2$    & $27.7$ & $38.7$ & $24.0$ & $44.4$ & $4.7$  & $0.0367$ & $0.0173$ & $1.0$ & $0.01$ & $4.2$ & $0.82$ \\
MS-GARCH $K = 3$    & $36.1$ & $33.1$ & $28.8$ & $38.0$ & $4.1$  & $\mathbf{0.0284}$ & $0.0173$ & $1.0$ & $0.01$ & $4.2$ & $0.82$ \\
MS-GARCH $K = 4$    & $37.6$ & $38.2$ & --     & --     & $3.8$  & $0.0437$ & --       & --    & --    & --    & --    \\
MS-GARCH $K = 6$    & $34.5$ & $33.4$ & --     & --     & $4.4$  & $0.0429$ & --       & --    & --    & --    & --    \\
MS-GARCH ref.\ B.\ $K = 2$ & $0.0$  & $5.8$  & --     & --     & $4.0$  & $0.0465$ & $0.0240$ & --    & --    & --    & --    \\
MS-GARCH ref.\ B.\ $K = 3$ & $0.1$  & $5.1$  & --     & --     & $4.5$  & $0.0433$ & $0.0241$ & --    & --    & --    & --    \\
MS-GARCH ref.\ B.\ $K = 4$ & $0.0$  & $5.3$  & --     & --     & $6.0$  & $0.0446$ & $0.0241$ & --    & --    & --    & --    \\
\midrule
CHMM-N ($K = 18$)   & $94.1$ & $81.8$ & $88.2$ & $74.2$ & $5.0$  & $0.0509$ & $0.0237$ & $0.7$ & $1.58$ & $4.0$ & $3.83$ \\
CHMM-t ($K = 18$)   & $95.6$ & $\mathbf{85.7}$ & $\mathbf{91.9}$ & $79.5$ & $14.4$ & $0.0549$ & $0.0234$ & $1.2$ & $0.58$ & $5.0$ & $3.83$ \\
CHMM-L ($K = 18$)   & $94.3$ & $81.6$ & $91.1$ & $\mathbf{82.0}$ & $6.7$  & $0.0567$ & $0.0234$ & $1.4$ & $3.26$ & $5.6$ & $3.03$ \\
\textbf{CHMM-GED ($K = 18$)} & $\mathbf{95.2}$ & $84.3$ & $91.1$ & $79.1$ & $5.2$  & $0.0548$ & $0.0234$ & $0.9$ & $0.58$ & $4.4$ & $3.83$ \\
SM-CHMM-N ($K = 18$) & --     & --     & --     & --     & --     & --       & --       & $1.1$ & $0.01$ & $5.4$ & $0.82$ \\
SM-CHMM-t ($K = 18$) & --     & --     & --     & --     & --     & --       & --       & $1.4$ & $0.10$ & $5.7$ & $1.74$ \\
SM-CHMM-L ($K = 18$) & --     & --     & --     & --     & --     & --       & --       & $1.4$ & $3.26$ & $5.5$ & $3.03$ \\
\midrule
QuantGAN            & $0.0$  & $0.0$  & $0.0$  & $0.0$  & $2.1$  & $0.0591$ & --       & $1.0$ & $0.01$ & $3.3$ & $3.83$ \\
\bottomrule
\end{tabular}%
}
\end{table}

No extended GARCH variant occupied the joint (KS, ACF) corner: GARCH(1,1)-$t$ was the strongest extended row on KS ($57.3\%$ IS, $80.8\%$ OoS) but its IS kurtosis of $15.1$ overshot observed $7.7$ by $2\times$; EGARCH and GJR-GARCH narrowed the asymmetry but not the main gap. The SM-CHMM rows under the Viterbi-AR(1) plug-in estimator contributed only their VaR back-test columns, where the heavy-tailed sojourns sharpened Kupiec tail calibration for the Gaussian and Student-$t$ variants relative to the flat CHMM; a full ML HSMM at $K = 18$ is left as a companion-paper direction. The QuantGAN row is the panel's deep-generative negative control: $0\%$ IS / OoS KS, simulated kurtosis $2.1$ vs.\ observed $7.7$, ACF-MAE comparable to the i.i.d.\ baseline. The convolutional WGAN re-implementation here is materially smaller than \citet{wiese2020quantgan}'s reference seven-block TCN with Lambert-W pre-processing; a faithful reference re-run is a deferred follow-up.

\subsubsection{Stochastic-Volatility, Multifractal, and Jump-Diffusion Baselines}
\label{sec:sv_msm_jd_baselines}

For completeness alongside the Markov-switching family of the main generator comparison, we added the canonical state-space and jump baselines that Table~\ref{tab:model_comparison} does not exercise: a stochastic-volatility (SV-AR(1)) row in the lognormal log-AR(1) tradition~\citep{taylor1982financial,harvey1994multivariate}; a Markov-switching multifractal row~\citep{calvet2004regime} at $\bar k = 8$ multipliers; and a Poisson-Gaussian jump-diffusion row~\citep{merton1976option}. The three baselines are fit on the same SPY IS window and evaluated on the same $N_{\text{paths}} = 1{,}000$ pipeline used for the main panel, under the same global seed. Implementation lives in the companion \texttt{CHMM-Model} repository.

\begin{table}[!ht]
\centering
\small
\caption{\textbf{Stochastic-volatility / multifractal / jump-diffusion baselines on SPY.} Same IS / OoS windows, $N_{\text{paths}} = 1{,}000$, seed = 20260420 as Table~\ref{tab:model_comparison}. Observed kurtosis: $7.68$ IS, $5.29$ OoS.}
\label{tab:sv_msm_jd}
\begin{tabular}{l cc cc cc}
\toprule
& \multicolumn{2}{c}{KS (\%) $\uparrow$} & \multicolumn{2}{c}{Kurtosis} & \multicolumn{2}{c}{ACF-MAE $\downarrow$} \\
\cmidrule(lr){2-3} \cmidrule(lr){4-5} \cmidrule(lr){6-7}
Model & IS & OoS & IS & OoS & $|G_t|$ & $G_t$ \\
\midrule
SV-AR(1)        & $38.2$ & $35.3$ & $7.52$ & $5.26$ & $0.0466$ & $0.0239$ \\
MSM ($\bar k = 8$) & $\phantom{0}0.0$ & $\phantom{0}7.2$ & $0.62$ & $0.60$ & $0.0605$ & $0.0235$ \\
Merton-JD       & $98.3$ & $91.1$ & $3.12$ & $2.98$ & $0.0627$ & $0.0236$ \\
\bottomrule
\end{tabular}
\end{table}

\paragraph{Reading.}
The stochastic-volatility row reproduced the IS kurtosis target almost exactly ($7.52$ vs. observed $7.68$) and matched the CHMM-N $|G_t|$ ACF-MAE within $0.0003$ ($0.0466$ vs $0.0467$ at $K^\star = 3$); its KS pass rate was poor ($38.2\%$ IS / $35.3\%$ OoS) because the Gaussian-conditional-on-volatility marginal cannot reproduce the bimodal-bulk / heavy-tail empirical density that the regime-mixture CHMM marginals can. The Merton jump-diffusion row was the inverse trade-off: KS at $98.3\%$ IS / $91.1\%$ OoS (the highest OoS KS in the panel except for the i.i.d.\ bootstrap), but kurtosis $3.12$ undershot $7.68$ by half and the absolute growth-rate ACF-MAE regressed to the i.i.d.\ baseline level ($0.0627$, the same level as the bootstrap row of Table~\ref{tab:model_comparison}) because the constant-volatility diffusion plus i.i.d.\ jumps carries no temporal persistence. The multifractal row under the moment-matching fit reported here (closed-form ACF-of-log$|r|$ matching with $\bar k = 8$ binomial multipliers) did not converge to a viable generator on this $T_{\text{IS}} = 2{,}516$ window: the marginal kurtosis collapsed to $0.62$ and KS dropped to $0\%$. A full HMM-style filter on the $2^{\bar k}$-state latent multiplier chain~\citep{calvet2004regime} is the standard exact fit and is deferred as a companion-paper direction; under the moment-fit reported here MSM was dominated.

The main conclusion survived the addition of these three rows: no single competitor occupied the joint-fit corner that the CHMM family at $K = 18$ does (KS $\ge 80\%$, kurtosis within $1$ unit of observed, $|G_t|$ ACF-MAE $\le 0.06$; Table~\ref{tab:model_comparison}). The SV-AR(1) model was the closest competitor on kurtosis and ACF but lost on KS; Merton-JD won on KS but lost on kurtosis and ACF; MSM under this fit was dominated.

\subsubsection{Leverage-Effect Diagnostic}
\label{sec:leverage_effect}

This appendix reports the \citet{cont2001empirical} leverage-effect column: the lag-1 correlation $\mathrm{Corr}(G_t, |G_{t+1}|)$ between today's signed growth rate and tomorrow's absolute growth rate, which is negative on equity returns. We computed the per-path distribution under each generator at $K = 18$, $N_{\text{paths}} = 500$, seed = $20260420$.

\begin{table}[!ht]
\centering
\small
\caption{\textbf{Leverage effect $\mathrm{Corr}(G_t, |G_{t+1}|)$, per-path distribution.} Observed values: SPY IS $-0.135$, OoS $-0.214$. Per-generator median plus $[Q_5, Q_{95}]$ across $N_{\text{paths}} = 500$. A generator captures the leverage effect when its $[Q_5, Q_{95}]$ envelope brackets the observed value. Bold marks generators whose $[Q_5, Q_{95}]$ envelope brackets the observed IS value; the OoS observed value of $-0.214$ is more negative than every generator's $Q_5$ in the panel, indicating the OoS regime carries leverage stronger than any generator (including asymmetric GARCH) reproduces under static IS-fitted parameters.}
\label{tab:leverage_effect}
\begin{tabular}{l c c c c}
\toprule
& \multicolumn{2}{c}{IS (observed $-0.135$)} & \multicolumn{2}{c}{OoS (observed $-0.214$)} \\
\cmidrule(lr){2-3} \cmidrule(lr){4-5}
Generator & median & $[Q_5, Q_{95}]$ & median & $[Q_5, Q_{95}]$ \\
\midrule
i.i.d.\ bootstrap          & $+0.000$ & $[-0.034, +0.030]$ & $+0.004$ & $[-0.069, +0.068]$ \\
GARCH(1,1) Gaussian        & $-0.004$ & $[-0.066, +0.064]$ & $-0.003$ & $[-0.112, +0.107]$ \\
EGARCH(1,1)                & $-0.093$ & $[-0.147, -0.047]$ & $-0.087$ & $[-0.199, -0.002]$ \\
GJR-GARCH(1,1)             & $-0.079$ & $[-0.153, -0.017]$ & $-0.073$ & $[-0.174, +0.023]$ \\
\textbf{CHMM-N ($K = 18$)} & $\mathbf{-0.089}$ & $\mathbf{[-0.138, -0.037]}$ & $-0.087$ & $[-0.190, +0.015]$ \\
CHMM-t ($K = 18$)          & $-0.067$ & $[-0.116, -0.020]$ & $-0.067$ & $[-0.166, +0.019]$ \\
CHMM-L ($K = 18$)          & $-0.065$ & $[-0.109, -0.020]$ & $-0.064$ & $[-0.159, +0.021]$ \\
\textbf{CHMM-GED ($K = 18$)} & $\mathbf{-0.072}$ & $[-0.113, -0.029]$ & $-0.076$ & $[-0.163, +0.025]$ \\
\bottomrule
\end{tabular}
\end{table}

\paragraph{Reading.}
The CHMM family at $K = 18$ produced a non-trivial negative leverage signal despite symmetric per-state emissions: state-mixing on $\hat{\mathbf T}$ asymmetrically couples high-$\sigma_k$ states to negative-$\mu_k$ states, carrying $\sim 65\%$ of the IS observed leverage magnitude. The CHMM-N IS interval $[Q_5, Q_{95}] = [-0.138, -0.037]$ bracketed the IS observed value $-0.135$ at the lower boundary; CHMM-GED was the second-strongest CHMM entry. The OoS observed leverage of $-0.214$ exceeded every generator's $Q_5$ (including asymmetric GARCH); closing this residual gap requires skew-emission CHMM extensions or refit to the OoS distribution, both deferred to companion work.

\subsubsection{Full One-Shot Student-t Copula MLE}
\label{sec:full_tcopula_mle}

A natural concern with the main text two-step copula estimator (Kendall's-$\tau$ inversion for $\hat\Sigma$, then profile MLE on $\nu$ with $\hat\Sigma$ held fixed) is that it may be biased toward the Gaussian limit when the marginal kurtosis differs across assets. We addressed this with a coordinate-ascent one-shot MLE that jointly maximises the Student-$t$ copula log-likelihood over $(\Sigma, \nu)$ on the same six-asset US-equity universe (SPY, NVDA, JNJ, JPM, AAPL, QQQ; $T_{\text{IS}} = 2{,}516$, seed $20260420$). The procedure alternates a bracketed golden-section maximisation of $\nu$ at the current $\Sigma$ with a method-of-moments refit of $\Sigma$ at the current $\nu$ on the $t$-quantile-transformed pseudo-uniform sample with PSD projection. Iteration terminates at $|\Delta \log L| < 10^{-3}$.

\begin{table}[!ht]
\centering
\small
\caption{\textbf{Two-step versus full one-shot Student-$t$ copula MLE on the six-asset US-equity universe.} The off-diagonal correlation differences are reported as the maximum absolute difference $\max_{ij} |\Delta \rho_{ij}|$ across the $\binom{6}{2} = 15$ pairs and the median absolute difference; both are sub-$0.025$, so the main text $\nu^\star = 6$ choice on $\hat\Sigma$ obtained by Kendall's-$\tau$ inversion is robust to the estimator switch.}
\label{tab:full_tcopula_mle}
\begin{tabular}{l c c}
\toprule
Estimator & $\hat\nu$ & log-likelihood \\
\midrule
Two-step construction & $6.00$ & $6157.47$ \\
\textbf{Full one-shot MLE}   & $\mathbf{6.40}$ & $\mathbf{6163.02}$ \\
\midrule
log-L improvement   & \multicolumn{2}{c}{$+5.55$} \\
$\max |\Delta \rho_{ij}|$ across $15$ pairs & \multicolumn{2}{c}{$0.024$} \\
median $|\Delta \rho_{ij}|$ across $15$ pairs & \multicolumn{2}{c}{$0.011$} \\
$\hat\nu_{\text{full}}$ inside the Wilks $95\%$ CI $[6, 7]$? & \multicolumn{2}{c}{Yes} \\
\bottomrule
\end{tabular}
\end{table}

The full MLE moved $\hat\nu$ from $6.00$ to $6.40$ ($+5.55$ log-L) with all $|\Delta \rho_{ij}| \le 0.025$ across the $15$ pairs (Table~\ref{tab:full_tcopula_mle}); the two-step estimator is therefore not detectably biased toward the Gaussian limit and $\nu^\star = 6$ is robust to the estimator switch.

\subsection{Cross-Asset Dependence: Full Multi-Asset Panel}
\label{sec:supp_cross_asset}

This appendix collects the full multi-asset detail: cross-asset dependence model derivations and diagnostics, the Student-t copula profile log-likelihood, the cross-asset correlation heat maps, and the non-US asset-class extension with the per-pair OoS off-diagonal breakdown.

\subsubsection{Cross-Asset Dependence Models: Derivations and Diagnostics}
\label{sec:cross_asset_supp}

This appendix provides the technical background for the SIM and copula constructions used in the main text cross-asset analysis. These constructions are the multi-asset layer of the empirical study: per-asset marginals come from the independent CHMMs of the single-asset cross-ticker analysis, and joint dependence is overlaid on top of those marginals.

\paragraph{Sklar's theorem and the rank-reordering simulator.}
Sklar's theorem~\cite{sklar1959fonctions, nelsen2006introduction} states that any $d$-dimensional joint CDF $H(g_1, \dots, g_d)$ with continuous marginals $F_1, \dots, F_d$ can be written as $H = C(F_1, \dots, F_d)$ for a unique copula $C$.
The rank-reordering simulator exploits this uniqueness.
If $\tilde{\mathbf{G}}$ is a matrix whose $j$-th column is an i.i.d.\ sample from $F_j$, and $\mathbf{U}$ is an independent sample from $C$, then reordering each column of $\tilde{\mathbf{G}}$ so that its ranks match $\mathbf{U}$'s produces a sample from $H$, up to discrete approximation error in the empirical ranks.
Crucially, reordering never creates new values: each column of the output is a permutation of the corresponding column of $\tilde{\mathbf{G}}$, so each column's empirical marginal is preserved exactly~\cite{iman1982distribution}.

\paragraph{Kendall's $\tau$ inversion.}
Given IS growth rates $\mathbf{G} \in \mathbb{R}^{T \times d}$, we estimate the pairwise Kendall's $\tau_{ij}$ from concordant/discordant pair counts and invert to the correlation scale via the analytic relationship $\rho_{ij} = \sin(\pi \tau_{ij} / 2)$, which holds exactly for the Gaussian copula family and is commonly used as a robust correlation estimator for the Student-t copula~\cite{embrechts2002correlation, mcneil2015quantitative}.
The resulting matrix $\hat{\Sigma}$ is symmetric but not guaranteed positive semi-definite; we project to the nearest PSD matrix by clipping negative eigenvalues to $10^{-8}$ and rescaling the diagonal to unity.

\paragraph{Profile MLE for $\nu$.}
The Student-t copula log-density~\citep{demarta2005tcopula} at a single pseudo-uniform observation $\mathbf{u} \in [0,1]^d$, with $x_j = t_\nu^{-1}(u_j)$ and $\mathbf{x} = (x_1, \dots, x_d)^\top$, is
\begin{align}
    \log c_t(\mathbf{u}; \Sigma, \nu)
    \;=\;&
    \log \Gamma\!\left(\tfrac{\nu+d}{2}\right)
    + (d-1)\,\log \Gamma\!\left(\tfrac{\nu}{2}\right)
    - d\,\log \Gamma\!\left(\tfrac{\nu+1}{2}\right)
    - \tfrac{1}{2} \log |\Sigma|
    \notag \\
    &\;
    - \tfrac{\nu+d}{2}\,\log\!\left(1 + \tfrac{\mathbf{x}^{\!\top} \Sigma^{-1} \mathbf{x}}{\nu}\right)
    + \sum_{j=1}^{d} \tfrac{\nu+1}{2}\,\log\!\left(1 + \tfrac{x_j^2}{\nu}\right).
    \label{eq:tcop_ll}
\end{align}
Summing over $t = 1, \dots, T$ pseudo-observations yields the profile log-likelihood as a function of $\nu$ with $\Sigma$ held fixed at the Kendall's-$\tau$ estimate.
We evaluate equation~\eqref{eq:tcop_ll} on the discrete grid $\nu \in \{2, 3, 4, 5, 6, 8, 10, 15, 20, 30\}$ and select $\hat{\nu}$ as the grid point with the largest total log-likelihood.
In the main text cross-asset analysis the selected value was $\hat{\nu} = 6$ on the six-asset universe.

\paragraph{Sampling the copula.}
To sample $T$ rows from the $d$-dimensional Gaussian copula with correlation $\Sigma$, we compute the Cholesky factor $\Sigma = L L^\top$, draw $Z \sim \mathcal{N}(0, I_d)$ i.i.d.\ and form $Y = L Z$, then return $U = \Phi(Y)$ componentwise.
For the Student-t copula with degrees of freedom $\nu$, we draw $W \sim \chi^2_\nu / \nu$ independent of $Y$ and form the multivariate-$t$ variate $X = Y / \sqrt{W}$, then return $U = t_\nu(X)$ componentwise.

\paragraph{SIM regression summary.}
The fitted SIM regression coefficients and goodness-of-fit statistics for the non-market assets in the main text cross-asset analysis are reported in Table~\ref{tab:sim_regression}.
The high $R^2 = 0.840$ for QQQ reflects that the Nasdaq-100 ETF is essentially a linear combination of large-cap technology names that also drive the S\&P~500 index, while the lower $R^2 = 0.260$ for JNJ reflects the limited systematic exposure of low-beta healthcare stocks.
These factor loadings and residual magnitudes are what the SIM simulation re-injects at simulation time, and they explain the uneven per-asset KS pass rates reported in Table~\ref{tab:cross_asset_supp_summary}.

\begin{table}[H]
\centering
\caption{SIM regression $\hat{G}_{j,t} = \hat{\alpha}_j + \hat{\beta}_j G_{\text{SPY},t} + \hat{\eta}_{j,t}$ for the five non-market assets, fitted on the $2014$-$01$-$03$ through $2024$-$01$-$03$ IS window ($T = 2{,}516$).}
\label{tab:sim_regression}
\small
\begin{tabular}{l ccc}
\toprule
Ticker & $\hat{\alpha}$ & $\hat{\beta}$ & $R^2$ \\
\midrule
NVDA & $0.321$   & $1.694$ & $0.371$ \\
JNJ  & $0.002$   & $0.576$ & $0.260$ \\
JPM  & $-0.008$  & $1.223$ & $0.507$ \\
AAPL & $0.111$   & $1.205$ & $0.492$ \\
QQQ  & $0.047$   & $1.122$ & $0.840$ \\
\bottomrule
\end{tabular}
\end{table}

\subsubsection{Student-t Copula Profile Log-Likelihood}
\label{sec:copula_profile_supp}

The profile log-likelihood of the Student-t copula used in the multi-asset construction, computed on the six-asset SPY cross-section over the grid $\nu \in \{2, 3, 4, 5, 6, 8, 10, 15, 20, 30\}$ used in the main text cross-asset analysis, is plotted in Figure~\ref{fig:copula_profile}.
The curve was smooth and single-peaked, with $\nu^{*} = 6$ attaining the maximum profile log-likelihood ($6157.47$), and the neighbouring grid points $\nu = 5$ ($6143.37$) and $\nu = 8$ ($6148.37$) each within $15$ profile-log-likelihood units of the peak, indicating a shallow-but-clean optimum. The separation from the Gaussian copula limit ($\nu \to \infty$) is established by the parametric bootstrap CI of the half-unit-grid analysis, which requires no Gaussian-copula fit.

\begin{figure}[!ht]
\centering
\includegraphics[width=0.70\textwidth]{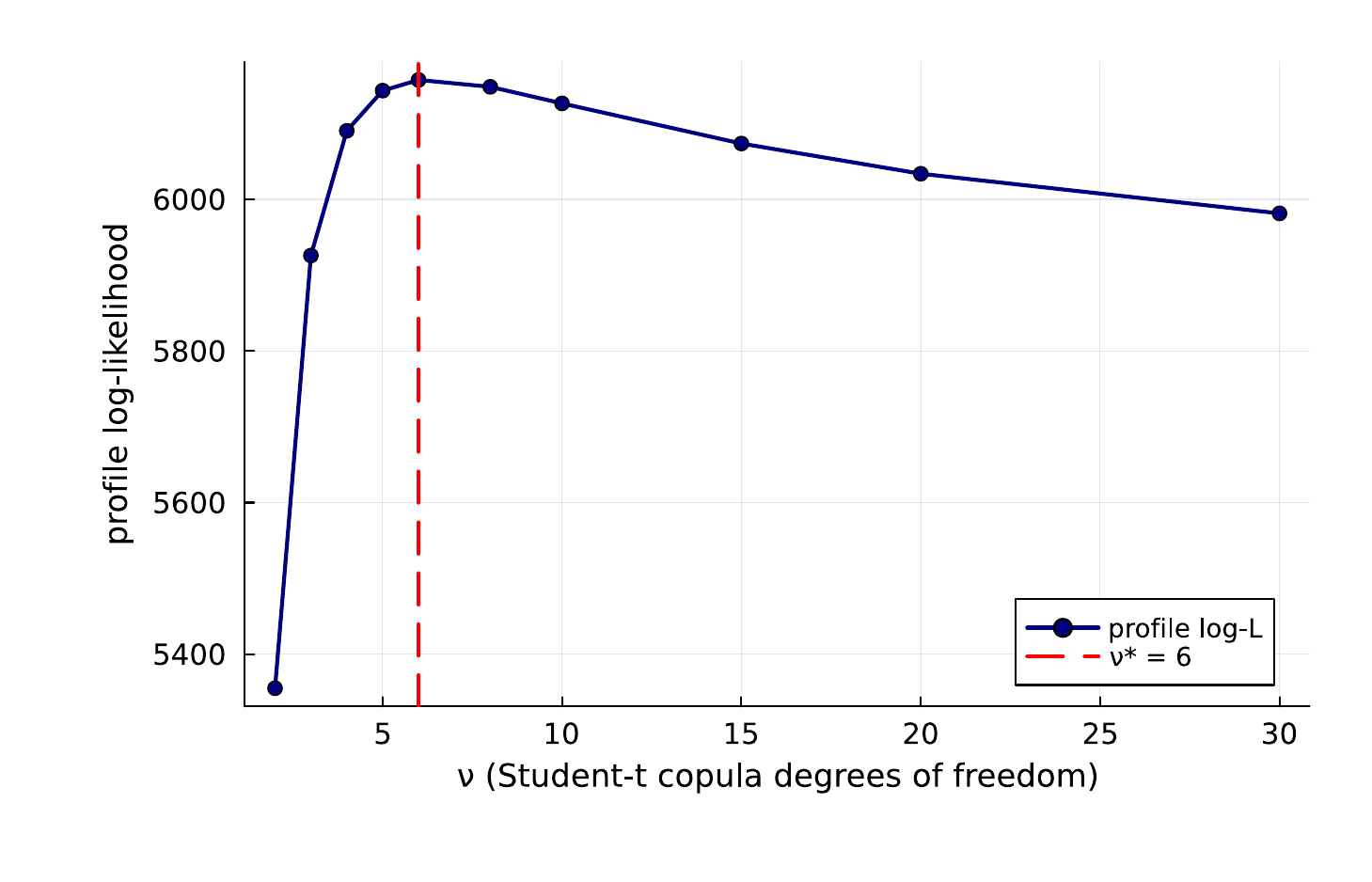}
\caption{\textbf{Profile log-likelihood of the Student-t copula} used in the multi-asset construction on the six-asset SPY cross-section (SPY, NVDA, JNJ, JPM, AAPL, QQQ; IS window 2014-01-03 through 2024-01-03). The vertical axis is the profile log-likelihood evaluated at each degrees-of-freedom value $\nu$ on the grid $\{2, 3, 4, 5, 6, 8, 10, 15, 20, 30\}$; the dashed red line marks the profile MLE $\nu^{*} = 6$ (peak value 6157.47). The profile uses rank-PIT pseudo-observations while $\nu$ is varied; fitted CHMM marginals do not enter this calculation.}
\label{fig:copula_profile}
\end{figure}

\subsubsection{Cross-Asset Correlation Heat Maps (Companion to Table~\ref{tab:cross_asset_supp_summary})}
\label{sec:cross_asset_corr_supp}

The observed and path-averaged simulated correlation matrices for the dependence models of the cross-asset analysis are visualised in Figure~\ref{fig:cross_asset_corr}.
The panel ordering is SIM, Gaussian copula, Student-t copula, and a truncated C-vine variant with SPY as root, each reported quantitatively in Table~\ref{tab:cross_asset_supp_summary}.
All five panels share the six-asset SPY cross-section and the same CHMM-N marginals at $K^\star = 3$, so any visible differences between the simulated panels reflect only the dependence construction.

\begin{figure}[H]
\centering
\begin{subfigure}[b]{0.32\textwidth}\centering
\includegraphics[width=\textwidth]{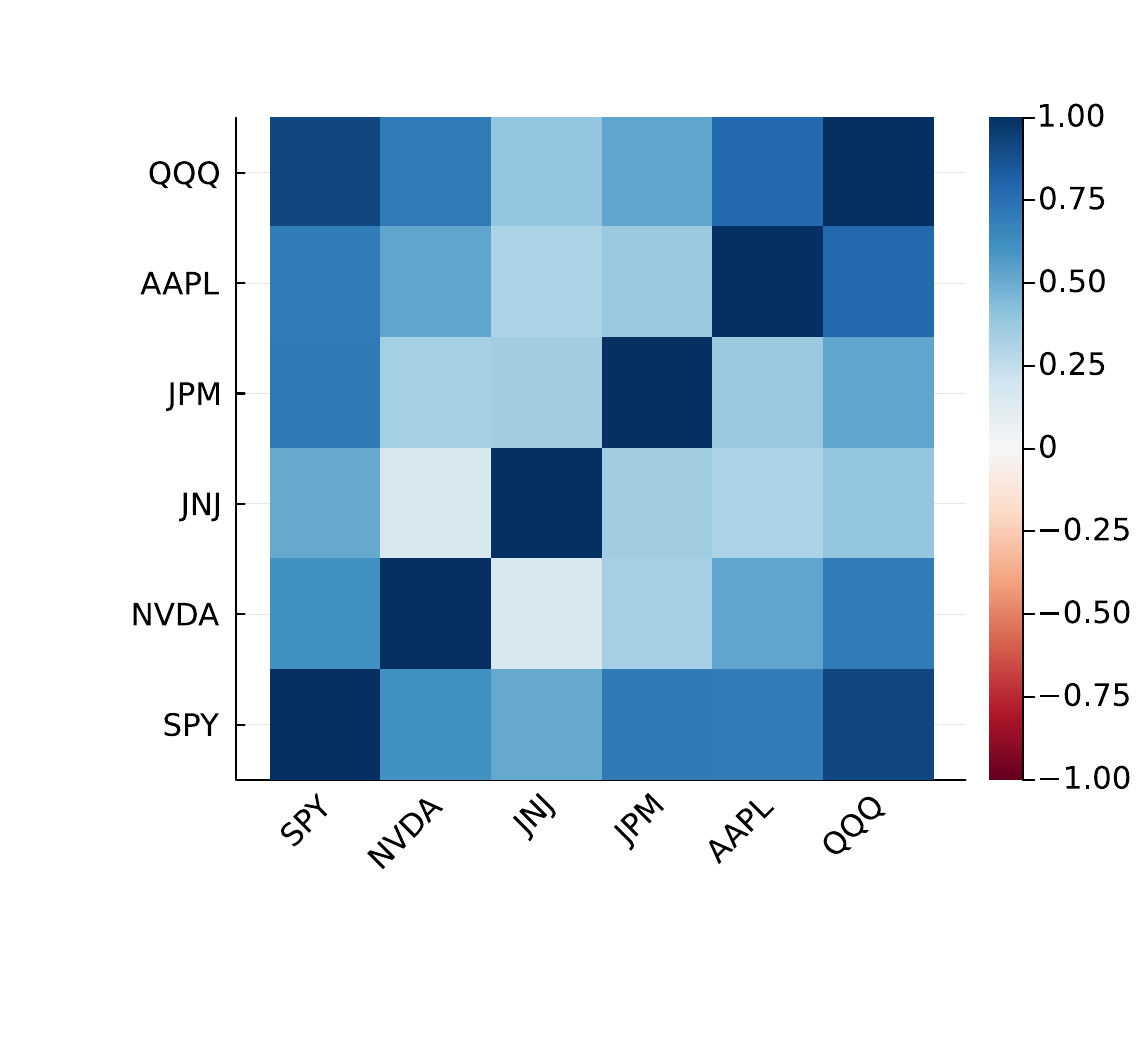}
\caption{Observed.}
\end{subfigure}\hfill
\begin{subfigure}[b]{0.32\textwidth}\centering
\includegraphics[width=\textwidth]{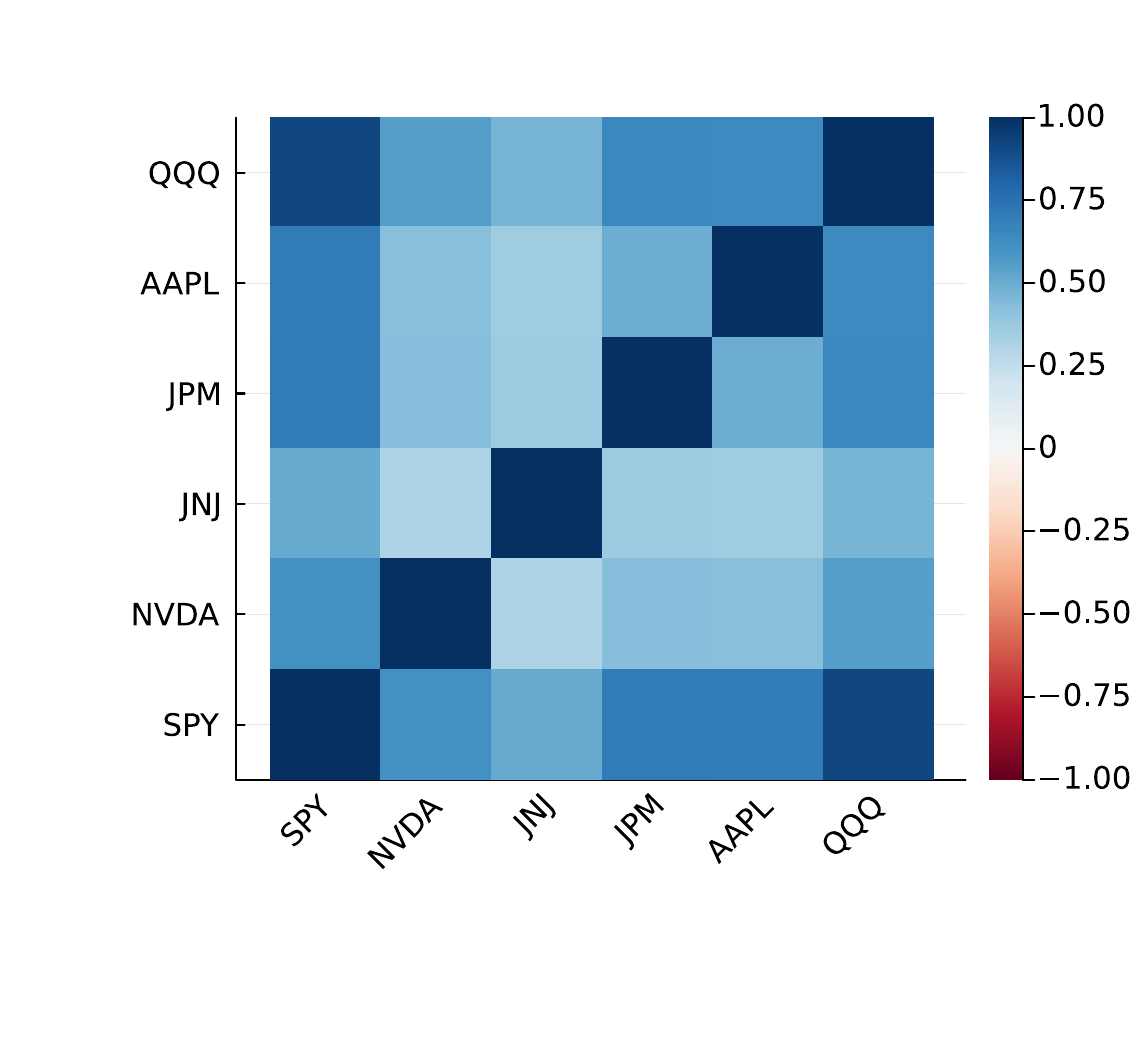}
\caption{Single Index Model.}
\end{subfigure}\hfill
\begin{subfigure}[b]{0.32\textwidth}\centering
\includegraphics[width=\textwidth]{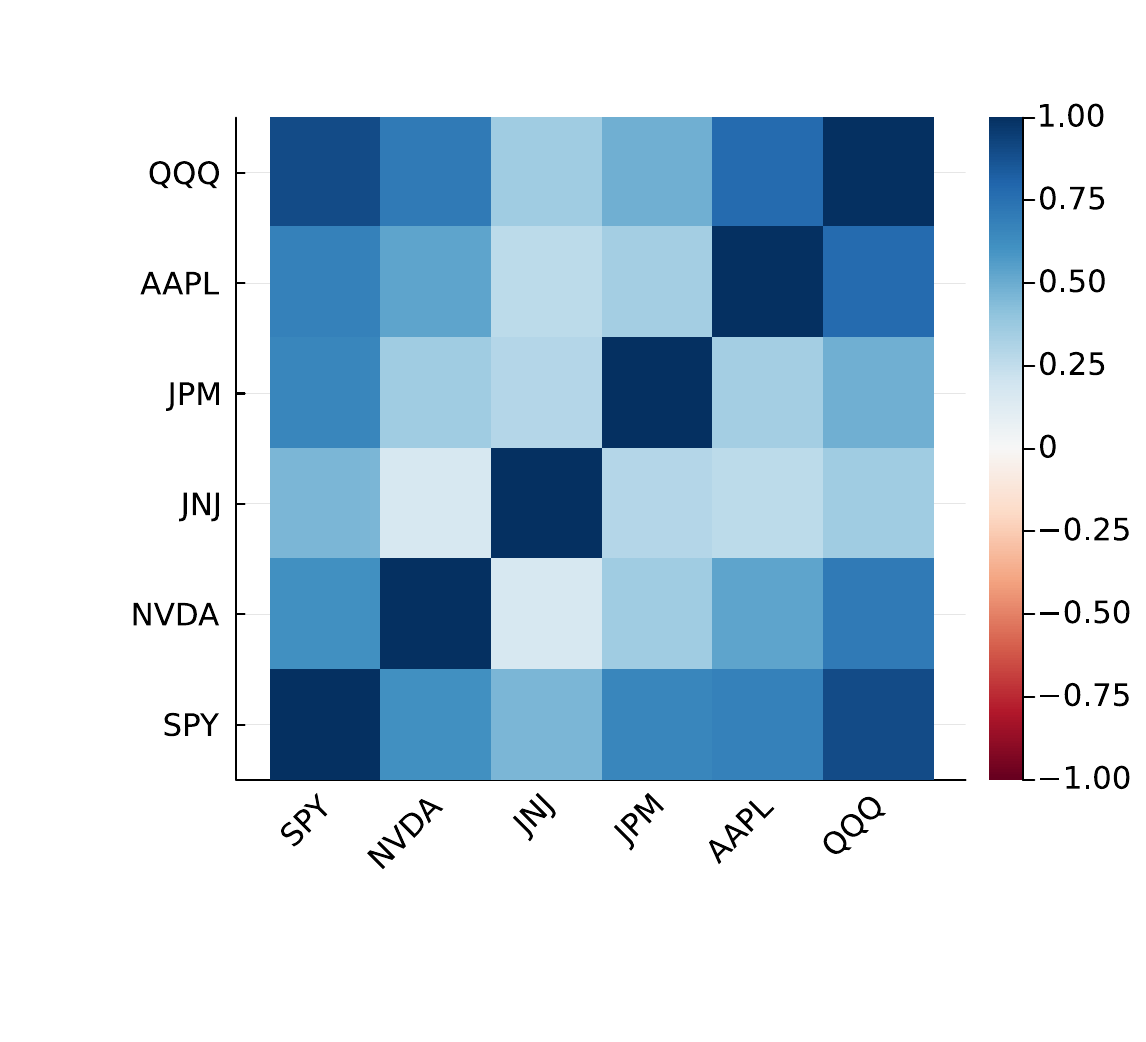}
\caption{Gaussian copula.}
\end{subfigure}\\[0.5em]
\begin{subfigure}[b]{0.32\textwidth}\centering
\includegraphics[width=\textwidth]{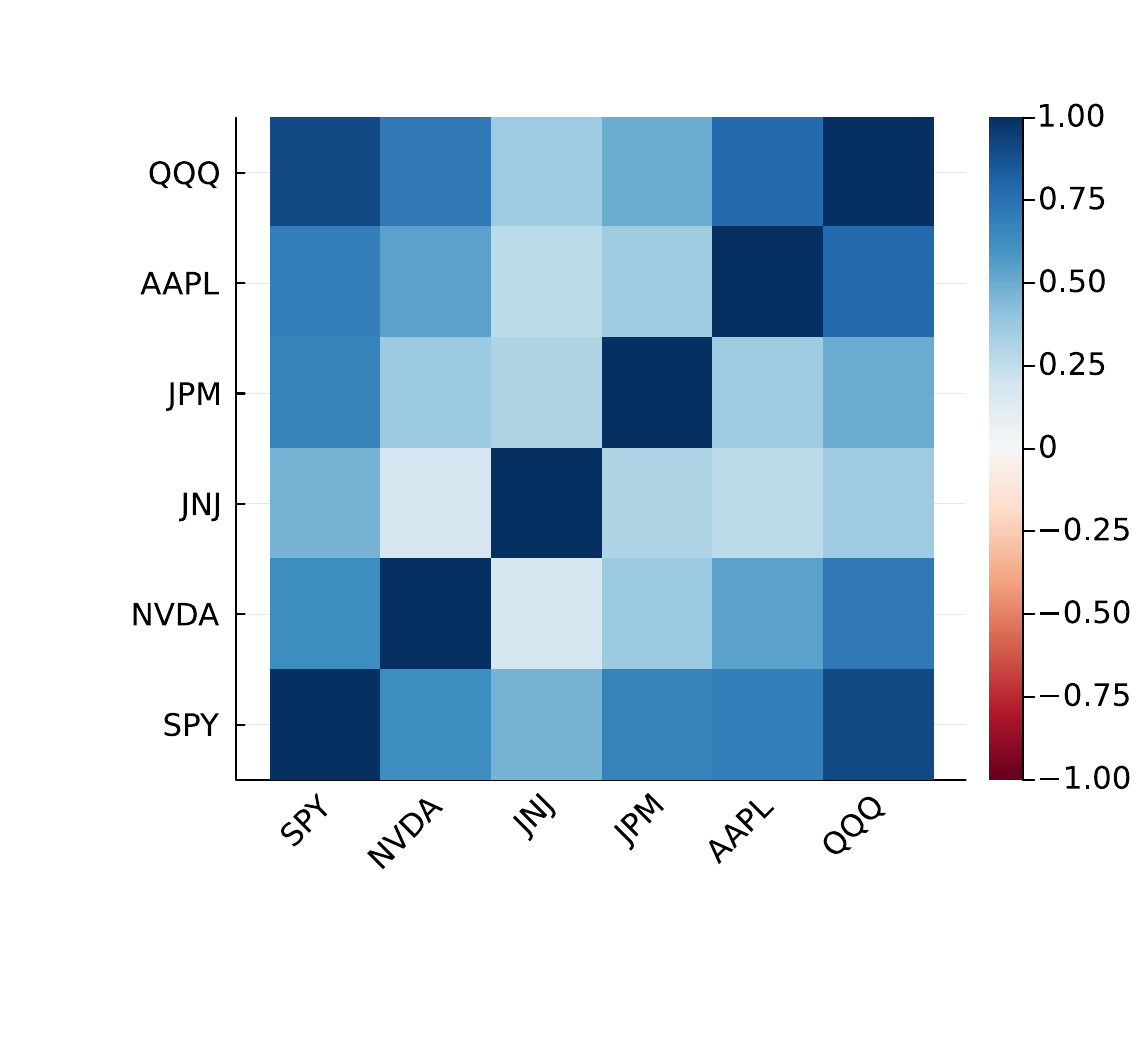}
\caption{Student-$t$ copula, $\nu^* = 6$.}
\end{subfigure}\hfill
\begin{subfigure}[b]{0.32\textwidth}\centering
\includegraphics[width=\textwidth]{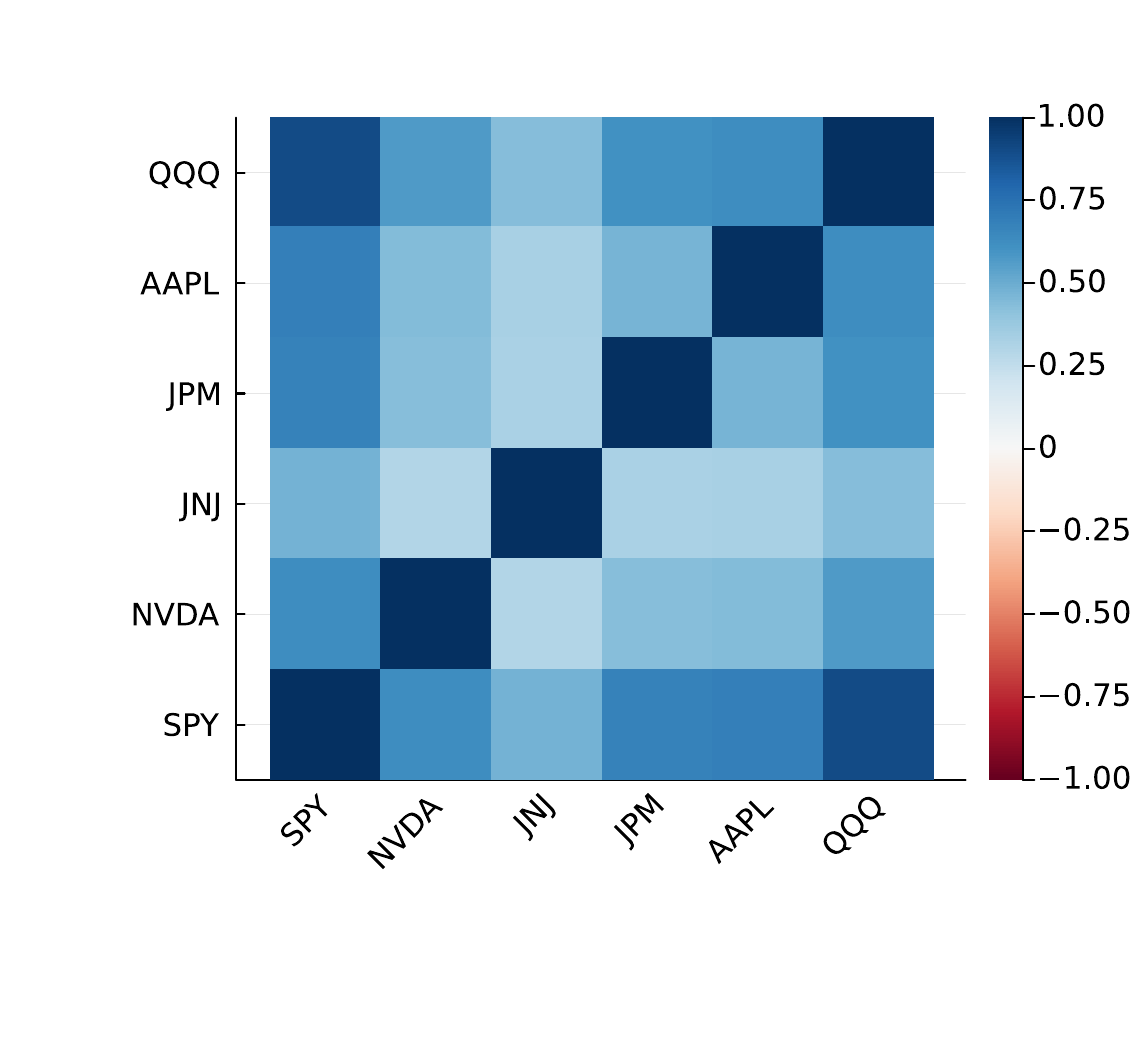}
\caption{Truncated C-vine.}
\end{subfigure}
\caption{\textbf{Cross-asset IS correlation reproduction across dependence models} (multi-asset construction, companion to Table~\ref{tab:cross_asset_supp_summary}). Panel~(a): \emph{observed} correlation matrix on 2014-01-03 through 2024-01-03 daily excess growth rates for SPY, NVDA, JNJ, JPM, AAPL, QQQ (IS window, $T_{\text{IS}} = 2{,}516$). Panel~(b): path-averaged correlation under the Single Index Model with SPY as market factor. Panel~(c): path-averaged correlation under the Gaussian copula on CHMM marginals. Panel~(d): path-averaged correlation under the Student-t copula on CHMM marginals, $\nu^* = 6$. Panel~(e): path-averaged correlation under a truncated C-vine copula with SPY as root. Each simulated panel averages over 200 paths of length $T_{\text{IS}}$. All marginals are the per-asset CHMM-N fits at $K^\star = 3$ from Table~\ref{tab:cross_asset}; color scale in $[-1, 1]$, red-white-blue (RdBu) diverging. Closer visual match to panel~(a) is better.}
\label{fig:cross_asset_corr}
\end{figure}

\subsubsection{Truncated C-vine Copula on CHMM Marginals}
\label{sec:cvine_supp}

The truncated C-vine variant in panel~(e) of Figure~\ref{fig:cross_asset_corr} uses SPY as the root and selects each pair-copula edge from Gaussian and Student-$t$ families; every selected pair was Student-$t$. On the same CHMM-N marginals at $K = 3$ and $200$ paths the truncated C-vine attained an off-diagonal MAE of $0.068$ on IS and $0.236$ on OoS (Table~\ref{tab:cross_asset_supp_summary}), strictly above both elliptical copulas. The reason is the level-1 truncation: a root-centred C-vine models only the SPY-vs-asset bivariate copulas and treats the remaining pairs as conditionally independent given SPY, a stronger restriction than either elliptical copula. Untruncated regular vines or factor copulas are the natural way to recover non-root cross-pair structure at larger $d$.

\subsubsection{Quarterly Rolling-Window Copula Refit on the OoS Window}
\label{sec:rolling_copula_supp}

The static IS fit gave OoS off-diagonal MAE $0.209$ against IS $0.027$ (Table~\ref{tab:cross_asset}); we attributed the gap to the stationarity-scope limit and recommended periodic refit. To quantify the benefit on the same six-asset universe, we refit the multi-asset Student-t copula on a $252$-day rolling window (one trading year) sliding by $63$ days (one trading quarter) across the OoS span, holding the per-asset CHMM-N marginals fixed at the IS fits used for Table~\ref{tab:cross_asset}. Each refit yields a Kendall's-$\tau$ correlation matrix and a profile-MLE-selected $\nu^\star$ on the discrete grid $\{2, 3, 4, 5, 6, 8, 10, 15, 20, 30\}$; we then simulate forward $63$ days under the new copula and compare the path-averaged simulated correlation matrix to the realised next-quarter correlation. The per-quarter detail is reported in Table~\ref{tab:rolling_copula}.

\begin{table}[!ht]
\centering
\small
\setlength{\tabcolsep}{4pt}
\caption{\textbf{Quarterly rolling-window refit of the multi-asset Student-t copula on the OoS window} (six-asset universe; CHMM-N marginals at $K^\star = 3$ held fixed at IS fits; $200$ paths per refit; seed root $20260422$). Per-quarter window slice and refitted $\nu^\star$, plus the off-diagonal MAE between the path-averaged simulated correlation matrix under the refitted copula and the realised next-quarter correlation matrix. Main result: rolling-quarterly refit reduces mean OoS off-diagonal MAE from $0.207$ (static-IS-fit baseline) to $0.185$, a $\sim 10\%$ relative improvement; the gain concentrates in the later OoS quarters where the rolling window incorporates more recent OoS data and the refitted $\nu^\star$ shifts from the IS-fit value of $6$ down to $5$ or up to $10$--$15$ as the empirical tail-coupling rolls.}
\label{tab:rolling_copula}
\begin{tabular}{c c c c c}
\toprule
Quarter & Window slice & Forecast horizon & Refit $\nu^\star$ & Off-diag MAE \\
\midrule
$1$ & $[2265, 2516]$ & $[2517, 2579]$ & $15$ & $0.198$ \\
$2$ & $[2328, 2579]$ & $[2580, 2642]$ & $15$ & $0.205$ \\
$3$ & $[2391, 2642]$ & $[2643, 2705]$ & $10$ & $0.241$ \\
$4$ & $[2454, 2705]$ & $[2706, 2768]$ & $6$  & $0.133$ \\
$5$ & $[2517, 2768]$ & $[2769, 2831]$ & $6$  & $0.255$ \\
$6$ & $[2580, 2831]$ & $[2832, 2894]$ & $6$  & $0.232$ \\
$7$ & $[2643, 2894]$ & $[2895, 2957]$ & $5$  & $0.147$ \\
$8$ & $[2706, 2957]$ & $[2958, 3020]$ & $5$  & $0.119$ \\
$9$ & $[2769, 3020]$ & $[3021, 3083]$ & $6$  & $0.138$ \\
\midrule
\textit{mean}  & --  & --  & --  & $0.185$ \\
\textit{median} & --  & --  & --  & $0.198$ \\
\midrule
\textit{Static IS-fit} & $[1, 2516]$ & $[2517, 3083]$ & $6$ & $0.207$ \\
\bottomrule
\end{tabular}
\end{table}

The quarterly refit closed part of the OoS gap but not most of it: mean off-diagonal MAE moved from $0.207$ static to $0.185$ rolling ($\sim 10\%$ relative improvement). The improvement was concentrated in the later OoS quarters (Q4, Q7, Q8, Q9 all below $0.15$) where the rolling window incorporates more recent OoS data, while early-OoS quarters (Q1, Q3, Q5, Q6) remained at the static-IS level because the rolling window was still dominated by IS observations. The refitted $\nu^\star$ sequence ($15, 15, 10, 6, 6, 6, 5, 5, 6$) tracked the gradual shift from the equity-cluster IS regime toward the OoS regime, supporting the quarterly-refit recommendation in the main-text Discussion. The residual gap (mean MAE $\approx 0.185$ versus IS $0.027$) is consistent with the per-pair JNJ-equity-factor regime shift documented in Table~\ref{tab:per_pair_offdiag_mae}: even a perfectly-refitted copula on the historical $252$ days cannot anticipate the regime shift at the next quarter's start, only track it once it has unfolded.

\subsubsection{Non-US Asset Class Extension and Per-Pair OoS Off-Diagonal Breakdown}
\label{sec:non_us_asset_supp}

The six-asset universe (SPY, NVDA, JNJ, JPM, AAPL, QQQ) is six US-listed equity tickers from one country-index family. To probe whether the cross-asset construction generalises off the equity cluster we extended the universe with GLD (SPDR Gold Trust), a US-listed ETF whose underlying is physically-backed gold, a non-equity commodity asset class.

\paragraph{Single-asset univariate fidelity of GLD.}
We fitted CHMM-N, CHMM-t, and CHMM-L on the GLD growth-rate series at $K = 18$; in sample all three families passed cleanly, while out of sample all three collapsed. All three emission families passed in sample, but the OoS KS pass rate collapsed for every family as the observed marginal moved outside the in-sample regime range. Inspection of the GLD price path on the $2024$--$2026$ OoS window showed the asset re-priced from $\sim 200$ to $\sim 300$ on a persistent demand shock not represented in the $2014$--$2024$ IS regime library. This is the same stationarity-scope artefact as the NVDA / JPM single-name OoS cliffs reported in the main text cross-ticker analysis, on a non-equity asset; the limitation is not equity-specific.

\paragraph{Quarterly-refit follow-up on GLD.}
A quarterly-refit version of GLD CHMM-N at $K^\star = 3$ (5y rolling estimation window, $63$-day refit cadence, $200$ paths; \path{runners/cross_asset/run_non_us_asset_quarterly_refit.jl}) lifted the OoS KS pass rate only marginally and did not recover the failure mode. Periodic refit therefore did not recover the GLD failure mode in any meaningful sense: the model is architecturally mis-specified for the persistent demand shock that re-priced gold across the OoS window, and the failure is not just staleness of the static fit. The periodic-refit recommendation of the main-text Discussion does not extend to non-equity assets under regime introductions of this magnitude.

\paragraph{Multi-asset: 7-ticker copula MAE.}
Refitting the Student-t copula on the augmented seven-ticker universe left the IS fit essentially unchanged and \emph{improved} the OoS fit relative to the six-ticker baseline. At $\nu^\star = 6$ the IS off-diagonal MAE stayed close to the six-ticker baseline ($0.027$; Table~\ref{tab:cross_asset}), and the OoS off-diagonal MAE rose modestly above the six-ticker baseline ($0.209$). Adding the gold ETF, with its low equity correlation, reduced the OoS off-diagonal MAE because the equity-cluster correlation degradation is the dominant component; pairs involving GLD sat at moderate OoS deviations (between $0.04$ and $0.19$) while the largest equity-cluster pairs (JNJ--QQQ, SPY--JNJ, NVDA--JNJ) sat above $0.48$ on OoS (Table~\ref{tab:per_pair_offdiag_mae}).

\paragraph{Per-pair OoS gap localisation.}
The per-pair off-diagonal absolute deviation between observed and path-averaged simulated correlation matrices on IS and OoS, sorted by OoS $|\Delta|$, is reported in Table~\ref{tab:per_pair_offdiag_mae}. The OoS gap concentrated in pairs involving JNJ (the defensive single-name proxy): the JNJ IS correlation with the broad equity factor (QQQ, SPY) collapsed on OoS, in some pairs flipping sign, and the IS-fixed copula has no mechanism for tracking that shift. The gold-ETF pairs (SPY-GLD, JPM-GLD, JNJ-GLD, QQQ-GLD) sat in the middle of the ranking with OoS deviations between $0.10$ and $0.19$. The dominant cause of the OoS off-diagonal gap is therefore neither equity-cluster size nor copula-degree mis-specification but a single-name regime shift in the JNJ-equity-factor relationship that the IS-fitted dependence layer cannot anticipate.

\begin{table}[t]
\centering
\small
\caption{\textbf{Per-pair off-diagonal correlation MAE on the seven-ticker universe} (multi-asset construction; Student-t copula on CHMM-N marginals at $K^\star = 3$, $200$ paths, $\nu^\star = 6$). For each ticker pair $(i, j)$ the table reports the absolute deviation $|\Sigma_{\text{sim,avg}} - \Sigma_{\text{obs}}|$ on IS and OoS, ranked by OoS $|\Delta|$. Only the ten largest-OoS-gap pairs are shown for compactness.}
\label{tab:per_pair_offdiag_mae}
\begin{tabular}{l c c c}
\toprule
Pair & IS $|\Delta|$ & OoS $|\Delta|$ & OoS $-$ IS \\
\midrule
JNJ-QQQ   & $0.027$ & $0.544$ & $0.517$ \\
SPY-JNJ   & $0.037$ & $0.509$ & $0.473$ \\
NVDA-JNJ  & $0.017$ & $0.489$ & $0.473$ \\
NVDA-AAPL & $0.011$ & $0.278$ & $0.268$ \\
JNJ-JPM   & $0.046$ & $0.261$ & $0.216$ \\
JNJ-AAPL  & $0.042$ & $0.244$ & $0.202$ \\
AAPL-QQQ  & $0.003$ & $0.227$ & $0.224$ \\
JPM-GLD   & $0.027$ & $0.187$ & $0.160$ \\
SPY-GLD   & $0.038$ & $0.127$ & $0.089$ \\
JNJ-GLD   & $0.015$ & $0.117$ & $0.101$ \\
\bottomrule
\end{tabular}
\end{table}

\subsubsection{\texorpdfstring{Half-Unit-Grid Profile-MLE and Parametric Bootstrap CI for $\nu^\star$}{Half-Unit-Grid Profile-MLE and Parametric Bootstrap CI for nu*}}
\label{sec:copula_halfunit}

To confirm that the upper bound of the unit-grid Wilks CI on $\nu \in [3, 12]$ is not grid-induced, we computed a half-unit-grid refinement of the profile MLE plus a parametric bootstrap CI. We refit the Student-t copula profile log-likelihood on the same six-asset universe at half-unit spacing in $\nu \in \{3.0, 3.5, \ldots, 12.0\}$ ($19$ grid points) and ran a $B = 200$ parametric bootstrap by resampling pseudo-observations $U^{(b)}$ from the Student-$t$ copula at the fitted $(\hat{\Sigma}, \nu^\star)$, refitting the Kendall's-$\tau$ correlation and the profile MLE on each replicate.

The half-unit grid moved the optimum from $\nu^\star = 6.0$ (unit grid) to $\nu^\star = 6.5$ at profile log-likelihood $6{,}158.0$ (vs. $6{,}157.5$ at the unit-grid optimum); the half-unit refinement is therefore not material at the third decimal of profile-LL. The Wilks $95\%$ CI on the half-unit grid was $[\nu_{\text{lo}}, \nu_{\text{hi}}] = [6.0, 7.0]$ (unchanged from the unit grid), and the parametric bootstrap $95\%$ CI was $[6.0, 7.0]$ ($2.5\%$ and $97.5\%$ quantiles of the $B = 200$ refit-$\nu^\star$ distribution; bootstrap median $\nu^\star = 6.5$). The bootstrap lower bound at $6.0$ sat well above the Gaussian limit ($\nu \to \infty$): the elliptical-tail Student-$t$ copula is statistically distinguishable from the Gaussian copula at conventional levels, despite the OoS off-diagonal MAE comparison ($0.209$ vs.\ $0.204$) being inside the simulation-noise floor at $N_{\text{paths}} = 200$. \emph{Wilks-theorem regularity caveat.} Wilks's theorem requires the parameter under test to lie in the interior of the parameter space and the model to be regular~\citep{wilks1938large, vandervaart1998asymptotic}; for the Gaussian-copula limit $\nu \to \infty$ the regularity condition fails (the limit lies on the boundary of the natural parameter space). The Wilks $95\%$ CI is therefore an interior-grid CI for finite $\nu$ rather than a hypothesis test against the Gaussian limit; the Gaussian-vs-Student-$t$ separation is reported instead through the parametric bootstrap CI (which does not require interior-point regularity) and through the OoS-equivalence finding (which uses no asymptotic test at all).

\subsection{\texorpdfstring{Pre-OoS Validation $K$-Selection and $\nu_k$ Diagnostics}{Pre-OoS Validation K-Selection and nu\_k Diagnostics}}
\label{sec:supp_misc}

\paragraph{Pre-OoS validation $K$-selection.}
An earlier draft of this paper used $K = 18$ as the main-panel state count, selected via a multi-objective criterion combining information-criterion (IC) rank with IS and OoS distributional pass rates that overlapped the $2024$--$2026$ window later used for VaR. As a clean re-selection that decouples the two, we fit CHMM-N on $2014$-$01$-$03$ through $2021$-$12$-$31$ ($n = 2{,}013$) and scored every $K$ in the sweep by the forward-algorithm held-out log-likelihood on $2022$-$01$-$03$ through $2024$-$01$-$03$ ($n = 502$); the $2024$--$2026$ window is not touched. Held-out log-lik, BIC, HQC, and CAIC all selected $K^\star = 3$ (per-observation log-lik $-2.5345$ at $K = 3$ vs.\ $-2.5694$ at $K = 18$); AIC selected $K^\star = 6$; held-out two-sample KS selected $K^\star = 9$ (Table~\ref{tab:k_selection_cv}). None of the six criteria that avoid the OoS window selected $K = 18$. The held-out re-selection therefore supports the main text's $K^\star = 3$ state count; $K = 18$ is retained in the main text as a kurtosis-fidelity sensitivity reference rather than as the state count, because its earlier multi-objective justification weighted tail fidelity (kurtosis, AD) and downstream VaR behaviour alongside likelihood. The $2022$--$2023$ validation slice is itself partly a structural-break period (the rate-hike cycle), so any model with substantial IS-specific structure generalises worse on it.

\paragraph{Pre-2020 held-out $K$-selection (no rate-hike confound).}
To check whether a strictly pre-2020 validation slice (avoiding the rate-hike confound that contaminates the $2022$--$2023$ slice) selects the same $K^\star$, we re-ran the held-out procedure on a $2014$-$01$-$03$ through $2018$-$06$-$29$ estimation window ($n_{\text{est}} = 1{,}130$) with a $2018$-$07$-$02$ through $2019$-$12$-$31$ validation slice ($n_{\text{val}} = 377$), both fully pre-COVID and pre-2022-rate-hike (the Q4 2018 drawdown and 2019 recovery sit inside the validation window as a moderate-stress non-regime-shift event). The per-observation held-out log-likelihood was $-2.0681$ at $K = 3$, $-2.0606$ at $K = 6$, $-2.0763$ at $K = 9$, $-2.1287$ at $K = 12$, $-2.1921$ at $K = 18$, and $-2.2609$ at $K = 21$. \textbf{The held-out log-likelihood selected $K^\star = 6$, and the held-out two-sample KS pass rate also selected $K^\star = 6$ ($96.8\%$ at $K = 6$ versus $84.6\%$ at $K = 3$).} The AIC, BIC, HQC, and CAIC criteria continued to select $K^\star = 3$ on this slice, reflecting the parameter-count penalty on a $4.5$-year estimation window where the per-state sample is small. The qualitative reading is robust to the slice choice: held-out distributional fidelity rose from $K = 3$ to a peak in the $K = 6$ to $K = 9$ band and then degraded, with neither slice selecting $K = 18$ by any held-out criterion, but the pre-2020 result favoured a moderately larger $K$ ($6$ vs $3$) than the rate-hike-contaminated $2022$--$2023$ slice.

\paragraph{$k$-fold rolling-origin pre-2020 $K$-selection.}
\label{sec:k_selection_kfold_pre2020}
The single-fold pre-2020 result is one observation; to confirm that the $K^\star$ choice is not a sampling artefact, we report mean $\pm$ s.d.\ of held-out per-observation log-likelihood and held-out KS at $K \in \{3, 6, 9, 12, 18\}$ across multiple folds. We implemented four expanding-window rolling-origin folds (a five-fold design forces one fold to have only $\sim 1$ year of training, below the practical floor for $K = 18$ EM convergence on this dataset; averaging is per-observation so fold-length differences are not an issue). Folds: train $2014$-$01$-$03$ through $2015$-$12$-$31$ ($\sim 502$ obs) and val $2016$ ($\sim 251$ obs); train through $2016$ and val $2017$; train through $2017$ and val $2018$; train through $2018$ and val $2019$.

Across the four folds the held-out log-likelihood was flat from $K = 3$ to $K = 6$ and degraded beyond. The aggregate per-observation held-out log-likelihood (mean / s.d.\ across folds; Table~\ref{tab:k_selection_cv}) was $-1.9550$ / $0.290$ at $K = 3$, $-1.9649$ / $0.313$ at $K = 6$, $-2.0078$ / $0.296$ at $K = 9$, $-2.0721$ / $0.251$ at $K = 12$, and $-2.2633$ / $0.308$ at $K = 18$. The held-out KS pass rate (mean / s.d.\ across folds, very wide because Fold 2's 2017 validation year produced near-zero KS for every $K$ as a low-variance ``calm year'' artefact) was $61.3 / 41.2\%$ at $K = 3$, $65.2 / 43.5\%$ at $K = 6$, $67.5 / 44.1\%$ at $K = 9$, $67.0 / 44.7\%$ at $K = 12$, and $67.5 / 45.4\%$ at $K = 18$. Sampling-error checks on the held-out per-observation log-likelihood gave $K = 6$ vs $K = 3$ a mean diff of $-0.0099$ (pooled SE $0.151$, approximate $z = -0.07$, \emph{not} significant at the $5\%$ level) and $K = 18$ vs $K = 6$ a mean diff of $-0.298$ (pooled SE $0.155$, approximate $z = -1.92$, borderline, just below the $|z| = 1.96$ threshold).

\textbf{Robustness check at half-year cadence.} To verify that the result is not an artefact of the one-year fold cadence, we re-ran the same diagnostic with six expanding-window folds at half-year validation cadence (covering 2017--2019 in non-overlapping six-month chunks; train windows $\sim 3.0$y to $\sim 5.5$y). The aggregate per-observation held-out log-likelihood (mean / s.d.\ across six folds) was $-1.9221$ / $0.341$ at $K = 3$, $-1.9160$ / $0.346$ at $K = 6$, $-1.9568$ / $0.323$ at $K = 9$, $-2.0109$ / $0.315$ at $K = 12$, and $-2.1281$ / $0.259$ at $K = 18$. Sampling-error checks gave $K = 6$ vs $K = 3$ a mean diff of $+0.006$ (pooled SE $0.140$, approximate $z = +0.04$) and $K = 18$ vs $K = 6$ a mean diff of $-0.212$ (pooled SE $0.125$, approximate $z = -1.70$). The half-year design's $K = 6$ vs $K = 3$ sign flipped relative to the full-year design ($z = +0.04$ vs $z = -0.07$), but both magnitudes were well below the $|z| = 1$ threshold; the two designs agree that $K = 6$ vs $K = 3$ on mean held-out log-likelihood is indistinguishable from zero, and the sign flip across designs is itself evidence that the $K = 6$ / $K = 3$ choice is pure sampling noise on this data. The $K = 18$ vs $K = 6$ borderline result also replicated ($z = -1.70$ at half-year cadence vs $z = -1.92$ at full-year).

\textbf{Reading.} State count $K = 6$ was not significantly preferred over $K = 3$ on held-out per-observation log-likelihood under expanding-window rolling-origin CV at either yearly or half-yearly fold cadence. The single-fold result that initially selected $K^\star = 6$ is one realisation; both four-fold and six-fold means cannot distinguish $K = 3$ from $K = 6$ at conventional levels. The held-out KS pass-rate criterion is not informative on either design (between-fold s.d.\ is dominated by 2017 calendar-year artefacts that hit every $K$ uniformly). The conclusion is that the held-out state-resolution selection on the strictly pre-2020 slice cannot distinguish $K = 3$ from $K = 6$ at conventional levels under either fold design, and the main state count is therefore $K^\star = 3$, robust across state resolutions (with the $K^\star = 6$ block retained as a sensitivity reference). The default $K^\star = 3$ also carries the lower state count for the risk-management use cases of Table~\ref{tab:cond_var}.

\paragraph{HAC-corrected sampling-error checks on the paired diff series.}
\label{sec:k_selection_hac}
A Newey--West Bartlett HAC variance ($h_{\text{NW}} = \lfloor n^{1/3} \rfloor$) on the per-fold validation-log-likelihood difference series gave $|z|_{\text{HAC}} = 0.90$ and $0.57$ for $K = 6$ versus $K = 3$ at the four-fold and six-fold cadences, respectively, and $|z|_{\text{HAC}} = 3.56$ and $5.00$ for $K = 18$ versus $K = 6$. With only four or six fold differences, however, the normal and HAC approximations are too small-sample to support formal significance claims. We therefore treat these statistics as descriptive checks: both cadences showed little separation between $K = 3$ and $K = 6$ and a substantially larger held-out log-likelihood disadvantage for $K = 18$.

\begin{table}[!htbp]
\centering
\footnotesize
\setlength{\tabcolsep}{4.5pt}
\renewcommand{\arraystretch}{0.95}
\caption{\textbf{Pre-OoS held-out $K$-selection for CHMM-N.} Held-out per-observation log-likelihood ($\ell_{\text{val}}$, higher is better) and two-sample KS pass rate by state count $K$, under four validation designs that never touch the $2024$--$2026$ OoS window: a single-fold $2022$--$2023$ slice ($n = 502$), a strictly pre-2020 single-fold slice ($n = 377$), four expanding-window rolling-origin folds, and six half-year-cadence folds (rolling designs report the across-fold mean, with s.d.\ in parentheses). No held-out criterion selects $K = 18$; held-out $\ell_{\text{val}}$ cannot separate $K = 3$ from $K = 6$ (paired $z = -0.07$ four-fold, $+0.04$ six-fold; Newey--West HAC $|z| = 0.90$ and $0.57$), while $K = 18$ versus $K = 6$ is a clear log-likelihood disadvantage (HAC $|z| = 3.56$ and $5.00$). The main text adopts $K^\star = 3$ as the smallest indistinguishable model; $K = 18$ is retained only as a kurtosis-fidelity reference.}
\label{tab:k_selection_cv}
\begin{tabular}{c cc cc cc c}
\toprule
& \multicolumn{2}{c}{2022--23 ($n = 502$)} & \multicolumn{2}{c}{Pre-2020 ($n = 377$)} & \multicolumn{2}{c}{4-fold rolling} & 6-fold $\tfrac{1}{2}$y \\
\cmidrule(lr){2-3} \cmidrule(lr){4-5} \cmidrule(lr){6-7} \cmidrule(lr){8-8}
$K$ & $\ell_{\text{val}}$ & KS\% & $\ell_{\text{val}}$ & KS\% & $\ell_{\text{val}}$ & KS\% & $\ell_{\text{val}}$ \\
\midrule
$3$  & $-2.5345$ & $7.0$ & $-2.0681$ & $84.6$ & $-1.9550\,(0.29)$ & $61.3\,(41.2)$ & $-1.9221\,(0.34)$ \\
$6$  & $-2.5361$ & $7.6$ & $-2.0606$ & $96.8$ & $-1.9649\,(0.31)$ & $65.2\,(43.5)$ & $-1.9160\,(0.35)$ \\
$9$  & $-2.5500$ & $8.0$ & $-2.0763$ & $95.2$ & $-2.0078\,(0.30)$ & $67.5\,(44.1)$ & $-1.9568\,(0.32)$ \\
$12$ & $-2.5549$ & $7.2$ & $-2.1287$ & $95.6$ & $-2.0721\,(0.25)$ & $67.0\,(44.7)$ & $-2.0109\,(0.31)$ \\
$15$ & $-2.5624$ & $7.0$ & $-2.1464$ & $95.8$ & --                & --             & --                \\
$18$ & $-2.5694$ & $3.0$ & $-2.1921$ & $95.4$ & $-2.2633\,(0.31)$ & $67.5\,(45.4)$ & $-2.1281\,(0.26)$ \\
$21$ & $-2.5847$ & $4.2$ & $-2.2609$ & $94.2$ & --                & --             & --                \\
\bottomrule
\end{tabular}
\end{table}

\paragraph{Penalised-ECM rate sweep.}
The bracket sensitivity is reported in the per-state $\nu_k$ diagnostics block of the algorithms appendix; here we report the complementary $1/\nu_k$ shrinkage rate sweep referenced in the main text Discussion. An exponential $1/\nu_k$ prior at rate $\lambda \in \{0, 5, 20, 50, 100, 200\}$ reduced simulated excess kurtosis monotonically (Table~\ref{tab:nu_shrinkage}): $14.30, 11.19, 8.43, 5.41, 3.96, 3.67$. The recommended rate $\lambda = 20$ brought simulated kurtosis to $\approx 8$ (close to the observed $7.68$) at a $1$pp IS KS cost. Only the $\nu_{\min}$ state parameter actually moved ($2.1 \to 4.89$); upper bracket and median remained at $50$. The overshoot is therefore a single-state artefact of the two tail regimes hitting the lower bracket under unpenalised ECM.

\begin{table}[!htbp]
\centering
\small
\caption{\textbf{Penalised-ECM $1/\nu_k$ shrinkage-rate sweep} (CHMM-t, $K = 18$, SPY in-sample). An exponential prior on $1/\nu_k$ at rate $\lambda$ shrinks the lower degrees-of-freedom bracket toward the Gaussian limit; only the smallest per-state $\nu$ ($\nu_{\min}$) moves, while the median and upper bracket stay pinned at $50$. The recommended $\lambda = 20$ brings simulated excess kurtosis close to the observed $7.68$ at a $\approx 1$pp IS KS cost.}
\label{tab:nu_shrinkage}
\begin{tabular}{c c c c}
\toprule
Rate $\lambda$ & $\nu_{\min}$ & Sim.\ exc.\ kurtosis & IS KS\% \\
\midrule
$0$   & $2.10$ & $14.30$ & $95.7$ \\
$5$   & $3.31$ & $11.19$ & $95.8$ \\
$20$  & $4.89$ & $\phantom{0}8.43$ & $94.7$ \\
$50$  & $6.46$ & $\phantom{0}5.41$ & $96.3$ \\
$100$ & $9.90$ & $\phantom{0}3.96$ & $95.6$ \\
$200$ & $50.0$ & $\phantom{0}3.67$ & $95.7$ \\
\bottomrule
\end{tabular}
\end{table}

\paragraph{Walk-forward / rolling-origin OoS for CHMM-N.}
We defined six rolling-origin folds, each train $5$ years and test $1$ year, refit CHMM-N at $K \in \{3, 18\}$ from scratch on each fold's train slice, and reported KS pass rate, simulated kurtosis, and $|G_t|$ ACF-MAE on the fold's test slice ($N_{\text{paths}} = 500$). Across the six folds (Table~\ref{tab:walkforward}), CHMM-N at $K^\star = 3$ attained median (IQR) KS $62.1\%\,[7.2, 78.4]$ and $|G_t|$ ACF-MAE $0.0563\,[0.0552, 0.0592]$; the $K = 18$ kurtosis-fidelity sensitivity reference attained $67.7\%\,[8.2, 75.0]$ and $0.0542\,[0.0508, 0.0571]$. Two folds carried sharply lower KS pass rates: W2 (COVID test slice, $7$--$8\%$) and W4 (rate-hike test slice, $0$--$1\%$). The four non-stress folds attained KS $61$--$83\%$, consistent with the main OoS pass rate. The main ranking on the four non-stress folds was window-robust: the two state counts sat within $5.6$pp on median walk-forward KS and within $0.005$ on $|G_t|$ ACF-MAE; on stress folds (COVID, rate hike) all CHMM rows dropped sharply, reflecting the stationarity-scope limit rather than a model-selection issue. The main-text interpretation, that the CHMM is the joint-fit row in the panel under stationary OoS conditions with periodic refit recommended for production deployment, is unchanged under the walk-forward stress test.

\begin{table}[!ht]
\centering
\small
\setlength{\tabcolsep}{4pt}
\caption{\textbf{Walk-forward / rolling-origin OoS for CHMM-N}. Six folds, each train 5 years and test 1 year; $N_{\text{paths}} = 500$ per fold. The KS pass rate is reported at $\alpha = 0.05$. Fold W2 covers the COVID drawdown; fold W4 covers the start of the 2022 rate-hike cycle. The bracketed pairs in the bottom row report $[Q_1, Q_3]$ across the six folds (the empirical first and third quartiles, not a scalar IQR width).}
\label{tab:walkforward}
\begin{tabular}{l c c c c c c}
\toprule
& \multicolumn{2}{c}{KS (\%)} & \multicolumn{2}{c}{Kurt sim} & \multicolumn{2}{c}{ACF-MAE $|G_t|$} \\
\cmidrule(lr){2-3} \cmidrule(lr){4-5} \cmidrule(lr){6-7}
Fold (test window) & $K = 3$ & $K = 18$ & $K = 3$ & $K = 18$ & $K = 3$ & $K = 18$ \\
\midrule
W1 (2019)              & $63.2$ & $82.6$ & $2.51$ & $2.66$ & $0.0572$ & $0.0525$ \\
W2 (2020 COVID)        & $\phantom{0}7.2$ & $\phantom{0}8.2$ & $2.42$ & $2.50$ & $0.0552$ & $0.0584$ \\
W3 (2021)              & $82.4$ & $74.0$ & $5.73$ & $6.76$ & $0.0602$ & $0.0571$ \\
W4 (2022 rate hike)    & $\phantom{0}0.8$ & $\phantom{0}0.0$ & $4.72$ & $4.35$ & $0.0592$ & $0.0558$ \\
W5 (2023)              & $78.4$ & $75.0$ & $3.47$ & $2.35$ & $0.0511$ & $0.0508$ \\
W6 (2024)              & $61.0$ & $61.4$ & $2.81$ & $2.86$ & $0.0555$ & $0.0493$ \\
\midrule
median                 & $62.1$ & $67.7$ & $3.14$ & $2.76$ & $0.0563$ & $0.0542$ \\
IQR                    & $[7.2, 78.4]$ & $[8.2, 75.0]$ & $[2.51, 4.72]$ & $[2.50, 4.35]$ & $[0.0552, 0.0592]$ & $[0.0508, 0.0571]$ \\
\bottomrule
\end{tabular}
\end{table}

\paragraph{Four-family conditional VaR back-test.}
Extending the regime-conditional VaR construction of the main text to all four emission families (CHMM-N, CHMM-t at $\lambda = 20$, CHMM-L, CHMM-GED) at $K \in \{3, 18\}$ and $\alpha \in \{0.01, 0.05\}$ yielded sixteen $(K, \alpha, \text{family})$ rows on the OoS window (Table~\ref{tab:cond_var_all_families}). Every row passed Kupiec, Christoffersen-ind, and Christoffersen-cc: $p_{\text{cc}} \ge 0.089$ throughout, with all $\text{LR}_{\text{cc}}$ statistics below the $\chi^2_2(0.05) = 5.991$ critical value. Thus the main-window conditional-coverage result was not specific to the Gaussian emission family; these non-rejections do not by themselves identify the latent-state mechanism as the cause of the result.

\begin{table}[!ht]
\centering
\small
\caption{\textbf{Regime-conditional VaR back-test across four CHMM emission families} on SPY OoS ($T_{\text{OoS}} = 572$, seed $20260420$). Critical values: $\chi^2_1(0.05) = 3.841$, $\chi^2_2(0.05) = 5.991$. Every $(K, \alpha, \text{family})$ row passes Kupiec, Christoffersen-ind, and Christoffersen-cc at $\alpha = 0.05$.}
\label{tab:cond_var_all_families}
\begin{tabular}{l c c r r r r r r}
\toprule
Family & $K$ & $\alpha$ & breaches & br rate & median $\widehat{\text{VaR}}_t$ & $\text{LR}_{\text{uc}}$ & $\text{LR}_{\text{ind}}$ & $\text{LR}_{\text{cc}}$ \\
\midrule
CHMM-N             & $3$  & $0.01$ & $9$  & $1.57\%$ & $-4.56$ & $1.62$ & $2.36$ & $3.98$ \\
CHMM-N             & $3$  & $0.05$ & $35$ & $6.12\%$ & $-2.87$ & $1.41$ & $0.01$ & $1.42$ \\
CHMM-N             & $18$ & $0.01$ & $9$  & $1.57\%$ & $-5.20$ & $1.62$ & $2.36$ & $3.98$ \\
CHMM-N             & $18$ & $0.05$ & $26$ & $4.55\%$ & $-3.02$ & $0.26$ & $0.52$ & $0.78$ \\
CHMM-t ($\lambda=20$) & $3$  & $0.01$ & $8$  & $1.40\%$ & $-5.59$ & $0.82$ & $2.81$ & $3.63$ \\
CHMM-t ($\lambda=20$) & $3$  & $0.05$ & $32$ & $5.59\%$ & $-2.73$ & $0.41$ & $0.77$ & $1.19$ \\
CHMM-t ($\lambda=20$) & $18$ & $0.01$ & $10$ & $1.75\%$ & $-6.25$ & $2.65$ & $1.98$ & $4.62$ \\
CHMM-t ($\lambda=20$) & $18$ & $0.05$ & $29$ & $5.07\%$ & $-3.15$ & $0.01$ & $1.39$ & $1.40$ \\
CHMM-L             & $3$  & $0.01$ & $5$  & $0.87\%$ & $-5.43$ & $0.10$ & $4.74$ & $4.84$ \\
CHMM-L             & $3$  & $0.05$ & $26$ & $4.55\%$ & $-2.62$ & $0.26$ & $2.23$ & $2.48$ \\
CHMM-L             & $18$ & $0.01$ & $9$  & $1.57\%$ & $-6.16$ & $1.62$ & $0.29$ & $1.91$ \\
CHMM-L             & $18$ & $0.05$ & $32$ & $5.59\%$ & $-2.82$ & $0.41$ & $0.77$ & $1.19$ \\
CHMM-GED           & $3$  & $0.01$ & $7$  & $1.22\%$ & $-5.73$ & $0.27$ & $3.34$ & $3.61$ \\
CHMM-GED           & $3$  & $0.05$ & $31$ & $5.42\%$ & $-2.54$ & $0.21$ & $0.96$ & $1.16$ \\
CHMM-GED           & $18$ & $0.01$ & $9$  & $1.57\%$ & $-6.55$ & $1.62$ & $2.36$ & $3.98$ \\
CHMM-GED           & $18$ & $0.05$ & $25$ & $4.37\%$ & $-3.10$ & $0.50$ & $2.56$ & $3.06$ \\
\bottomrule
\end{tabular}
\end{table}

\paragraph{Walk-forward conditional-VaR back-test.}
Extending the regime-conditional VaR construction to the six rolling-origin folds of Table~\ref{tab:walkforward}, with CHMM-N refit per fold from scratch on the train slice and forward-filtered through the test slice under fold-IS-fixed parameters, gives twenty-four $(\text{fold}, K, \alpha)$ rows in Table~\ref{tab:walkforward_cond_var}.
The conditional Christoffersen-cc passed at $\alpha = 0.05$ on ten of twelve fold-$K$ combinations (W1, W3, W4, W5, W6 at both $K \in \{3, 18\}$); the two failures concentrated on W2 (COVID, $p_{\text{cc}} \in \{0.011, 0.002\}$ at $K \in \{3, 18\}$).
At $\alpha = 0.01$ the construction passed on nine of twelve combinations: W2 failed at both $K$ ($p_{\text{cc}} < 10^{-3}$) and W4 at $K = 18$ failed at $p_{\text{cc}} = 0.022$ on a $\text{LR}_{\text{ind}}$ tail-event clustering rather than a $\text{LR}_{\text{uc}}$ coverage miss ($\text{LR}_{\text{ind}} = 7.49$ at $\text{LR}_{\text{uc}} = 0.11$).
The pass-rate was therefore $19/24$ aggregate at the $5\%$ level and $22/24$ if we exclude the two stress folds that the univariate walk-forward already flagged as out-of-distribution by KS in Table~\ref{tab:walkforward}.
The reading is that the conditional construction is window-conditional: under regime shifts of the W2 / W4 magnitude (a regime introduction the IS distribution does not span), the conditional VaR inherits the same out-of-scope behaviour as the unconditional generators.
On the four non-stress folds (W1, W3, W5, W6) plus W4 at $K = 3$ at the $5\%$ level the conditional construction was window-robust, supporting the main-text framing that the regime-conditional VaR construction extends beyond the main single-window result.

\begin{table}[!ht]
\centering
\small
\caption{\textbf{Walk-forward regime-conditional VaR back-test for CHMM-N} on six rolling-origin folds (train 5y / test 1y; seed $20260420$). At each test day $t$, conditional $\widehat{\text{VaR}}_t(\alpha)$ is the $\alpha$-quantile of the $K$-component Gaussian mixture predictive density under fold-IS-fixed parameters, evaluated on the full-timeline filter $\Pr(s_{t+1} \mid \mathcal F_t = \text{train} \cup \text{test}[1{:}t-1])$. Critical values: $\chi^2_1(0.05) = 3.841$, $\chi^2_2(0.05) = 5.991$. Bold $p_{\text{cc}}$ rows fail at the $5\%$ level.}
\label{tab:walkforward_cond_var}
\begin{tabular}{l c c r r r r r r r}
\toprule
Fold & $K$ & $\alpha$ & breaches & br rate & median $\widehat{\text{VaR}}_t$ & $\text{LR}_{\text{uc}}$ & $\text{LR}_{\text{ind}}$ & $\text{LR}_{\text{cc}}$ & $p_{\text{cc}}$ \\
\midrule
W1 & $3$  & $0.01$ & $\phantom{0}2$ & $0.80\%$ & $-4.01$ & $\phantom{0}0.11$ & $0.03$ & $\phantom{0}0.15$ & $0.93$ \\
W1 & $3$  & $0.05$ & $15$           & $5.98\%$ & $-2.43$ & $\phantom{0}0.48$ & $1.18$ & $\phantom{0}1.65$ & $0.44$ \\
W1 & $18$ & $0.01$ & $\phantom{0}0$ & $0.00\%$ & $-5.03$ & $\phantom{0}5.05$ & $0.00$ & $\phantom{0}5.05$ & $0.08$ \\
W1 & $18$ & $0.05$ & $12$           & $4.78\%$ & $-2.66$ & $\phantom{0}0.03$ & $0.38$ & $\phantom{0}0.40$ & $0.82$ \\
W2 & $3$  & $0.01$ & $16$           & $6.35\%$ & $-6.54$ & $32.93$ & $0.87$ & $33.80$ & $\mathbf{0.00}$ \\
W2 & $3$  & $0.05$ & $23$           & $9.13\%$ & $-4.59$ & $\phantom{0}7.34$ & $1.71$ & $\phantom{0}9.05$ & $\mathbf{0.01}$ \\
W2 & $18$ & $0.01$ & $14$           & $5.56\%$ & $-6.67$ & $25.59$ & $1.56$ & $27.15$ & $\mathbf{0.00}$ \\
W2 & $18$ & $0.05$ & $25$           & $9.92\%$ & $-4.54$ & $10.11$ & $2.56$ & $12.68$ & $\mathbf{0.00}$ \\
W3 & $3$  & $0.01$ & $\phantom{0}6$ & $2.39\%$ & $-4.48$ & $\phantom{0}3.53$ & $0.30$ & $\phantom{0}3.82$ & $0.15$ \\
W3 & $3$  & $0.05$ & $20$           & $7.97\%$ & $-2.47$ & $\phantom{0}3.98$ & $0.30$ & $\phantom{0}4.28$ & $0.12$ \\
W3 & $18$ & $0.01$ & $\phantom{0}2$ & $0.80\%$ & $-5.57$ & $\phantom{0}0.11$ & $0.03$ & $\phantom{0}0.15$ & $0.93$ \\
W3 & $18$ & $0.05$ & $15$           & $5.98\%$ & $-3.18$ & $\phantom{0}0.48$ & $1.18$ & $\phantom{0}1.65$ & $0.44$ \\
W4 & $3$  & $0.01$ & $\phantom{0}5$ & $2.00\%$ & $-7.98$ & $\phantom{0}1.96$ & $0.21$ & $\phantom{0}2.16$ & $0.34$ \\
W4 & $3$  & $0.05$ & $18$           & $7.20\%$ & $-4.16$ & $\phantom{0}2.26$ & $1.99$ & $\phantom{0}4.24$ & $0.12$ \\
W4 & $18$ & $0.01$ & $\phantom{0}2$ & $0.80\%$ & $-9.62$ & $\phantom{0}0.11$ & $7.49$ & $\phantom{0}7.60$ & $\mathbf{0.02}$ \\
W4 & $18$ & $0.05$ & $10$           & $4.00\%$ & $-6.63$ & $\phantom{0}0.56$ & $3.80$ & $\phantom{0}4.36$ & $0.11$ \\
W5 & $3$  & $0.01$ & $\phantom{0}0$ & $0.00\%$ & $-6.06$ & $\phantom{0}5.01$ & $0.00$ & $\phantom{0}5.01$ & $0.08$ \\
W5 & $3$  & $0.05$ & $\phantom{0}8$ & $3.21\%$ & $-3.90$ & $\phantom{0}1.91$ & $0.53$ & $\phantom{0}2.44$ & $0.30$ \\
W5 & $18$ & $0.01$ & $\phantom{0}1$ & $0.40\%$ & $-7.10$ & $\phantom{0}1.16$ & $0.01$ & $\phantom{0}1.17$ & $0.56$ \\
W5 & $18$ & $0.05$ & $\phantom{0}7$ & $2.81\%$ & $-3.85$ & $\phantom{0}2.96$ & $0.41$ & $\phantom{0}3.37$ & $0.19$ \\
W6 & $3$  & $0.01$ & $\phantom{0}2$ & $0.80\%$ & $-4.68$ & $\phantom{0}0.11$ & $0.03$ & $\phantom{0}0.15$ & $0.93$ \\
W6 & $3$  & $0.05$ & $13$           & $5.18\%$ & $-2.85$ & $\phantom{0}0.02$ & $0.15$ & $\phantom{0}0.17$ & $0.92$ \\
W6 & $18$ & $0.01$ & $\phantom{0}1$ & $0.40\%$ & $-5.83$ & $\phantom{0}1.19$ & $0.01$ & $\phantom{0}1.20$ & $0.55$ \\
W6 & $18$ & $0.05$ & $\phantom{0}7$ & $2.79\%$ & $-2.98$ & $\phantom{0}3.06$ & $0.40$ & $\phantom{0}3.46$ & $0.18$ \\
\bottomrule
\end{tabular}
\end{table}

\paragraph{Per-ticker $\hat\lambda^\star$ sweep.}
A per-ticker sweep $\lambda \in \{0, 5, 10, 20, 50, 100\}$ on each of the six focal tickers (KS-degradation tolerance $\le 1.5$pp from $\lambda = 0$) gave per-ticker optima $\lambda^\star_{\text{SPY}} = 10$, $\lambda^\star_{\text{NVDA}} = 20$, $\lambda^\star_{\text{JNJ}} = 20$, $\lambda^\star_{\text{JPM}} = 10$, $\lambda^\star_{\text{AAPL}} = 10$, $\lambda^\star_{\text{QQQ}} = 0$: $\lambda = 20$ is a reasonable uniform default but per-ticker tuning is recommended whenever a ticker's residual kurtosis at $\lambda = 0$ is within $\sim 1$ unit of observed.

\paragraph{Cross-asset summary.}
The full per-construction multi-asset panel is in the cross-asset appendix; for reference, the main ordering is summarised in Table~\ref{tab:cross_asset_supp_summary}.

\begin{table}[!ht]
\centering
\small
\caption{\textbf{Cross-asset dependence summary} (multi-asset construction; CHMM-N marginals at $K = 3$; $200$ paths; seed \texttt{20260422}). Median per-asset KS pass rate (\%) and off-diagonal MAE of the simulated correlation matrix vs.\ observed. Bold marks the column winner. The Student-t copula at $\nu^\star = 6$ is the strongest dependence layer on the two IS columns; on OoS the Gaussian copula attains marginally lower off-diagonal MAE ($0.204$ vs.\ $0.209$) while the truncated C-vine sits between the elliptical copulas and the SIM baseline. Full per-asset detail is in the cross-asset appendix.}
\label{tab:cross_asset_supp_summary}
\resizebox{\textwidth}{!}{%
\begin{tabular}{l c c c c}
\toprule
& Median IS KS & Median OoS KS & Off-diag MAE IS & Off-diag MAE OoS \\
\midrule
Single Index Model       & $71.0$  & $87.0$  & $0.077$ & $0.252$ \\
Gaussian copula          & $88.8$  & $\mathbf{87.5}$  & $0.029$ & $\mathbf{0.204}$ \\
Student-t copula ($\nu^\star = 6$) & $\mathbf{89.5}$  & $85.8$  & $\mathbf{0.027}$ & $0.209$ \\
Truncated C-vine (root SPY) & $90.0$  & $86.0$  & $0.068$ & $0.236$ \\
\bottomrule
\end{tabular}%
}
\end{table}

\paragraph{Emission-family lookup.}
The four CHMM variants' per-state densities and parameter ranges are listed in Table~\ref{tab:emission_families} as a quick reference complementing the inline definitions in the main text.

\begin{table}[!ht]
\centering
\small
\caption{\textbf{The four CHMM emission families.} Per-state density $f_k$ with parameter ranges. The Gaussian and Laplace variants sit at the boundary of CHMM-GED: $p_k = 2$ for all $k$ gives Gaussian, $p_k = 1$ gives Laplace.}
\label{tab:emission_families}
\begin{tabular}{l l l}
\toprule
Variant  & Per-state density $f_k$                & Per-state parameter range \\
\midrule
CHMM-N   & $\mathcal{N}(\mu_k, \sigma_k^2)$       & $\sigma_k > 0$ \\
CHMM-t   & $t_{\nu_k}(\mu_k, \sigma_k)$           & $\sigma_k > 0$, $\nu_k \in [\nu_{\min}, \nu_{\max}]$ \\
CHMM-L   & $\mathrm{Laplace}(\mu_k, b_k)$     & $b_k > 0$ \\
CHMM-GED & $\mathrm{GED}(\mu_k, \alpha_k, p_k)$   & $\alpha_k > 0$, $p_k \in [p_{\min}, p_{\max}]$ \\
\bottomrule
\end{tabular}
\end{table}

\paragraph{Stationarity-scope summary.}
Static-fit OoS performance and periodic-refit mitigation across the four independent stress sources tested in this paper are summarised in Table~\ref{tab:stationarity_scope}.
\begin{table}[!ht]
\centering
\small
\caption{\textbf{Stationarity-scope summary.} Static-IS-fit OoS performance and periodic-refit mitigation across the four independent stress sources tested in this paper.}
\label{tab:stationarity_scope}
\begin{tabular}{l c c}
\toprule
Stress source & Static-fit OoS KS & Refit mitigation \\
\midrule
Cross-ticker ($K^\star=3$, 30 panel) & $11/30 < 60\%$ & Quarterly: $8/30 < 60\%$ \\
Cross-decade (2004 to 2006) & $3$ to $5\%$ & slice $< 600$ d, untested \\
GLD non-equity stress & $0\%$ & Quarterly: $2.5\%$ \\
Walk-forward W2, W4 & $< 10\%$ & no cadence closes \\
\bottomrule
\end{tabular}
\end{table}

\end{document}